\documentclass[acmtocl]{acmsmall}

\usepackage{versions}
\includeversion{ARXIV}
\excludeversion{ACM} 

\usepackage[UKenglish]{babel}

\usepackage{tikz}
\usetikzlibrary{arrows,positioning,matrix}

\usepackage{amssymb}
\usepackage{amsmath}
\usepackage{xcolor}
\usepackage{MnSymbol}

\usepackage{verbatim} 

\usepackage[boxed]{algorithm2e}

\usepackage[all]{xy}

\usepackage{nicefrac}

\usepackage{xspace}

\newcommand*{\domain}[0]{\ensuremath{\mathrm{dom}}}

\newcommand*{\Fraisse}[0]{Fra\"iss\'e\xspace}

\newcommand*{\FO}[1]{\ensuremath{\mathrm{FO}_{#1}}}

\newcommand{\N}{\ensuremath{\mathbb{N}}}  
\newcommand*{\height}[0]{\ensuremath{\mathrm{hgt}}}
\newcommand*{\TOP}[1]{\ensuremath{\mathrm{top}_{#1}}}
\newcommand*{\Push}[1]{\ensuremath{\mathrm{push}_{#1}}}
\newcommand*{\Pop}[1]{\ensuremath{\mathrm{pop}_{#1}}}
\newcommand*{\Clone}[1]{\ensuremath{{\mathrm{clone}_{#1}}}}
\newcommand*{\Op}[0]{\mathsf{OP}}
\newcommand*{\op}[0]{\mathsf{op}}

\newcommand*{\trans}[1]{\ensuremath{\mathrel{{\vdash^{#1}}}}}
\newcommand*{\invtrans}[1]{\ensuremath{\mathrel{{\dashv^{#1}}}}}
\newcommand*{\prefixeq}[0]{\ensuremath{\mathop{\trianglelefteq}}}
\newcommand*{\notprefixeq}[0]{\ensuremath{\mathop{\not\trianglelefteq}}}

\newcommand*{\length}[0]{\mathrm{len}}

\newcommand*{\RelAnc}[2]{\ensuremath{{\mathrm{RA}_{#1}({#2})}}}
\newcommand*{\AncestorClass}[4]{\mathfrak{R}_{#1, #2, #3, #4}}
\newcommand*{\RelAncequiv}[4]{\mathrel{{_{#4}^{#1}}{\equiv}_{#2}^{#3}}}

\newcommand*{\HONPT}[0]{\ensuremath{\mathrm{NPT}}\xspace}

\newcommand*{\BoundHeight}[0]{\zeta} 
\newcommand*{\BoundWidth}[0]{\eta}
\newcommand*{\BoundRunLength}[0]{\theta}

\newcommand*{\FuncBoundReturnLength}[2]{\ensuremath{ \mathrm{RL}_{#2}^{#1}}}  
\newcommand*{\FuncBoundLoopLength}[2]{\ensuremath{ \mathrm{LL}_{#2}^{#1}}}  
\newcommand*{\FuncBoundHighLoopLength}[2]{\ensuremath{ \mathrm{HLL}_{#2}^{#1}}}

\newcommand*{\Lin}[3]{\ensuremath\mathfrak{Lin}_{#1}^{#2;#3}}
\newcommand*{\Typ}[3]{\ensuremath\mathrm{Type}_{#1}^{#2;#3}}
\newcommand*{\wordequiv}[2]{\equiv_{#1}^{#2}}
\newcommand*{\stackequiv}[3]{\mathrel{{_{#3}}{\equiv}_{#1}^{#2}}}
\newcommand*{\stackequivTyp}[3]{{\stackequiv{#1}{#2}{#3}}\text{-}\mathrm{Type}}

\newcommand*{\ReturnFunc}[1]{\ensuremath{\mathrm{\#Ret}^{#1}}}
\newcommand*{\LoopFunc}[1]{\ensuremath{\mathrm{\#Loop}^{#1}}}
\newcommand*{\HighLoopFunc}[1]{\ensuremath{\mathrm{\#HLoop}^{#1}}}

\newcommand*{\FuncBoundTopWord}[0]{\alpha}
\newcommand*{\FuncBoundWidthWord}[0]{\beta}
\newcommand*{\ConstBoundHeightWord}[0]{\mathrm{B_{hgt}}}

\newcommand*{\BoundHeightOnestepConstructionSimultanious}[0]{\gamma}

\newcommand*{\Milestones}[0]{\ensuremath{\mathrm{MS}}}
\newcommand*{\genMilestones}[0]{\ensuremath{\mathrm{GMS}}}

\newcommand*{\jumpedge}[0]{\curvearrowright}
\newcommand*{\jumpleftedge}[0]{\curvearrowleft}

\newcommand*{\plusedge}[0]{\lcurvearrowne}
\newcommand*{\plusleftedge}[0]{\rcurvearrownw}

\newcommand*{\titlefootnote}{%
  This paper was mainly written while the author was a student at
  Technische Universit\"at Darmstadt and was funded by the DFG via
  the project 
  ``Strukturkonstruktionen und modelltheoretische Spiele in
  speziellen Strukturklassen''.
}

\title{First-Order Logic on Higher-Order Nested Pushdown Trees}
\author{ALEXANDER KARTZOW
  \processifversion{ARXIV}{\footnote{\titlefootnote}}
  \\ Universit\"at Leipzig}
\begin{abstract}
  We introduce a new hierarchy of \emph{higher-order
   nested pushdown trees} generalising  
    Alur et al.'s concept of nested pushdown trees.
   Nested pushdown trees are useful representations of control flows
   in the verification of programs with recursive calls of
   first-order functions. 
   Higher-order nested pushdown trees are expansions 
   of unfoldings of 
   graphs generated by higher-order pushdown systems. 
   Moreover, the class of nested pushdown trees of level $n$ is
   uniformly first-order interpretable in the class of collapsible
   pushdown graphs of level $n+1$. The relationship between the class
   of higher-order pushdown graphs and the class of collapsible
   higher-order pushdown graphs is not very well understood. 
   We hope that the further study of the nested pushdown tree
   hierarchy leads to a better 
   understanding of these two hierarchies.
   In this paper, we are concerned with the first-order model checking
   problem on 
   higher-order nested pushdown trees.
   We show that the first-order 
   model checking on the first two levels of this hierarchy is
   decidable. Moreover, we obtain an $2$-EXPSPACE algorithm  for the
   class of nested pushdown trees of level $1$. The proof technique
   involves a pseudo-local analysis of strategies in the
   Ehrenfeucht-\Fraisse games on two identical copies of a nested
   pushdown tree. Ordinary 
   locality arguments in the spirit of Gaifman's lemma do not apply
   here because nested pushdown trees tend to have small diameters. 
   We introduce the  notion of relevant ancestors which provide a
   sufficient 
   description of the $\FO{k}$-type of each element in a higher-order
   nested pushdown tree. The local analysis of these ancestors allows
   us to prove the existence of restricted  winning strategies in the
   Ehrenfeucht-\Fraisse game. These strategies are then used to create
   a first-order model checking algorithm.
\end{abstract}
\category{F.1.1.}{Computation by Abstract Devices}{Models of Computation}[Automata]  
\terms{Theory}

\begin{ARXIV}
  \runningfoot{}
  \gdef\firstfoot{}
\end{ARXIV}

\begin{document}

\begin{ACM}
  \begin{bottomstuff}
    \titlefootnote
  \end{bottomstuff}
\end{ACM}

\maketitle

\section{Introduction}
During the last decade, different generalisations of pushdown systems
have gained attention in the field of software verification and model
checking.  
Knapik et al.\
\citeyear{KNU02} showed that Higher-order pushdown systems, first defined
by Maslov \citeyear{Maslov74,Maslov76}, generate the same
class of trees as \emph{safe} higher-order recursion schemes.
Safety is a syntactic condition concerning the order of the 
output compared to the order of the inputs of higher-order recursion schemes.
A higher-order pushdown system is a pushdown system that uses a
nested stack structure instead of an ordinary stack. This means that
a level $2$ pushdown system ($2$-PS) uses a stack of stacks, a level
$3$ pushdown system ($3$-PS) uses a stack of stacks of stacks,
etc.  Hague et al.\ \citeyear{Hague2008} defined the class of
collapsible pushdown systems by adding a new stack operation called
collapse. 
They proved that the trees generated by level $l$ collapsible
pushdown systems coincide with the trees generated by level $l$
recursion schemes. The exact relationship between the higher-order
pushdown hierarchy and the collapsible pushdown hierarchy remains 
an open problem. It is not known whether the trees generated by safe
recursion schemes are a proper subclass of the trees generated by all
recursion schemes.\footnote{Recently,
  P. Parys \citeyear{parys:LIPIcs:2011:3047} proved the uniform
  safety conjecture for level $2$: there is a level $2$ recursion
  scheme that is not generated by any level $2$ safe scheme.} 
Due to the correspondence of recursion schemes and pushdown trees,
this question can be equivalently formulated as follows: is there
some collapsible pushdown system that generates a tree which is not
generated by any higher-order pushdown system?

Also from a model theoretic perspective these hierarchies are
interesting classes. 
The graphs generated by higher-order pushdown systems are
exactly the graphs in the Caucal-hierarchy \cite{cawo03}. Thus, they
are one of the largest known classes with decidable monadic
second-order theories. In contrast,  Broadbent
\citeyear{Broadbent2012} recently showed that the first-order theories
of  graphs generated by
level $3$ collapsible pushdown systems are in general undecidable
(level $2$ collapsible pushdown graphs have undecidable monadic
second-order theories but decidable first-order theories
\cite{Hague2008,Kartzow10}).
Thus, the collapse  operation induces a drastic
change with respect to classical decidability issues. 

Furthermore, collapsible pushdown graphs have decidable modal
$\mu$-calculus theories \cite{Hague2008}. In fact, the class of collapsible
pushdown graphs and the class of nested pushdown trees are the only
known natural classes with decidable modal $\mu$-calculus theories
but undecidable monadic second-order theories. 

Further study of the higher-order pushdown hierarchy and the
collapsible pushdown hierarchy is
necessary for a better understanding of these results. It  may also
reveal an answer to the 
question whether safety implies a semantical restriction for recursion
schemes. 

In this paper we introduce the hierarchy of higher-order nested
pushdown trees. This is a new hierarchy between the hierarchy of
higher-order pushdown trees and that of collapsible pushdown
graphs. We hope that its study reveals more insights into the
structure of these
hierarchies. 

Nested pushdown trees were first introduced by 
Alur et al.\ \citeyear{Alur06languagesof}. These are trees generated by
pushdown systems (of level 1) enriched by a new \emph{jump relation}
that connects each push operation with the corresponding pop
operations. They introduced these trees in order to 
verify  specifications concerning  pre/postconditions on function
calls/returns in 
recursive first-order programs. ``Ordinary'' pushdown trees offer suitable
representations of control flows of recursive first-order functions.
Since these trees have decidable monadic second-order theories
\cite{MullerS85}, one
can use these representations fruitfully for verification
purposes. But monadic
second-order logic does not provide the expressive power necessary for
defining the position before and after the call of a certain
function in such a pushdown tree. Alur et al.'s new jump relation makes
these pairs of positions 
definable by a quantifier-free formula.  Unfortunately, this new
relation turns the monadic second-order theories undecidable. 
But they showed that modal
$\mu$-calculus model checking is still decidable on the class of nested
pushdown trees. Thus, nested pushdown trees form a suitable
representation for control flows of first-order recursive programs
for the verification of  modal $\mu$-calculus definable
properties of the control flows including  pre/postconditions on
function calls/returns.

Of course, the idea of making corresponding push and pop operations
visible is not restricted to pushdown
systems of level $1$. We define a level $n$ nested pushdown tree 
($n$-NPT) to be a tree generated by a level $n$ pushdown
system (without 
collapse!) expanded by a jump relation that connects every push of
level $n$ with the corresponding pop operations (of level $n$). 

This new hierarchy contains by definition expansions of
higher-order pushdown trees. Moreover, we show that the class of 
$n$-NPT is uniformly first-order interpretable in the class of level
$n+1$ collapsible pushdown graphs. 

We then study first-order model checking on the first two levels of
this new hierarchy. In particular, we provide a $2$-EXPTIME alternating
Turing machine  deciding the model
checking problem
for the class of $1$-NPT. We already proved the same complexity bound in
\cite{Kartzow09}. Here, we reprove the statement with a different
technical approach that generalises to the higher levels of the nested
pushdown tree hierarchy.

\paragraph{Outline}
Section \ref{sec:Preliminaries} contains some basic definitions
concerning first-order logic and higher-order pushdown
systems. Moreover, we recall the basics of Ehrenfeucht-\Fraisse games
and explain how the analysis of strategies in these games can be used
to derive first-order model checking algorithms on certain classes of
structures. 
In Section \ref{sec:HONPT}, we then introduce the hierarchy of nested
pushdown 
trees. We relate this hierarchy to the hierarchies of pushdown
trees and of collapsible pushdown graphs. 
From that point on, we only focus on the first-order model checking
problem for the first two levels of the
nested pushdown tree hierarchy. 
In Section \ref{sec:TowardsModelChecking}, we explain the rough
picture 
how the ideas of Section \ref{sec:EFGame} lead to a model checking
algorithm for this class. 
We then give an outline of the Sections
\ref{sec:TheoryRuns}--\ref{sec:FODecidability} which
provide the details of the correctness proof for the model checking
algorithm presented in Section \ref{sec:FODecidability}. 
Finally, we give some concluding remarks and point to open problems
in Section \ref{sec:Conclusion}.

\section{Preliminaries and Basic Definitions}
\label{sec:Preliminaries}
 We denote first-order logic by $\FO{}$. The quantifier
 rank of some formula $\varphi\in\FO{}$
 is the maximal number of nestings of
 existential and universal quantifiers in $\varphi$.
 We denote by $\FO{\rho}$ the set of first-order
 formulas of quantifier rank up to $\rho$ and 
 by $\equiv_\rho$ equivalence of structures (with parameters) with
 respect to all 
 $\FO{\rho}$ formulas. This means that for 
structures $\mathfrak{A}, \mathfrak{B}$ and parameters $\bar a\in
\mathfrak{A}^n$, $\bar b\in \mathfrak{B}^n$,
$\mathfrak{A}, \bar a
\equiv_\rho \mathfrak{B}, \bar b$ if and only if for all
$\varphi\in\FO{\rho}$ (with $n$ free variables), 
\mbox{$\mathfrak{A}\models\varphi(\bar a) \Leftrightarrow
\mathfrak{B}\models\varphi(\bar b)$} holds. 

The $\FO{}$ model checking problem on a class $\mathcal{C}$
asks for an algorithm that  determines whether
$\mathfrak{A}\models\varphi$ on input $(\mathfrak{A}, \varphi)$ 
where $\mathfrak{A}\in\mathcal{C}$ and
$\varphi\in\FO{}$. In Section \ref{sec:EFGame}, we develop
a translation from \emph{dynamic-small-witness strategies} in
Ehrenfeucht-\Fraisse games to 
\FO{} model checking algorithms on nice classes of structures. A
dynamic-small-witness strategy in the Ehrenfeucht-\Fraisse game allows
Duplicator to answer any challenge of Spoiler by choosing some element
with a short representation. 
In Section \ref{sec:HPSDef} we introduce higher-order pushdown
systems. 

\subsection{Ehrenfeucht-\Fraisse Games and First-Order Model Checking}
\label{sec:EFGame}
The equivalence $\equiv_\rho$  has a nice characterisation via
\emph{Ehrenfeucht-\Fraisse  
games}. Based on the work of \Fraisse \citeyear{Fraisse54},
Ehrenfeucht \citeyear{Ehren60} 
introduced these games which have become one of the most important
tools for proving inexpressibility of properties in first-order
logic. 
In this paper, we use a nonstandard application of
Ehrenfeucht-\Fraisse game 
analysis to the $\FO{}$ model checking problem: strategies of
Duplicator that 
only choose elements with small representations can be turned into a
model checking 
algorithm. After briefly recalling the basic definitions, we explain
this approach to model checking in detail. In the main part of this
paper, we will see that this approach yields an $\FO{}$ model checking
algorithm on the class of nested pushdown trees of level $2$. 
\begin{definition}
  Let $\mathfrak{A}_1$ and $\mathfrak{A}_2$ be $\sigma$-structures. 
  For tuples
  \begin{align*}
    &\bar a^1=a^1_1, a^1_2, \dots, a^1_m\in A_1^m\text{ and }
    \bar a^2=a^2_1, a^2_2, \dots, a^2_m\in A_2^m
  \end{align*}
  we write $\bar a^1
  \mapsto \bar a^2$ for the map that maps $a^1_i$ to $a^2_i$ for all
  $1\leq i \leq m$. 
  In the  $n$-round Ehrenfeucht-\Fraisse game on
  $\mathfrak{A}_1, a^1_1, a^1_2, \dots, a^1_m$ and $\mathfrak{A}_2,
  a^2_1, a^2_2, \dots, a^2_m$ for $a^j_i\in A_j$ there are two
  players, Spoiler 
  and Duplicator, which play according to the following rules. 
  The game is played for $n$ rounds. The $i$-th round consists of the
  following steps.
  \begin{enumerate}
  \item Spoiler chooses one of the structures, i.e., he chooses $j\in\{1,2\}$.
  \item Then he chooses one of the elements of his structure, i.e., he
    chooses some $a^j_{m+i}\in A_j$.
  \item Now, Duplicator chooses some $a^{3-j}_{m+i} \in A_{3-j}$. 
  \end{enumerate}
  Having executed $n$ rounds, Spoiler and Duplicator have chosen
  tuples 
  \begin{align*}
  &\bar a^1:= a^1_1, a^1_2, \dots, a^1_{m+n} \in A_1^{m+n}\text{ and }
  \bar a^2:= a^2_1, a^2_2, \dots, a^2_{m+n} \in A_2^{m+n}.    
  \end{align*}
  Duplicator wins 
  the play if $f:\bar a^1 \mapsto \bar a^2$ is a partial isomorphism,
  i.e., if $f$ satisfies 
  \begin{enumerate}
  \item $a^1_i = a^1_j$ if and only if  $a^2_i = a^2_j$ for all $1\leq
    i \leq j \leq m+n$, and
  \item for each $R_i\in \sigma$ of arity $r$ the following holds: for
    $i_1, i_2, \dots, i_r$ numbers between $1$ and $m+n$, 
    \mbox{$\mathfrak{A_1}, \bar a^1 \models R_i x_{i_1} x_{i_2} \dots x_{i_r}$}
    if and only if
    \mbox{$\mathfrak{A_2}, \bar a^2 \models R_i x_{i_1} x_{i_2} \dots x_{i_r}$}.
  \end{enumerate}
\end{definition}

\begin{lemma}[\cite{Fraisse54,Ehren60}]
  Let $\mathfrak{A}_1$, $\mathfrak{A}_2$ be structures and let 
  \mbox{$\bar a^1 \in \mathfrak{A}_1^n$}, \mbox{$\bar a^2\in
    \mathfrak{A}_2^n$} be $n$-tuples. 
  Duplicator has a winning strategy in the $\rho$-round 
  Ehrenfeucht-\Fraisse game on 
  $\mathfrak{A}_1, \bar a^1$ and $\mathfrak{A}_2, \bar a^2$ if and
  only if 
  $\mathfrak{A}_1, \bar a^1 \equiv_\rho\mathfrak{A}_2, \bar a^2$.
\end{lemma}

In this paper we apply the  analysis of Ehrenfeucht-\Fraisse
games  to the \FO{} model checking problem. We use a variant of the
notion of $H$-boundedness of Ferrante and Rackoff \citeyear{Fer79}. The
existence of 
certain restricted strategies in the 
game played on 
two identical copies of a structure yields an 
$\FO{}$ model checking algorithm. 

We  
consider the game played on two 
copies of the same structure, i.e., the game on $\mathfrak{A}, \bar a^1$ and
$\mathfrak{A}, \bar a^2$ with identical choice of the initial
parameter 
$\bar a^1 = \bar a^2 \in \mathfrak{A}$. Of course, Duplicator has a
winning 
strategy in this setting: she can copy each move of Spoiler. 
But we look for winning strategies
with certain constraints. In our application the constraint is that
Duplicator is only allowed to choose elements 
that are represented by short runs of higher-order pushdown
systems, but the idea 
can be formulated more generally. 
\begin{definition}
  Let $\mathcal{C}$ be a class of structures. Assume that
  $S^{\mathfrak{A}}(m) \subseteq \mathfrak{A}^m$ is a subset of the
  $m$-tuples of the structure $\mathfrak{A}$ for each
  $\mathfrak{A}\in\mathcal{C}$ and each $m\in\N$. Set
  $S:=(S^{\mathfrak{A}}(m))_{m\in\N, \mathfrak{A}\in\mathcal{C}}$.
  We
  call $S$ a \emph{constraint for Duplicator's strategy} and we say 
  Duplicator has an  
  \emph{$S$-preserving} winning strategy if she has a strategy for
  each game 
  played on two copies of $\mathfrak{A}$ for some $\mathfrak{A}\in
  \mathcal{C}$ with parameters, i.e., a game on
  $\mathfrak{A}, \bar a^1$ and $\mathfrak{A}, \bar a^2$ for $n$-tuples
  $\bar a^1, \bar a^2$, with the
  following property. 
  Let $\bar b^1 \mapsto \bar b^2$ be a position 
  reached after $m$ rounds where Duplicator used her strategy.
  If $\bar b^2\in S^{\mathfrak{A}}(m+n)$ and Spoiler chooses some
  element in the first 
  copy of $\mathfrak{A}$, then her
  strategy chooses 
  an element $a^2_{m+1}$ such that 
  $\bar n^2,a^2_{m+1}\in S^{\mathfrak{A}}(m+n+1)$.
\end{definition}
\begin{remark}
  We write $S(m)$ for $S^{\mathfrak{A}}(m)$ if
  $\mathfrak{A}$ is clear from the context.
\end{remark}

We now want to turn an $S$-preserving strategy of Duplicator into a
model checking algorithm. The idea is to restrict the search for
witnesses of existential quantifications to the sets defined by $S$. 
In order to obtain a terminating algorithm, $S$ must be finitary in
the sense of the following definition.
\begin{definition}
  Given a  class $\mathcal{C}$ of finitely represented structures, 
  we call a constraint $S$ for Duplicator's strategy 
  \emph{finitary} on $\mathcal{C}$, if
  for each $\mathfrak{A}\in\mathcal{C}$ we can compute a monotone
  function 
  $f_{\mathfrak{A}}$ such that for all $n\in \N$ 
  \begin{itemize}
  \item $S^{\mathfrak{A}}(n)$ is finite,
  \item there is a representation for each $\bar a\in
    S^{\mathfrak{A}}(n)$ in space $f_{\mathfrak{A}}(n)$, and
  \item $\bar a \in S^\mathfrak{A}(n)$ is effectively decidable. 
  \end{itemize}
\end{definition}

Recall the following fact:
if Duplicator uses a winning strategy in the $n$ round game,
 her choice in the $(m+1)$-st round is an
element $a^i_{m+1}$ for appropriate $i\in\{1,2\}$ such that  
\mbox{$\mathfrak{A},\bar a^1, a^1_{m+1} \equiv_{n-m-1} \mathfrak{A}, \bar a^2,
a^2_{m+1}$}. 
If Duplicator has an $S$-preserving winning
strategy, then
for every formula
\mbox{$\varphi(x_1, x_2, \dots, x_{m+1})\in \FO{n-m-1}$} and for all
$\bar 
a\in A^m$ with $\bar a\in S(m)$ the following holds:
\begin{align}
 \label{eqn:RestrictedSearch1} &\text{there is an element }a\in A\text{ such that }\bar a, a\in
  S(m+1)\text{ and } \mathfrak{A}, \bar a, a \models \varphi \\ 
  \label{eqn:RestrictedSearch2}
  \text{iff }&\text{there is an element }a\in A\text{ such that }\mathfrak{A},
  \bar a, a \models  \varphi \\
  \label{eqn:RestrictedSearch3}
  \text{iff }
  &\mathfrak{A}, \bar a \models \exists x_{m+1} \varphi.
\end{align}
All implications except for \eqref{eqn:RestrictedSearch2}
$\Rightarrow$ 
\eqref{eqn:RestrictedSearch1} are trivial. 
This implication follows from the definition of an
$S$-preserving winning strategy applied to a game starting in
position
$\mathfrak{A}, \bar a \equiv_{n-m} \mathfrak{A}, \bar a$. If
Spoiler chooses any $a'\in\mathfrak{A}$ such that
$\mathfrak{A}, \bar a, a\models \varphi$, then by definition of
$S$-preserving winning strategy, Duplicator may respond with some
$a\in \mathfrak{A}$ such that $\bar a, a \in S(m+1)$ such that
$\mathfrak{A}, \bar a,a \equiv_{n-m-1} \mathfrak{A}, \bar a,a'$.

Let us fix some class $\mathcal{C}$ of finitely represented structures
and let $S$ be some finitary constraint on $\mathcal{C}$ such that
Duplicator has an $S$-preserving winning strategy. 
Iterated use of the equivalence of \eqref{eqn:RestrictedSearch1} and
\eqref{eqn:RestrictedSearch3} shows that 
the alternating Turing machine described in Algorithm
\ref{AlgoSPReservingModelCheck} solves the \FO{} model
checking problem on $\mathcal{C}$.%
\footnote{Without loss of generality we
  assume first-order 
  formulas to be generated from atomic and negated atomic formulas
  only by means of 
  disjunction $\lor$, conjunction $\land$, existential quantification
   $\exists$ and universal quantification $\forall$.}
Note that for a universal quantifier $\forall x \varphi$ we actually
apply the equivalence to  the formula $\neg\varphi$. 
The running time of this algorithm on input
$(\mathfrak{A}, \varphi)$ is $f_{\mathfrak{A}}(\lvert \varphi
\rvert)\cdot \lvert \varphi\rvert$ and it uses at most $\lvert \varphi
\rvert$ many alternations where $\lvert \varphi \rvert$ denotes the
length of the formula $\varphi$. 

\begin{algorithm}[h]
  SModelCheck\\
  \SetKwFunction{ModelCheck}{SModelCheck}
  \SetKw{accept}{accept}
  \SetKw{reject}{reject}
  \SetKw{guess}{guess}
  \SetKw{universally}{universally choose}
  \KwIn{a structure $\mathfrak{A}$
    , a formula $\varphi\in\FO{\rho}$, an
    assignment $\bar x \mapsto \bar a$ for tuples $\bar x, \bar a$ of
    arity $m$ such that $\bar a\in S(m)$}
  \If{ $\varphi$ is a (possibly negated) atom } {
    \lIf{ $\mathfrak{A}, \bar a \models \varphi(\bar x)$} 
    {\accept }\lElse{\reject\;}}
  \If{$\varphi = \varphi_1 \lor \varphi_2$}{
    \lIf{\ModelCheck{$\mathfrak{A}, \bar a, \varphi_1$} $=$ accept}{\accept}
    \Else{ \lIf {\ModelCheck{$\mathfrak{A}, \bar a, \varphi_2$}$=$
        accept} {\accept}
      \lElse{\reject\;}
    }
  }
  \If{$\varphi = \varphi_1 \land \varphi_2$}{
    \If{\ModelCheck{$\mathfrak{A}, \bar a, \varphi_1$} $=$ accept}
    {\lIf {\ModelCheck{$\mathfrak{A}, \bar a, \varphi_2$}$=$
        accept} {\accept}
      \lElse{\reject\;}}
    \lElse{  \reject\;}
  }
  \If{$\varphi=\exists x \varphi_1(\bar x,x)$}{
    \guess an $a\in\mathfrak{A}$ with
    $\bar a, a\in S(m+1)$ and
    \ModelCheck{$\mathfrak{A}, \bar a a, \varphi_1$}\;}
  \If{$\varphi=\forall x_i \varphi_1$}{
    \universally  
    an $a\in\mathfrak{A}$ with
    $\bar a, a\in S(m+1)$ and
    \ModelCheck{$\mathfrak{A}, \bar a a, \varphi_1$}\;}
  \label{AlgoSPReservingModelCheck}
  \caption{\FO{}-model checking on $S$-preserving structures}
\end{algorithm}

\subsection{Higher-Order Pushdown Systems}
\label{sec:HPSDef}
In order to define what a level $n$ nested pushdown tree is, we first
have to introduce pushdown systems of level $n$ ($n$-PS). 

An $n$-PS can be seen as a finite
automaton with access to an $n$-fold nested stack structure.  This
generalises the notion of a pushdown system by replacing a single
stack with a structure that is a stack of stacks of stacks  \dots
of stacks. 
A nested stack of level $n$ can be
manipulated by level $l$ push and pop
operations for each level $l\leq n$. 
For $2\leq l \leq n$, the level $l$ push operation $\Clone{l}$ duplicates the
topmost entry of the topmost level $l$ stack. 
The level $1$ push operation $\Push{\sigma}$ writes the symbol $\sigma$
on top of the topmost level $1$ stack. 
For $1\leq l \leq n$, the level $l$ pop operation $\Pop{l}$ removes the topmost
entry of the topmost level $l$ stack.

For some alphabet $\Sigma$, we inductively define the \emph{set of level $n$
stacks over $\Sigma$} ($n$-stacks), denoted by $\Sigma^{+n}$ as follows.
Let $\Sigma^{+1}:=\Sigma^+$ denote the set of all
nonempty finite words over alphabet $\Sigma$. We then define
\mbox{$\Sigma^{+(n+1)}:=(\Sigma^{+n})^+$}.  

Let us fix an $(n+1)$-stack $s\in \Sigma^{+(n+1)}$.
This stack $s$ consists of an ordered list 
\mbox{$s_1, s_2, \dots, s_m \in\Sigma^{+n}$}. If 
we want to state this list explicitly, we
separate them by colons writing $s=s_1: s_2 : \dots : s_m$. 
By $\lvert s \rvert$ we denote the number of $n$-stacks $s$
consists of, i.e., $\lvert  s\rvert = m$. We call $\lvert s \rvert$ 
the \emph{width} of $s$. We also use the notion of the \emph{height} of
$s$. This is  $\height(s):=\max\{\lvert 
s_i \rvert: 1\leq i \leq m\}$, i.e., the width of the widest
$n$-stack occurring in $s$.

Let $s$ and $s'$ be $(n+1)$-stacks such that
$s=s_1:s_2:\dots: s_m \in \Sigma^{+(n+1)}$ and
$s'=s_1':s_2':\dots: s_l' \in \Sigma^{+(n+1)}$. We write 
$s:s'$ for the concatenation  $s_1: s_2 : \dots :s_m:s_1':s_2':\dots:s_l'$.

If $s\in\Sigma^{+(n-1)}$, we denote by $[s]$ the $n$-stack
that only consists of a list of length $1$ that contains 
$s$. We regularly omit the brackets if no confusion arises. 

Let $\Sigma$ be some finite alphabet with a distinguished
bottom-of-stack symbol $\bot\in\Sigma$. 
The \emph{initial stack} $\bot_l$ of level $l$ over $\Sigma$ is inductively
defined by  $\bot_1:=[\bot]$ and $\bot_n:=[\bot_{n-1}]$.

Before we formally define the stack operations, we
introduce an auxiliary function $\TOP{k}$ that returns the topmost
entry of the topmost $k$-stack. 
Let \mbox{$s=s_1:s_2:\dots: s_n\in\Sigma^{+l}$}
be some stack and let $1\leq k \leq l$.
We define the \emph{topmost level $k-1$ stack of
  $s$} to be  
\mbox{$\TOP{k}(s):=
\begin{cases}
  s_n & \text{if } k=l,\\
  \TOP{k}(s_n) &\text{otherwise.}
\end{cases}$}
\begin{definition}
  For $s=s_1:s_2:\dots: s_n\in \Sigma^{+l}$,
  for $\sigma\in\Sigma\setminus\{\bot\}$,
  for $1\leq k \leq l$ and for $2\leq j \leq l$, we define
  the stack operations  
  \begin{align*}
    \Clone{j}(s):=
    &\begin{cases}
      s_1: s_2 : \dots : s_{n-1}: s_n : s_n & \text{if } j = l\geq 2,\\
      s_1: s_2 : \dots : s_{n-1}: \Clone{j}(s_n) & \text{otherwise.}
    \end{cases}\\
    \Push{\sigma}(s):=
    &\begin{cases}
      s\sigma  &\text{if } l=1,\\
      s_1: s_2 : \dots : s_{n-1}: \Push{\sigma}(s_n)& \text{otherwise.}
    \end{cases} \\
    \Pop{k}(s):=
    &\begin{cases}
      s_1: s_2 : \dots : s_{n-1}: \Pop{k}(s_n) & \text{if } k<l,\\
      s_1: s_2 : \dots : s_{n-1} & \text{if } k=l, n>1,\\
      \text{undefined} & \text{otherwise, i.e.,} k=l, n=1.
    \end{cases}
  \end{align*}
  The \emph{set of level $l$ operations} is denoted by
  $\Op_l$.
\end{definition}
For $2\leq i \leq n$ and $\sigma\in\Sigma$, we call $\Push{\sigma}$
a push of level $1$ and $\Clone{i}$ a push of level $i$. 

For stacks $s, s'$ we write $s\leq s'$ and say $s$ is a substack of $s'$
if $s$ is generated from $s'$ by application of a
sequence of pop operations (of possibly different levels).
Note that on $1$-stacks, i.e., on words, $\leq$ coincides with the
usual prefix relation. 

Having defined $l$-stacks, we present pushdown
systems of level $l$. 

\begin{definition}
  A \emph{pushdown system} of level $l$ ($l$-PS) is
  a tuple 
  $\mathcal{S} = (Q,\Sigma, \Delta, q_0)$
  where $Q$ is a
  finite set of states, $\Sigma$  a
  finite stack  alphabet with a distinguished bottom-of-stack symbol
  $\bot\in\Sigma$, $q_0\in Q$ the initial state, and 
  \mbox{$\Delta\subseteq
    Q\times \Sigma \times Q \times \Op_l$}
  the transition relation.

  An  $l$-\emph{configuration} is a pair $(q,s)$ where $q\in Q$ and
  $s\in \Sigma^{+l}$. 
  For \mbox{$q_1,q_2\in Q$}, \mbox{$s,t\in \Sigma^{+l}$} and for
  $\delta= (q_1, \sigma,  q_2, op)\in\Delta$, we 
  define the $\delta$-relation $\trans{\delta}$ as follows. Set
  $(q_1,s) \trans{\delta} (q_2, t)$ if  
  $\op(s)=t$ and
  \mbox{$\TOP{1}(s) = \sigma$}. 
  We call $\trans{}:=\bigcup_{\delta\in\Delta} \trans{\delta}$ the
  transition relation of $\mathcal{S}$. 
\end{definition}
\begin{definition}
Let $\mathcal{S}$ be an $l$-PS. 
A \emph{run} $\rho$ of $\mathcal{S}$ is a sequence of configurations
that are connected by transitions, i.e., a sequence
\begin{align*}
  c_0 \trans{\delta_1} c_1\trans{\delta_2} c_2 \trans{\delta_3} \cdots
  \trans{\delta_n}c_n.    
\end{align*}
\end{definition}
We also write $\rho(i):=c_i$ for the $i$-th configuration occurring
within $\rho$. 
We call $n$ the \emph{length} of $\rho$ and set $\length(\rho):=n$. 
If some run $\pi$ is an initial segment of the run $\rho$, we write
$\pi \preceq \rho$. We write $\pi \prec \rho$ if $\pi$ is a proper
initial segment of $\rho$.

If $\pi$ and $\rho$ are runs
\begin{align*}
  &\pi = c_0 \trans{\delta_1} c_1\trans{\delta_2} c_2 \trans{\delta_3} \cdots
  \trans{\delta_n} c_n  \\
  &\rho = c_n \trans{\delta_{n+1}} c_{n+1}\trans{\delta_{n+2}} c_{n+2}
  \trans{\delta_{n+3}} \cdots 
  \trans{\delta_{n+m}}c_{n+m},   
\end{align*}
then we denote by $\pi\circ\rho$ the
composition of $\pi$ 
and $\rho$ which is the run from $c_0$ to $c_{n+m}$ defined by 
\begin{align*}
  \pi\circ\rho = c_0 \trans{\delta_1} c_1\trans{\delta_2} c_2
  \trans{\delta_3} \cdots 
  \trans{\delta_n}c_n \trans{\delta_{n+1}} c_{n+1}\trans{\delta_{n+2}} c_{n+2}
  \trans{\delta_{n+3}} \cdots  \trans{\delta_{n+m}}c_{n+m}.   
\end{align*}

\section{The Nested Pushdown Tree Hierarchy}
\label{sec:HONPT}

Generalising the definition of nested pushdown trees
(cf.\ \cite{Alur06languagesof}),  
we define a hierarchy of higher-order nested pushdown
trees. 
A nested pushdown tree is the unfolding of the configuration graph of
a pushdown system expanded by a new relation (called
\emph{jump relation}) which connects each push operation with the
corresponding pop operations. 
Since higher-order pushdown systems have push and pop operations for
each stack level, there is no unique generalisation of this concept to
trees generated by higher-order pushdown systems. 
We choose the following (simplest) version: we connect corresponding push
and pop operations of the highest stack level. This choice ensures
that the jump edges form a well-nested relation. We discuss possible
other choices for the definition at the end of this section. 

\begin{definition}
  Let $\mathcal{N}=(Q,\Sigma,q_0,\Delta)$ be an $n$-PS.
  The  \emph{level $n$ nested pushdown tree} (\mbox{$n$-\HONPT})
  $\HONPT(\mathcal{N})=(R, (\trans{\delta})_{\delta\in\Delta},
  \jumpedge)$ is the  
  unfolding of the pushdown graph of $\mathcal{N}$ from its initial
  configuration expanded by a new 
  \emph{jump relation} $\jumpedge$ which is formally defined as
  follows.
  \begin{itemize}
  \item $R$ is the set of all runs of $\mathcal{N}$ that start in the
    initial configuration $(q_0, \bot_n)$ of $\mathcal{N}$,
  \item $\trans{\delta}$ contains a pair
    of runs 
    $(\rho_1,\rho_2)$ if $\rho_2=\rho_1\circ \rho'$ for some run
    $\rho'$ of length $1$ such that $\rho' = \rho_1(\length(\rho_1))
    \trans{\delta} c$ for some 
    configuration $c$, and
  \item $\jumpedge$ is the binary relation such that
    $\rho_1\jumpedge \rho_2$
    if $\rho_2$ decomposes as
    $\rho_2=\rho_1\circ \rho$ for some run $\rho$ of length $m\geq 2$
    starting and ending 
    in the same stack $s$
    such that
    \begin{align*}
      &\rho(0) \text{ and } \rho(1) \text{ are connected by a level }
      n\text{ push operation}, \\  
      &\rho(m-1) \text{ and } \rho(m) \text{ are connected by a level }
      n \text{ pop operation, and}\\
      &\rho(i)\neq(\hat q, s)\text{ for all } 1\leq i < m
      \text{ and all }\hat q\in Q.
    \end{align*}
  \end{itemize}
  We use $\trans{}$ as abbreviation for the union
  $\trans{}:=\bigcup_{\delta\in\Delta} \trans{\delta}$. 
\end{definition}
\begin{remark}
  Note that for all $\rho_2\in\HONPT(\mathcal{N})$ there is at most
  one $\rho_1\in\HONPT(\mathcal{N})$ with $\rho_1\jumpedge
  \rho_2$, but for each $\rho_1\in\HONPT(\mathcal{N})$ there may be
  infinitely many $\rho_2\in\HONPT(\mathcal{N})$ with
  $\rho_1\jumpedge\rho_2$. 
\end{remark}
\begin{example}
  Consider the $2$-PS $\mathcal{N}=(Q, \Sigma, q_0, \Delta)$ with
  \begin{itemize}
  \item $Q=\{q_0,q_1,q_2\}$,
  \item $\Sigma=\{\bot, a\}$, and
  \item $\Delta=\{ (q_0, \bot, q_1, \Clone{2}), (q_1, \bot, q_1,
    \Push{a}), (q_1, a, q_1, \Push{a}), (q_1, a, q_2, \Pop{2})\}$. 
  \end{itemize}
  The $2$-\HONPT $\mathcal{N}$ is the  graph
  $\HONPT(\mathcal{N})$ depicted in Figure
  \ref{fig:NPTExample}. 
  \begin{figure}[h]\label{fig:NPTExample}
    \centering
      \begin{tikzpicture}
    [
    level distance=3 cm,
    growth parent anchor=south,
    konf/.style={
      rectangle, draw, anchor=south},
    Up/.style={
      grow=up, level distance=2.3 cm},
    Right/.style={
      grow=right, level distance=2 cm},
    invjump/.style={<-left hook, shorten >=2pt, shorten <=2pt, thick}
    ]
    \node [konf] (root) {$\bot$}
    [->]
    child [grow=down] { 
      node[konf,matrix] { \node{$\bot$}; \pgfmatrixnextcell \node{$\bot$}; \\}
      child [Right] {
        node[konf, matrix] 
        { \pgfmatrixnextcell \node{$a$};\\
          \node{$\bot$}; \pgfmatrixnextcell \node{$\bot$}; \\}
        child[Up]{ node[konf] {$\bot$}}
        child[Right] {
          node[konf, matrix] 
          { \pgfmatrixnextcell \node{$a$};\\
            \pgfmatrixnextcell \node{$a$};\\
            \node{$\bot$}; \pgfmatrixnextcell \node{$\bot$}; \\}
          child[Up] { node[konf] {$\bot$}
          }
          child[Right] {
            node[konf, matrix] 
            { \pgfmatrixnextcell \node{$a$};\\
              \pgfmatrixnextcell \node{$a$};\\
              \pgfmatrixnextcell \node{$a$};\\
              \node{$\bot$}; \pgfmatrixnextcell \node{$\bot$}; \\}
            child[Up] { node[konf] {$\bot$}
              child[grow=right , arrows= - , dotted,
              thick,draw opacity =  1, shorten <=5pt]{
                node[anchor=south,transparent]{$\bot$}}
            }
            child[parent anchor=east,  
            grow=right, arrows= - , dotted,
            thick,draw opacity =  1, shorten <=5pt ]{ 
              node[minimum height=2cm, anchor=south] {}
            }
          }
        }
      }
    }
    ;
    \draw[invjump, bend angle=15, bend right ]   (root-1-1-1) to (root);
    \draw[invjump, bend angle=15, bend right ]   (root-1-1-2-1) to (root);
    \draw[invjump, bend angle=15, bend right ]   (root-1-1-2-2-1) to (root);

  \end{tikzpicture}

  \caption{Example of a $2$-\HONPT (edge labels are omitted)}    
  \end{figure}
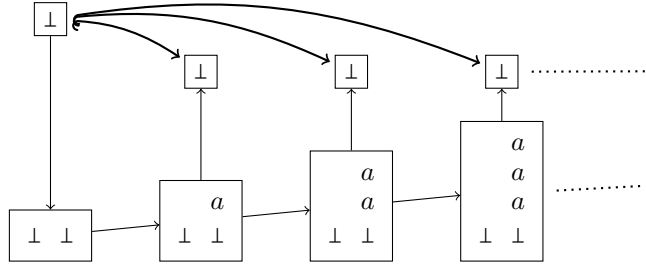

\end{example}

For each $l<l'$ the class of nested pushdown trees of
level $l$ are uniformly first-order interpretable in the class of
nested pushdown trees of level $l'$. 
If $l=1$, one just replaces any $\Push{\sigma}$ transition by a 
$\Clone{l'}$ transition followed by a $\Push{\sigma}$
transition. Furthermore, one replaces
each $\Pop{1}$ transition by a $\Pop{l'}$ transition. 
In all other cases, we just replace $\Clone{l}$ and $\Pop{l}$ by
$\Clone{l'}$ and $\Pop{l'}$. 

\paragraph{Comparison of Nested Pushdown Tree Hierarchy and 
  Other Hierarchies}
The hierarchy of higher-order nested pushdown trees is a hierarchy
strictly extending the hierarchy of trees generated by higher-order
pushdown systems  (without $\varepsilon$-contraction!). 
Furthermore, it is first-order interpretable in the
collapsible pushdown graph hierarchy.\footnote{The hierarchy of collapsible
  pushdown graphs is the class of configuration graphs of
  \emph{collapsible pushdown systems}
  (cf.~\cite{Hague2008}). Collapsible pushdown systems of level 
  $l$ are defined analogously to $l$-PS but with transitions that may
  also use a  
  stack-operation called collapse. This new operation summarises
  several pop operations.}
We now sketch the proof of this claim. 
Fix an $l$-PS $\mathcal{N}$. $\mathcal{N}$ generates an $l$-NPT
$\mathfrak{N}$. 
Each node of $\mathfrak{N}$
represents a run of $\mathcal{N}$ starting in the
initial configuration. A run can be seen as a list of configurations. 
This is a list of pairs of states and stacks. Pushing the state on top
of the stack, a run can be represented as a list of $l$-stacks. Let
$s_1, s_2, s_3, \dots, s_n$ be the stacks representing some run. Then 
\mbox{$s_1:s_2:\dots:s_n$} is an  $(l+1)$-stack representing the run. 
In this representation, an edge in the $l$-NPT corresponds to the
extension of $s_1:s_2:\dots:s_n$ to a list
\mbox{$s_1:s_2:\dots:s_n:s_{n+1}$} where $s_{n+1}$ is generated from
$s_n$ by removing the state written on top of $s_n$, applying a stack
operation and writing the new final state on top of the stack. 
Hence, we can use this representation and define a level $(l+1)$-PS
$\mathcal{S}$ such that the tree generated by $\mathcal{N}$ is
first-order interpretable in the configuration graph of $\mathcal{S}$. We 
interpret each edge of $\HONPT(\mathcal{N})$ as a path of length $4$
in the configuration graph of $\mathcal{S}$. 
Such a path performs the 
operations $\Clone{n}$ - $\Pop{1}$ - $\op$ - $\Push{q}$ for
some level $l$ operation $\op$. 
Replacing $\mathcal{S}$ by a certain collapsible pushdown system
$\hat{\mathcal{S}}$, we can generate the same graph but with additional
collapse-transitions that form exactly the reversals of the jump edges
of $\mathfrak{N}$. This means that if $a,b\in\mathfrak{N}$ such that
$a\jumpedge b$, and $a', b'$ are the representatives of $a$ and
$b$ in the configuration graph of $\hat{\mathcal{S}}$, then there is a
collapse edge from $b$ to $a$. 
A detailed proof of this claim can be found in \cite{KartzowPHD}.

Of course, these observations immediately lead to the following
questions concerning the relationship between nested pushdown trees
and collapsible pushdown graphs. In the following, we use the term
interpretation for any kind of logical interpretation that allows to
transfer properties definable in some fixed
logic from one class to another. 
\begin{enumerate}
\item Can level $n$ nested pushdown trees be interpreted in the class
  of level $n$ collapsible pushdown graphs?
  We have a weak conjecture that the answer to this question is no. 
  Since $n$-NPT contain the trees generated by $n$-PS, they seem to be
  on level $2.5$ with respect to graph hierarchies: recall that the
  graphs of $(n+1)$-PS are obtained from the graphs of $n$-PS using an
  unfolding followed by a monadic second order interpretation. The
  unfolding operation yields the class of trees of $n$-PS. In this
  sense, $n$-NPT expand level  $n+\frac{1}{2}$ of the pushdown graph
  hierarchy and we conjecture that even this class is not
  interpretable in collapsible pushdown graphs of level $n$. 
\item Is there a meaningful notion of logical interpretation that
  allows to interpret the class of collapsible pushdown graphs of
  level $n+1$  (or $n$) in the class level $n$ nested pushdown trees?
  Note that we know that collapsible pushdown graphs of level $3$ are
  not first-order interpretable in $2$-NPT because the former do not
  enjoy decidable first-order theories while the latter do. 
  Also note that collapsible pushdown graphs of level $1$ are just
  pushdown graphs and level $1$-NPT expand the unfoldings of such
  graphs. As discussed in question (1), we expect that there is no
  interpretation of level 1 collapsible pushdown graphs in $1$-NPT. 
  Thus, we tend to expect that the answer to both questions is no, but
  further investigation is needed for clarification. 
\end{enumerate}

\paragraph{Other possible definitions of $n$-NPT}

For the definition of jump edges, one could consider push operations
of every level $l\leq n$. Here the following possibilities seem to be
plausible.
\begin{enumerate}
\item Only connect runs $\rho_1$ and $\rho_2$ if $\rho_2$ extends
  $\rho_1$ by a run $\rho$ such that $\rho$ starts with
  a $\Push{l}$-operation and ends with a $\Pop{l}$ operation that
  removes the $l-1$ stack created in the1 first transition of $\rho$. 
  This choice implies that not every push operation induces a jump
  edge along every path of the nested pushdown tree: if a $\Push{l}$
  operation is followed by a $\Pop{l'}$ operation for $l<l'$ the stack
  that was created by the push is removed together with the higher
  level stack it is contained in. Thus, this generalisation lacks the
  nice property of $1$-NPT that if a jump edges starts at some run
  $\rho$, then $\rho$ has a jump successor in each branch starting at
  $\rho$. On the positive side we expect that this definition would
  still allow an interpretation of $n$-NPT in collapsible pushdown
  graphs of level $n+1$ as described in the previous paragraph.
\item A jump edge starts in each branch whenever we perform a push
  operation of any level and it points to the first descendant along
  this branch where the newly created stack disappears for the first
  time. As opposed to the previous possibility, here we would have a
  jump edge from $\rho$ to $\rho \trans{\Push{l}} c' \trans{\Pop{l'}}
  c$ if $l'>l$. This choice leads to fairly complicated
  structures. In the models discussed before, the indegree is always
  bounded by $2$. In this version of nested pushdown trees the
  indegree is finite but unbounded because many $\Push{l}$ operations
  may have the same  corresponding  $\Pop{l'}$ operation if $l'>l$. 
  In particular, it seems to be difficult to find an interpretation of
  this version in the class of collapsible pushdown graphs. The given
  interpretation for the version that we chose relies heavily on the
  functionality of the inverse jump edges.
\end{enumerate}
Both possible extensions of our definition of $n$-NPT seem to make the
structures more complicated. We do not know whether the techniques
presented in this paper can be adapted to treat those structures as
well.

\subsection{Towards Model Checking on Level 2 NPT}
\label{sec:TowardsModelChecking}
 In the following,  we develop an \FO{} model
checking algorithm on $2$-NPT.
In fact, we prove that the general approach via the
dynamic-small-witness property developed in Section
\ref{sec:EFGame} is applicable in this case. In other words, we
prove that we can compute a  
finitary constraint for Duplicator's strategy on an arbitrary $2$-NPT
$\mathfrak{N}$. 

Fix some $2$-PS $\mathcal{N}$ of level $2$ and set
$\mathfrak{N}:=\HONPT(\mathcal{N})$.
We show the following. 
If \mbox{$\mathfrak{N}, \bar \rho \models\exists
x\varphi$} for some
formula $\varphi\in\FO{}$, then there is a short 
witness \mbox{$\rho\in \mathfrak{N}$} for this existential
quantification. Here, the length of an  
element is given by 
the length of the run representing this element. 
We consider a run to be short if
its size is bounded in terms of the length of the 
runs representing the parameters $\bar\rho$. 

We stress that
locality arguments in the spirit of Gaifman's theorem do not
apply in this setting: the jump edges tend to make the diameter
of $2$-NPT small (recall that the $2$-NPT in figure
\ref{fig:NPTExample} has diameter 4). 

The rough picture of our proof is as follows.
We analyse the $\alpha$-round Ehrenfeucht-\Fraisse game on two
copies of \mbox{$\mathfrak{N}$} and show that Duplicator has a
restricted winning-strategy. 
Our main technical tool is the concept of
\emph{relevant ancestors}. For each element of
$\mathfrak{N}$, the relevant 
$l$-ancestors are a finite set of initial subruns of this element. 
Intuitively, some run $\rho'$ is a relevant $l$-ancestor of a  run $\rho$
if it is an  ancestor of $\rho$ which may be connected to $\rho$ via a path of
length up to $l$ that witnesses the fact that $\rho'$ is an ancestor
of $\rho$. It turns out that there are at most $4^l$ such ancestors. 
Surprisingly, the set of $2^l$-ancestors 
characterises the $\FO{l}$-type of $\rho$. 
Thus, Duplicator has a winning strategy choosing small
runs if for every element of $\mathfrak{N}$ there is a small one that
has an isomorphic set of relevant ancestors.

The analysis of relevant ancestors reveals that
a relevant
ancestor $\rho_1$ is connected to the next one, say $\rho_2$,
by either a single transition or 
by a run $\pi$ of a certain kind. 
This run $\pi$ satisfies the following conditions:
$\rho_2$ decomposes as 
$\rho_2=\rho_1\circ\pi$, 
the initial
stack of $\pi$ is $s:w$ where $s$ is some stack and $w$ is some
word. The final stack of $\pi$ is $s:w:v$ for some word $v$ and
$\pi$ does never pass a proper substack of $s:w$.

Due to this result, a typical set of relevant ancestors is of the form
\begin{align*}
  \rho_1 \prec \rho_2 \prec \rho_3 \prec \dots \prec \rho_m=\rho,
\end{align*}
where
$\rho_{n+1}$ extends $\rho_n$ by either one transition or by a run
that extends the last stack of $\rho_n$ by a new word $v$. 
If we want to construct a short run $\rho'$ with isomorphic relevant
ancestor set, we have to provide short runs 
\begin{align*}
  \rho'_1 \prec \rho'_2 \prec \rho'_3 \prec \dots \prec \rho'_m=\rho'  
\end{align*}
where $\rho_{n+1}'$ extends $\rho_n'$ in exactly the same manner as
$\rho_{n+1}$ extends $\rho_n$. 

We first concentrate on one step of this
construction.  Assume that $\rho_1$ ends in some configuration $(q,s:w)$ and
$\rho_2$ extends $\rho_1$ by a run creating the stack $s:w:v$.
How can we find another stack $s'$ and words $w',v'$ such that
there is a short run $\rho_1'$ to $(q,s':w')$ and a short run $\rho_2'$ that
extends $\rho_1'$ by a run from $(q,s':w')$ to the stack $s':w':v'$?

We introduce a family of equivalence relations on words
that preserves the existence of such runs. 
If we find some $w'$ that is equivalent to $w$
with respect to the $i$-th equivalence relation, then for each run from
$s:w$ to $s:w:v$ there is a run from $s':w'$ to $s':w':v'$ for $v$
and $v'$ equivalent with respect to the $(i-1)$-st equivalence
relation. 

Let us explain these equivalence relations. 
Let $\rho_1$ be a run to some stack $s:w$ and let $\rho_2$ be a run
that extends $\rho_1$ and ends in a stack $s:w:v$. 
We can prove that this extension is of the form
$  \op_n\circ\lambda_{n}\circ\op_{n-1}\circ\lambda_{n-1}\circ\dots\circ
  \op_1\circ\lambda_1  
$
where the $\lambda_i$ are loops, i.e., runs that start and end with
the same stack and $\op_{n},
\op_{n-1}, \dots, \op_1$ is the minimal sequence generating $s:w:v$
from $s:w$. 
Thus, we are especially interested in the loops of each prefix
$\Pop{1}^k(w)$ of $w$ and 
each prefix $\Pop{1}^k(w')$ of $w'$. 
For this purpose we consider the word models of $w$ and $w'$ enriched by 
information on runs between certain prefixes of $w$ or $w'$.
Especially, each prefix is annotated with the number of possible loops
of each prefix.  
$w$ and $w'$ are equivalent with respect to the first equivalence
relation if the $\FO{k}$-types
of their enriched word structures coincide. 
The higher-order equivalence relations are then defined as follows.
We colour every element of the word model of some word $w$ by the
equivalence class of the corresponding prefix with respect to the
$(i-1)$-st equivalence relation. 
For the $i$-th equivalence relation we compare the $\FO{k}$-types of
these coloured word models. This means that 
two words $w$ and $w'$ are equivalent with respect to the $i$-th
equivalence relation if the $\FO{k}$-types of their word models
expanded by predicates encoding the $(i-1)$-st equivalence class of
each prefix coincide.

This iteration of equivalence of prefixes leads to the following
result. Let $w$ and $w'$ be equivalent with respect to the $i$-th
relation. Then we can transfer runs creating $i$ words in the
following sense: if
$\rho$ is a run creating $w:v_1:v_2:\dots:v_i$ from $w$, then there is
a run $\rho'$ creating $w':v'_1:v'_2:\dots:v'_i$ from $w'$ such that
$v_k$ and $v'_k$ are equivalent with respect to the $(i-k)$-th
relation. 
This property then allows us to construct isomorphic relevant ancestors
for a given set of relevant ancestors of some run $\rho$. We only have
to start with a 
stack $s':w'$ such that $w'$ is $i$-equivalent to the topmost word of
the minimal element of the relevant ancestors of $\rho$ for some large
$i\in\N$. 

This observation reduces the problem of
constructing runs with isomorphic relevant ancestors to the problem of
finding runs whose last configurations have equivalent topmost
words (with respect to the $i$-th equivalence relation for some
sufficiently large $i$) such that one of these runs is always short. 

We solve this problem by developing  shrinking constructions
that allow the preservation
of the equivalence class of the topmost word of the final configuration
of a run  while shrinking the length of the run. 

Putting all these results together, 
for every nested pushdown tree of level $2$, we can compute a finitary
constraint $S$ such that
Duplicator has an
$S$-preserving strategy. This shows that the general model checking
algorithm from Section \ref{sec:EFGame} solves the the $\FO{}$ 
model checking problem on $2$-NPT. 

\subsection{Outline of the Proof Details}

In the next section, we discuss the theory of \emph{loops} first
developed in \cite{Kartzow10}. We also develop shrinking lemmas for
long runs.  
Then we introduce the central notion of \emph{relevant ancestors}
and develop some basic theory concerning these sets in Section
\ref{sec:RelevantAncestors}. 
In Section \ref{subsec:EquivalenceonWords} we define
equivalence relations on words and trees which can be used to construct
runs with isomorphic relevant 
ancestors. In Section \ref{sec:EquivalenceonTuples}, we first lift
these equivalences on stacks to equivalences on tuples of elements of
$2$-NPT by pairwise comparison of the equivalence type of the topmost
stacks of each relevant ancestor of each element in the tuples. We
then prove that the preservation of equivalence classes of relevant
ancestors is a winning strategy in the Ehrenfeucht-\Fraisse game. 
Furthermore, Duplicator can always choose small representatives. 
Finally, in  Section \ref{sec:FODecidability} we derive an \FO{} model
checking algorithm on $2$-NPT. 
In that section we also show that this algorithm restricted to 
the class of all $1$-NPT is in $2$-EXPSPACE. Unfortunately, we do not
obtain any complexity bounds for the algorithm on the class of
$2$-NPT.

\section{Shrinking Lemmas for Runs of 2-PS}
\label{sec:TheoryRuns}

In this section, we analyse runs of $2$-PS between certain
configurations. Especially, we look at runs starting at the initial
configuration and at runs extending its starting stack $s$ by some
word $w\in\Sigma^+$, i.e., runs starting in $s$ and ending in $s:w$
that do not visit substacks of $s$. 
If there is a run $\rho$ of one of these types, we will see that there
is also a short run $\rho'$ with the same initial and final
configuration as $\rho$. ``short'' means that $\length(\rho')$ is
bounded by some function depending on the pushdown system, on the
width and height of $s$ and on the length of $w$. 
We first developed this theory in \cite{Kartzow10}. 
We briefly recall this theory and sketch the main proofs.
For a detailed presentation, 
see \cite{KartzowPHD}. 

We first introduce the concept of milestones and generalised
milestones of a $2$-stack $s$. 
A stack $t$ is a generalised milestone of $s$ if every run of every
$2$-PS that leads from the initial
configuration to $s$ has to pass $t$. Milestones are those generalised
milestones that are substacks of $s$. 
Generalised milestones of $s$ induce a natural decomposition of any
run from the initial 
configuration to some configuration with stack $s$. 
Due to a result of Carayol (that we state soon),
generalised milestones can be characterised as follows.

\begin{definition}
  Let $s=w_1:w_2:\dots :w_k$ be a $2$-stack. 
  We define the set of \emph{generalised milestones of $s$}, denoted by
  $\genMilestones(s)$, as follows.
  A stack $m$ is in $\genMilestones(s)$ if 
  \begin{align*}
    &m=w_1:w_2:\dots: w_i:v_{i+1} \text{ where }
    0\leq i<k,\\
    &w_i \sqcap w_{i+1}\leq v_{i+1}\text{ and }\\
    &v_{i+1}\leq w_i\text{ or }v_{i+1}\leq w_{i+1}.
  \end{align*}
  where $ w \sqcap v $ denotes the maximal common prefix of $w$ and
  $v$. 
  If $v_{i+1} \leq w_{i+1}$, we call $m$ a milestone of $s$. The set
  of milestones of $s$ is denoted by $\Milestones(s)$. 
\end{definition}
\begin{remark} \label{Rem:NumberMilestones}
  Note that
  $\lvert \genMilestones(s)\rvert \leq 2 \cdot \height(s) \cdot
  \lvert s\rvert$.
\end{remark}

In order to show that this definition fits
our informal description, we define the notions of
loops and returns. Loops appear as a central notion in our formulation
of Carayol's
result. Moreover, the theory of loops and returns is also the key
ingredient to our shrinking constructions. We will generate short runs
from longer ones by replacing long loops by shorter ones. 

\begin{definition} \label{DefLoop}
  Let $s=t:w$ be some stack with topmost word $w$ and $q,q'\in Q$.
  A run $\lambda$  
  from $(q,s)$ to $(q',s)$ is called \emph{loop} if it does not pass 
  $t$. It is called a \emph{high loop} if it additionally
  does not pass 
  $\Pop{1}(s)$.  
  \label{DefReturn}
  A run $\rho$ from $(q,t:w)$ to $(q',t)$ is called \emph{return} if
  it visits $t$ 
  only in its final configuration.
\end{definition}

The following lemma of Carayol shows that generalised milestones and
loops play a 
crucial role in understanding the existence of runs.

\begin{lemma}[\cite{Carayol05}] \label{Lemma:Carayol05Milestones}
  For each stack $s$ there is
  a minimal sequence of operations $\op_1, \op_2, \dots, \op_n\in\{
  \Push{\sigma}, \Pop{1}, \Clone{2}\}$ such that 
  $s=\op_n(\op_{n-1}(\dots(\op_1(\bot_2))\dots ))$.

  For each $0\leq j\leq n$, 
  the stack $\op_j (\op_{j-1}(\dots\op_0(\bot_2)))$      
  is a  generalised milestone of $s$ and every  generalised milestone
  is generated by such a sequence. 
  
  Furthermore, every run $\rho$ from $(q_0, \bot_2)$ to $s$
  passes all generalised milestones of $s$. More precisely, $\rho$
  decomposes as  
  $\rho = \lambda_1 \circ \rho_1 \circ \lambda_{2} \circ \rho_{2}
  \circ \ldots \circ \lambda_n \circ \rho_n \circ \lambda_{n+1}$, where
  $\rho_i$ is a run of length $1$ that performs the operation $\op_i$
  and $\lambda_i$ is a loop of the $i$-th generalised milestone of $s$. 
\end{lemma}

Thus, loops play a crucial
role in understanding the existence of runs from one configuration to
another. Moreover, returns and loops of smaller stacks provide a
natural decomposition of loops. Hence, counting the number of
loops and returns of 
certain stacks can be used in order to count the number of runs
between certain configurations. We define the following notation
concerning counting loops and returns up to some threshold
$k\in\N$.\footnote{Counting the size of a set up to some threshold
  $k\in\N$ means that  we assign the value $k$  to the set if it
  contains at least $k$ elements, i.e., $k$ stands for ``$k$ or
  more''.}

\begin{definition} \label{Def:ReturnFunc}
  Let $\mathcal{S}=(Q,\Sigma, q_I, \Delta)$ be a $2$-PS. 
  Let
  $\ReturnFunc{k}_{\mathcal{S}}(s):Q\times Q \rightarrow
  \{0,1,\dots,k\}$ be the function that assigns $(q,q')$ to the number
  of returns from $(q,s)$ to $(q',\Pop{2}(s))$ up to threshold $k$. 
  \label{Def:LoopFunc}
  Analogously, let 
  $\LoopFunc{k}_{\mathcal{S}}(s)$ and
  $\HighLoopFunc{k}_{\mathcal{S}}(s)$ count the loops and high loops,
  respectively,  from
  $(q,s)$ to 
  $(q',s)$ up to  threshold $k$.
  
  If $\mathcal{S}$ is clear from the context, we omit it and write
  $\LoopFunc{k}$ for $\LoopFunc{k}_{\mathcal{S}}$, etc. 
\end{definition}
\begin{remark}
  Note that loops and returns of some stack $s:w$ never look into
  $s$. Hence, $\LoopFunc{k}(s:w) = \LoopFunc{k}(t:w)$ 
  for every word $w$ and arbitrary stacks $s$ and $t$. If $s$ and $t$ are
  both 
  nonempty stacks, the analogous statement
  holds for $\ReturnFunc{k}$. 
  Thus, we will write $\LoopFunc{k}(w)$ for $\LoopFunc{k}(s:w)$ where
  $s$ is an arbitrary stack. Analogously
  $\ReturnFunc{k}(w):=\ReturnFunc{k}(s:w)$ for any nonempty stack $s$.  
\end{remark}

The numbers of loops and of returns of some stack $s$ inductively depend
on the number of loops and returns of its substacks. 
In preparation for the proof of this observation we  recall the notion of
\emph{stack replacement} introduced in \cite{Blumensath2008}.

\begin{definition}
  For some level $2$ stack $t$ and some substack $s\leq t$ we say that
  $s$
  is a \emph{prefix} 
  of $t$ and write 
  $s\prefixeq t$, if there are $n\leq m
  \in\N$ such that $s=w_1:w_2:\dots : w_{n-1}: w_n$ and 
  $t=w_1:w_2 \dots :w_{n-1} : v_n : v_{n+1} : \dots : v_m$ such that $w_n\leq
  v_j$ for all $n\leq j \leq m$. This means that $s$ and $t$ agree on the
  first $\lvert s\rvert -1$ words and the last word of $s$ is a prefix of all
  other words of $t$. 
  We extend this definition to runs by writing
  $s\prefixeq \rho$ if $\rho$ is some run and
  $s\prefixeq t_i$ for $t_i$ the stack at $\rho(i)$ and for all
  $i\in\domain(\rho)$.  

  Let $s,t,u$ be $2$-stacks such that $s\prefixeq t$. Assume that 
  \begin{align*}
   &s=w_1:w_2:\dots : w_{n-1}: w_n,\\
   &t=w_1:w_2 \dots :w_{n-1} : v_n : v_{n+1} : \dots : v_m, \text{
     and}\\
   &u=x_1: x_2: \dots : x_p
  \end{align*}
  for numbers $n,m,p\in\N$ with $n\leq m$. 
  For $n\leq i \leq m$, let $\hat v_i$ be the unique word such that
  $v_i=w_n\circ \hat v_i$. Set
  \begin{align*}
    t[s/u]:= x_1: x_2: \dots: x_{p-1} : (x_p\circ \hat v_n) :
    (x_p \circ \hat v_{n+1}): \dots : (x_p\circ \hat v_m)
  \end{align*}
  and call $t[s/u]$ 
  \emph{the stack obtained from $t$ by replacing the prefix $s$ by
    $u$}. 
  For $c=(q,t)$ a configuration, set $c[s/u]:=(q,t[s/u])$. 
\end{definition}

\begin{lemma}[\cite{Blumensath2008}] \label{Lem:BlumensathHOLevel2}
  Let $\rho$ be a run of some $2$-PS
  $\mathcal{S}$. Let $s,u\in \Sigma^{+2}$ be stacks such that
  $s\prefixeq \rho$ and $\TOP{1}(u)=\TOP{1}(s)$. 
  Then the function 
  $\rho[s/u]$ defined by \mbox{$\rho[s/u](i):=\rho(i)[s/u]$} is a run of
  $\mathcal{S}$.  
\end{lemma}
We extend the notion of prefix replacement to
runs that are $s$-prefixed at the beginning and at the end and that
never visit the substack $\Pop{2}(s)$.  
Such a run may contain ``holes'', i.e., parts that are not prefixed by
$s$. We show that these holes are always loops or returns. Thus, we
can replace a prefix by another one if these two share the same types
of loops and returns. 
The following lemma prepares this new kind of prefix replacement. 

\begin{lemma} \label{Lem:prefixedstartstopdecomposition}
  Let $\mathcal{N}$ be some $2$-PS and let
  $\rho$ be a run of $\mathcal{N}$ of length $n$. 
  Let $s$ be a stack with topmost word $w:=\TOP{2}(s)$ such that
  \begin{align*}
    s\prefixeq\rho(0), s\prefixeq\rho(n), \text{ and }
    \lvert s \rvert \leq \lvert \rho(i) \rvert\text{ for all }0\leq i
    \leq n.
  \end{align*}
  There is a unique sequence $0=i_0 \leq j_0 < i_1 \leq j_1 < \dots <
  i_{m-1} \leq  j_{m-1} < i_m \leq j_m=n$ such that
  \begin{enumerate}
  \item $s\prefixeq \rho{\restriction}_{[i_k,j_k]}$ for all $0\leq k
    \leq m$ and 
  \item $\TOP{2}(\rho(j_k+1))=\Pop{1}(w)$,
    $\rho{\restriction}_{[j_k,i_{k+1}]}$ is either a loop or a return,
    and $\rho{\restriction}_{[j_k,i_{k+1}]}$ does not visit the stack of
    $\rho(j_k)$ between its initial configuration and its final
    configuration for all $0\leq k < m$.
  \end{enumerate}
\end{lemma}
\begin{proof}
  If $s\prefixeq\rho$, then we set $m:=0$ and we are done. 

  Otherwise, we proceed by induction on the length of $\rho$. 
  There is a minimal position $j_0+1$ such that
  $s\notprefixeq\rho(j_0+1)$. By assumption on $s$,
  \mbox{$\rho(j_0+1)\neq\Pop{2}(s)$}. Thus,
  $\TOP{2}(\rho(j_0))=w$ and  
  $\TOP{2}(\rho(j_0+1))=\Pop{1}(w)$. 
  Now, let $i_1>j_0$ be minimal such that $s\prefixeq\rho(i_1)$.
  Concerning the stack at $i_1$ there are the following possibilities. 
  \begin{enumerate}
  \item 
    If $\rho(i_1)=\Pop{2}(\rho(j_0))$
    then
    $\rho{\restriction}_{[j_0,i_1]}$ is a return.
  \item 
    Otherwise, 
    the stacks of $\rho(j_0)$ and $\rho(i_1)$ coincide whence
    $\rho{\restriction}_{[j_0,i_1]}$ is a loop (note that between $j_0$
    and $i_1$ the stack $\Pop{2}(\rho(j_0))$ is never visited 
    because $i_1$ is minimal, we are not in the first case, and by
    assumption 
    $\lvert s \rvert \leq \lvert \rho(i_1)$).
  \end{enumerate}
  $\rho{\restriction}_{[i_1,n]}$ is shorter than $\rho$. Thus, it
  decomposes by induction hypothesis and 
  the lemma follows immediately.  
\end{proof}

This lemma gives rise to the following extension of the prefix
replacement for prefixes $s,u$ where $s$ and $u$
share similar loops and returns. 
\begin{definition}\label{Def:StackreplacementinMilestones}
  Let $s$ be some stack and $\rho$ be a run of some $2$-PS
  $\mathcal{N}$  such that 
  $s\prefixeq \rho(0)$, $s\prefixeq\rho(\length(\rho))$ and
  $\lvert s \rvert \leq \lvert \rho(i) \rvert$ for all
  $i\in\domain(\rho)$. 
  Let $u$ be 
  some stack such that \mbox{$\TOP{1}(u)=\TOP{1}(s)$}, 
  $\LoopFunc{1}(u) = \LoopFunc{1}(s)$ and
  $\ReturnFunc{1}(u)=\ReturnFunc{1}(s)$. 

  Let $0=i_0 \leq j_0 < i_1 \leq j_1 < \dots <
  i_{m-1} \leq  j_{m-1} < i_m \leq j_m=\length(\rho)$ be the sequence
  corresponding to $\rho$
  in the sense of the previous lemma. 
  We set $(q_k,s_k):=\rho(j_k)$ and $(q'_k,s'_k):=\rho(i_{k+1})$. 
  By definition,
  $\rho{\restriction}_{[j_k,i_{k+1}]}$ is a loop or a return 
  from $(q_k,s_k)$ to $(q_k', s_k')$ and
  $\TOP{2}(s_k)=\TOP{2}(s)$ and $s\prefixeq s_k$. 
  Thus, $\TOP{2}(s_k[s/u])=\TOP{2}(u)$. Since
  \mbox{$\ReturnFunc{1}(u)=\ReturnFunc{1}(s)$} and
  \mbox{$\LoopFunc{1}(u) = \LoopFunc{1}(s)$}, there is a run from
  $(q_k,s_k[s/u])$ to $(q'_{k+1}, s'_{k+1}[s/u])$. 
  We set $\rho_k$ to be the length-lexicographically shortest run from
  $(q_k,s_k[s/u])$ to $(q'_{k+1}, s'_{k+1}[s/u])$.\footnote{We assume
    $\Delta$ to be equipped with some fixed but arbitrary linear order.}  

  Then we define the run $\rho[s/u]$ by 
  \begin{align*}
    \rho[s/u]:=\rho{\restriction}_{[i_0,j_0]}[s/u] \circ \rho_0
    \circ \rho{\restriction}_{[i_1,j_1]}[s/u] \circ \rho_1 \circ \dots
    \circ \rho_{m-1} \circ \rho{\restriction}_{[i_{m},j_{m}]}[s/u].  
  \end{align*}
\end{definition}
\begin{remark}
  Note that $\rho[s/u]$ is a well-defined run from $\rho(0)[s/u]$ to
  $\rho(\length(\rho))[s/u]$. 
\end{remark}
Next, we turn to the analysis of loops and returns. In this part, we show
that for each $2$-PS there is a function relating the height of the
topmost word of a stack with a bound on the length of the shortest
loops and returns of this stack. 

First, we characterise loops and returns in terms of loops
and returns of smaller stacks. The proof is completely
analogous to the proof of Lemma
\ref{Lem:prefixedstartstopdecomposition}.
Every loop  
or return of some stack $s$ decomposes into parts that are prefixed by
$s$ and parts that are returns or loops of stacks with topmost word
$\Pop{1}(\TOP{2}(s))$. This observation gives rise to the observation
that the numbers of loops and of returns of a stack $s$ only depend on
its topmost symbol and the number of loops and returns of
$\Pop{1}(s)$. 

\begin{lemma} \label{Lem:UniqueSequencesLoopsReturns}
  Let $\rho$ be some return from some stack $s:w$ to $s$ of
  length $n$. Then there is a 
  unique sequence 
  $0=i_0 \leq j_0 < i_1 \leq j_1 < \dots < i_{m-1} \leq  j_{m-1} <
  i_m \leq j_m=n$ such that the following holds.
  \begin{enumerate}
  \item  $s:w\prefixeq \rho{\restriction}_{[i_k,j_{k}]}$ for all
    $0\leq k \leq m$ and
  \item $\rho{\restriction}_{[j_{k}+1, i_{k+1}]}$ is a return from some
    stack $s':\Pop{1}(w)$ to $s'$. 
  \end{enumerate}
  Let $\lambda$ be some loop of length $n$ starting and ending
  in some stack $s:w$. Then there is a 
  unique sequence 
  $0=i_0 \leq j_0 < i_1 \leq j_1 < \dots < i_{m-1} \leq  j_{m-1} <
  i_m \leq j_m=n$ such that the following holds.
  \begin{enumerate}
  \item  $s:w\prefixeq \lambda{\restriction}_{[i_k,j_{k}]}$ for all
    $0\leq k \leq m$ and
  \item $\lambda{\restriction}_{[j_k+1, i_{k+1}]}$ is either a 
    return starting in
    some stack $s':\Pop{1}(w)$ 
    or it is a loop of $\Pop{1}(s:w)$ followed by a
    $\Push{\TOP{1}(w)}$ transition. In this case $\lambda(j_k+1)$ is
    the first  and $\lambda(i_{k+1}-1)$ is the last occurrence of
    $s: \Pop{1}(w)$ in $\lambda$. 
  \end{enumerate}
\end{lemma}
\begin{remark}
  Note that property (2) of the characterisation of returns implies
  that
  $\rho{\restriction}_{[j_{k}, i_{k+1}]}$ is a return from some
  stack $s':w$ to $s'$. 
  Similarly, property (2) of the characterisation of loops implies
  that
  $\lambda{\restriction}_{[j_{k}, i_{k+1}]}$ is a return from some
  stack $s':w$ to $s'$ or a loop of $s$ containing all occurrences of
  $\Pop{1}(s)$ in the run $\lambda$. In particular, there is at most
  one $k$ such that 
  $\lambda{\restriction}_{[j_{k}, i_{k+1}]}$ is a loop. 
\end{remark}
This lemma gives rise to the inductive computation of 
the number of loops and returns of a stack.

\begin{proposition} \label{Prop:InductiveComputability}
  Given a  $2$-PS $\mathcal{S}$,
  $\LoopFunc{k}(w)$ and $\ReturnFunc{k}(w)$
  only depend on  
  $\TOP{1}(w)$, $\Pop{1}(\TOP{2}(w))$, $\LoopFunc{k}(\Pop{1}(w))$,
  and 
  $\ReturnFunc{k}(\Pop{1}(w))$.
  Moreover, 
  $\HighLoopFunc{k}(w)$ only depends on $\TOP{1}(w)$ and 
  $\ReturnFunc{k}(\Pop{1}(w))$.
\end{proposition}
\begin{proof}[sketch]
  We only consider the return case and prove the following claim.
  Let $s:w$ and $s':w'$ be stacks such that $s$ and $s'$ are nonempty,
  $\TOP{1}(w)=\TOP{1}(w')$,  and $\ReturnFunc{k}(\Pop{1}(w)) =
  \ReturnFunc{k}(\Pop{1}(w'))$. We have to show that the number of
  returns of $s:w$ and of 
  $s':w'$ coincide, i.e., $\ReturnFunc{k}(w) = \ReturnFunc{k}(w')$.

  For reasons of simplicity, we assume that
  $R(q,q'):=\ReturnFunc{k}(\Pop{1}(w))(q,q')<k$  for each $q,q'\in Q$
  and we fix an enumeration 
  $\rho^{q,q'}_1, \rho^{q,q'}_2, \dots, \rho^{q,q'}_{R(q,q')}$ of the
  returns from \mbox{$(q,s:\Pop{1}(w))$ to $(q', s)$.}
  Similarly, let 
  $\hat\rho^{q,q'}_1, \hat\rho^{q,q'}_2, \dots,
  \hat\rho^{q,q'}_{R(q,q')}$ be an enumeration of the returns from
  $(q, s':\Pop{1}(w'))$ to $(q', s')$. 

  For each return from $(q_1, s:w)$ to $(q_2, s)$ we construct a
  return from $(q_1, s':w')$ to $(q_2, s')$ as follows. 
  Let $\rho$ be such a return and let 
  \mbox{$0=i_0 \leq j_0 < i_1 \leq j_1 < \dots 
  < i_{m-1} \leq  j_{m-1} < i_m \leq j_m$} be the unique sequence
  according to Lemma \ref{Lem:UniqueSequencesLoopsReturns}. 
  Now, $\rho_k:=\rho{\restriction}_{[i_k,j_k]}$ is prefixed by
  $s:w$ and ends in topmost word $w$. Hence,  there is a run
  $\hat\rho_k:=\rho_k[s:w/s':w']$ ending in topmost word $w'$.  
  Furthermore, $\pi_k:=\rho{\restriction}_{[j_k, i_{k+1}]}$ is a
  $\Pop{1}$ transition followed by a return $\sigma_k$. 
  Since $\sigma_k$ starts with topmost word $\Pop{1}(w)$, it is
  equivalent to some $\rho^{q,q'}_l$ in the sense that it 
  performs the same transitions as $\rho^{q,q'}_l$. Now let
  $\hat\sigma_k$ be the return that copies the transitions of
  $\hat\rho^{q,q'}_l$ and starts in $(q,\Pop{1}(\hat\rho_k))$. 
  Since $\TOP{1}(w)=\TOP{1}(w')$ we can define a run $\hat\pi_k$ that
  starts in the final configuration of $\hat\rho_k$ performs the same
  first transition as $\pi_k$ and then agrees with $\hat\sigma_k$. 

  It is straightforward to prove that 
  $\hat\rho:=\hat\rho_0 \circ\hat\pi_0 \circ \hat\rho_1 \circ \dots
  \circ \hat\rho_{m-1} \circ \hat\pi_{m-1} \circ \hat\rho_m$ is a
  return from $(q_1, s':w')$ to $(q_2, s')$. 
  Furthermore, this construction transforms distinct returns from
  $(q_1, s:w)$ to $(q_2, s)$ into distinct returns from $(q_1, s':w')$
  to $(q_2, s')$. Hence, there are at least as many returns from 
  $(q_1, s':w')$  to $(q_2, s')$ as from $(q_1, s:w)$ to $(q_2,
  s)$. 

  Reversing the roles of $s:w$ and $s':w'$ we obtain the reverse
  result and we conclude that
  $\ReturnFunc{k}(w)(q_1,q_2) = \ReturnFunc{k}(w')(q_1,q_2)$. 
\end{proof}

\begin{proposition} \label{Prop:FuncBoundLoopLengthLemma}  
  There is an algorithm that,   on input some $2$-PS
  $\mathcal{S}$ and a natural number $k$,  computes  a
  function
  $\FuncBoundLoopLength{\mathcal{S}}{k}:\N \to\N$ 
  with  the
  following properties. 
  \begin{enumerate}
  \item 
    For all stacks $s$, all $q_1,q_2\in Q$ and for 
    $i:=\LoopFunc{k}(s)(q_1,q_2)$, the
    length-lexicographically shortest loops
    $\lambda_1,\dots, \lambda_i$ from
    $(q_1,s)$ to $(q_2,s)$ satisfy 
    \mbox{$ \length(\lambda_j)\leq
      \FuncBoundLoopLength{\mathcal{S}}{k}(\lvert 
      \TOP{2}(s)\rvert)$}
    for all $1\leq j \leq i$.
  \item If there is a loop $\lambda$ from $(q_1,s)$ to $(q_2, s)$ with
    $\length(\lambda) > \FuncBoundLoopLength{\mathcal{S}}{k}(\lvert
    \TOP{2}(s)\rvert)$, then there are $k$ loops from $(q_1,s)$ to
    $(q_2, \Pop{2}(s))$ of length at most 
    $\FuncBoundLoopLength{\mathcal{S}}{k}(\lvert \TOP{2}(s)\rvert)$.  
  \end{enumerate}
  
  Analogously, there are functions 
  $\FuncBoundReturnLength{\mathcal{S}}{k}:\N \to\N$ and
  $\FuncBoundHighLoopLength{\mathcal{S}}{k}:\N \to\N$
  (computable from
  $\mathcal{S}$) that satisfy the
  same assertions for the set of returns and high loops,
  respectively. 
\end{proposition}
\begin{proof}[sketch]
  Again, we only sketch the proof for the case of returns. 
  
  The previous proof showed that for stacks $s:w$ and $s':w'$ with
  $\TOP{1}(w)=\TOP{1}(w')$ and $\ReturnFunc{k}(w)=\ReturnFunc{k}(w')$ 
  the returns of $s:w$ and of $s':w'$ are closely connected via the
  replacement $[s:w/s':w']$. A return from $s:w$ to
  $s$ decomposes into parts prefixed by $s:w$ and parts that are returns
  of stacks with topmost word $\Pop{1}(w)$. It is a straightforward
  observation that the $k$ shortest returns from $s:w$ to $s$ only
  contain returns from stacks with topmost word $\Pop{1}(w)$ among the
  $k$ shortest of such returns. Furthermore, the decomposition of
  returns of $s:w$ and $s':w'$ agree on their $s:w$-prefixed and
  $s':w'$-prefixed parts in the sense that they perform the same
  transitions. Now, let 
  $R(q_1, q_2):=\ReturnFunc{k}(w)(q_1,q_2)$ and let 
  $\rho_1, \rho_2, \dots, \rho_{R(q_1, q_2)}$ be the $R(q_1, q_2)$
  shortest returns from $(q_1, s:w)$ to $(q_2, s)$. 
  Let $m(\rho_i)$ denote the number of positions in $\rho_i$ that are
  $s:w$ prefixed and let $n(\rho_i)$ denote the number of returns from
  stacks with topmost word $\Pop{1}(w)$ occurring in $\rho_i$. 
  Let $m$ be the maximum over all $m(\rho_i)$ and $n$ the maximum over
  all $n(\rho_i)$. 
  Analogously to the proof of the previous proposition, we can
  construct
  $\ReturnFunc{k}(w')(q_1, q_2)$
  many  
  returns from $(q_1, s':w')$ to $(q_2, s')$ that have at most $m$
  positions that are $s':w'$ prefixed and that contain at most $n$
  many returns of stacks with topmost word $\Pop{1}(w')$. 
  Thus, if we have already defined
  $\FuncBoundReturnLength{\mathcal{S}}{k}(\lvert \Pop{1}(w')\rvert)$,
  then the shortest $\ReturnFunc{k}(w')(q_1,q_2)$ many returns from 
  $(q_1, s':w')$ to $(q_2, s')$ have length at most 
  $m + n \cdot \FuncBoundReturnLength{\mathcal{S}}{k}(\lvert
  \Pop{1}(w')\rvert)$. 
  
  Due to the previous lemma, we can compute 
  a maximal finite sequence of words $w_1, w_2, \dots, w_k$ such that
  for each pair $w_i, w_j$, $\TOP{1}(w_i)\neq\TOP{1}(w_j)$ or
  $\ReturnFunc{k}(w_i) \neq 
  \ReturnFunc{k}(w_j)$.

  Repeating this  construction for each of these words,
  we obtain numbers $m_i$ and $n_i$ for each $1\leq i \leq k$. 
  Set $m_{\max}:=\max\{m_i: 1\leq i \leq k\}$ and 
  $n_{\max}:=\max\{n_i: 1\leq i \leq k\}$.
  Since for any word $w$ there is some $i\leq k$ such that
  $\TOP{1}(w)=\TOP{1}(w_i)$ and $\ReturnFunc{k}(w)=\ReturnFunc{k}(w_i)$
  the shortest $\ReturnFunc{k}(w)(q_1,q_2)$ many returns from 
  $(q_1, s':w)$ to $(q_2, s')$ have length at most 
  \begin{align*}
    m_i + n_i \cdot 
    \FuncBoundReturnLength{\mathcal{S}}{k}(\lvert\Pop{1}(w')\rvert)\leq
    m_{\max} + n_{\max} \cdot 
    \FuncBoundReturnLength{\mathcal{S}}{k}(\lvert\Pop{1}(w')\rvert).
  \end{align*}
  Thus, setting 
  $\FuncBoundReturnLength{\mathcal{S}}{k}(0):=0$ and 
  $\FuncBoundReturnLength{\mathcal{S}}{k}(n+1):=
  m_{\max} + n_{\max} \cdot \FuncBoundReturnLength{\mathcal{S}}{k}(n)$
  settles the claim. 
\end{proof}
\begin{remark}
  Note that we do not know any bound on $m_{\max}$ and $n_{\max}$ in
  terms of  $\lvert \mathcal{S} \rvert$. 
\end{remark}

We conclude this section with two corollaries that allow us to shrink
runs between certain configurations. These corollaries are important
in the proof that $2$-NPT have the dynamic-small-witness
property. 

\begin{corollary} \label{Cor:GlobalBoundRun}
  Let $\mathcal{S}$ be some level $2$ pushdown
  system. Furthermore, let  $(q,s)$ be some configuration and 
  $\rho_1, \dots, \rho_n$ be pairwise distinct runs from the initial
  configuration to $(q,s)$. 
  There is a run $\hat\rho_1$ from the initial configuration to
  $(q,s)$ such that the following holds.
  \begin{enumerate}
  \item $\hat\rho_1\neq\rho_i$ for $2\leq i \leq n$ and
  \item $\length(\hat\rho_1)\leq 2 \cdot \lvert s\rvert \cdot \height(s) 
    ( 1+ \FuncBoundLoopLength{\mathcal{S}}{n}(\height(s)))$.
  \end{enumerate}
\end{corollary}
\begin{proof}
  If $\length(\rho_1)\leq
  2 \cdot \lvert s \rvert \cdot \height(s) 
  ( 1+ \FuncBoundLoopLength{\mathcal{S}}{n}(\height(s)))$, set
  $\hat\rho_1:=\rho_1$ and we are done. 
  Otherwise,  
   Remark 
  \ref{Rem:NumberMilestones} and Lemma
  \ref{Lemma:Carayol05Milestones} implies that $\rho_1$ decomposes as
  \begin{align*}
    \rho_1 = \lambda_0 \circ \op_1 \circ \lambda_1
    \circ \dots \circ \lambda_{m-1} \circ \op_m \circ \lambda_{m}    
  \end{align*}
  where every $\lambda_i$ is a loop, every $\op_i$ is a run of
  length $1$, and $m\leq 2\cdot \lvert s \rvert \cdot
  \height(s)$. Proposition \ref{Prop:FuncBoundLoopLengthLemma} implies
  the following: 
  If $\length(\lambda_i) >
  \FuncBoundLoopLength{\mathcal{S}}{n}(\height(s)))$, then there are
  $n$ loops from $\lambda_i(0)$ to $\lambda_i(\length(\lambda_i))$ of length 
  at most $\FuncBoundLoopLength{\mathcal{S}}{n}(\height(s)))$. 
  At least one of these can be plugged into the position of
  $\lambda_i$ such that the resulting run does not coincide with any
  of the $\rho_2, \rho_3, \dots, \rho_n$. In other words, there is
  some loop $\lambda_i'$ of length at most
  $\FuncBoundLoopLength{\mathcal{S}}{n}(\height(s)))$ such that
  \begin{align*}
    \hat\rho_1:=
    \lambda_0 \circ \op_1 \circ \lambda_1
    \circ \dots \circ \op_i \circ \lambda_i'\circ\op_{i+1}\circ
    \lambda_{i+1} \circ\dots \circ  \lambda_{m-1} \circ \op_m \circ
    \lambda_{m}     
  \end{align*}
  is a run to $(q,s)$ distinct from $\rho_2, \rho_3 \dots, \rho_n$ and
  shorter than $\rho_1$.
  Iterated replacement of large loops results in a run 
  $\rho_1'$ with the desired properties. 
\end{proof}

We state a second corollary that is quite similar to the previous
one but deals with runs of a different form. These runs become
important in Section \ref{sec:EquivalenceonTuples}.

\begin{corollary}\label{Cor:GlobalBoundRun2}
  Let $\hat\rho_1, \hat\rho_2, \dots \hat\rho_n$ be runs from the
  initial configuration to some configuration $(q,s)$. 
  Furthermore, let $w$ be some word and 
  $\rho_1, \rho_2, \dots \rho_n$ be runs from $(q,s)$ to $(q',s:w)$
  that do not visit proper substacks of $s$. 
  If $\hat\rho_1 \circ \rho_1, \hat\rho_2\circ\rho_2, \dots,
  \hat\rho_n\circ\rho_n$ are pairwise distinct, then there is a run
  $\rho_1'$ from $(q,s)$ to $(q',s:w)$ that satisfies the following.
  \begin{enumerate}
  \item $\rho_1'$does not visit a proper
    substack of $s$,
  \item   
    $\length(\rho_1') \leq 2\cdot \height(s:w)  \cdot
    ( 1+ \FuncBoundLoopLength{\mathcal{S}}{n}(\height(s:w)))$, and
  \item $\hat\rho_1\circ\rho_1'$ is distinct from each
    $\hat\rho_i\circ\rho_i$ for $2\leq i \leq n$. 
  \end{enumerate}
\end{corollary}
\begin{proof}
  It is straightforward to see that $\rho_1$ decomposes as 
  \begin{align*}
    \rho_1= \lambda_0 \circ \op_1 \circ \lambda_1
    \circ \dots \circ \lambda_{m-1} \circ \op_m \circ \lambda_{m}    
  \end{align*}
  where every $\lambda_i$ is a loop and every $\op_i$ is a run of
  length $1$ such that $m\leq 2\cdot \height(s:w)$. 
  Now apply the construction from the previous proof again. 
\end{proof}

\section{Relevant Ancestors}
\label{sec:RelevantAncestors}

This section aims at identifying those ancestors of a run $\rho$ in a
$2$-\HONPT $\mathfrak{N}$ that characterise its
$\FO{k}$-type. We show that  
only finitely many ancestors of a certain kind already fix the
$\FO{k}$-type of $\rho$.
We call these finitely many ancestors the
\emph{relevant $2^k$-ancestors} of $\rho$.

\subsection{Definition and Basic Observations}

Before we formally define relevant ancestors, we introduce
some sloppy notation concerning runs.
We apply functions defined on stacks to
configurations. For example if $c=(q,s)$ we write $\lvert c \rvert$
for $\lvert s \rvert$ and $\Pop{2}(c)$ for $\Pop{2}(s)$. 
We further abuse this notation by application
of functions defined on stacks to runs, meaning that we
apply the function to the last stack occurring in a run. For instance,
we write 
$\TOP{2}(\rho)$ for $\TOP{2}(s)$ and $\lvert \rho \rvert$ for $\lvert s
\rvert$  if \mbox{$\rho(\length(\rho))=(q,s)$}.
In the same sense one has to understand equations like
$\rho(i)=\Pop{1}(s)$. This equation says that
$\rho(i)=(q,\Pop{1}(s))$ for some $q\in 
Q$. Keep in mind that $\lvert \rho \rvert$ denotes the width of  
the last stack of $\rho$ and not its length $\length(\rho)$. 
Also recall that we write $\rho \preceq \rho'$ if the run $\rho$ is an
initial segment of the run $\rho'$. 

\begin{definition}
  Let $\mathfrak{N}$ be some $n$-NPT.
  Define the relation  $\plusedge \subseteq
  \mathfrak{N}\times \mathfrak{N}$ by
  \begin{align*}
  \rho \plusedge \rho'  \text{ if }
  \rho\prec \rho', \lvert \rho \rvert = \lvert \rho' \rvert -1, \text{
    and }
    \lvert \pi \rvert > \lvert \rho \rvert \text{ for all } \rho\prec
    \pi \prec \rho'. 
  \end{align*}
  We define the set of 
  \emph{relevant $l$-ancestors of $\rho$} by induction on $l$.
  The set of relevant $0$-ancestors of $\rho$ is
  $\RelAnc{0}{\rho}:=\{\rho\}$. 
  Set 
  \begin{align*}
    \RelAnc{l+1}{\rho}&:= \RelAnc{l}{\rho} 
    \cup \left\{\pi\in \mathfrak{N}: \exists \pi'\in \RelAnc{l}{\rho}\
      \pi\trans{} \pi' \text{ or } \pi\jumpedge \pi' 
      \text{ or } \pi\plusedge \pi' \right\}.
  \end{align*}
  For $\bar\rho=(\rho_1, \rho_2, \dots, \rho_n)$, we write
  $\RelAnc{l}{\bar\rho}:=\bigcup\limits_{i=1}^n \RelAnc{l}{\rho_i}$.
\end{definition}

\begin{remark} \label{rem:CharacterisePlusOneRelation}
  For $1$-NPT, the relation $\plusedge$ can be
  characterised as follows: for runs $\rho, \rho'$, we have
  $\rho\plusedge \rho'$ if and only if 
  $\rho'=\rho\circ \pi$ for  a run $\pi$ starting at some
  $w_\rho$ and 
  ending in $w_\rho a$, the first operation of $\pi$ is a $\Push{a}$
  and $\pi$ visits $w_\rho$ only in its initial configuration. 

  For $2$-NPT, there is a similar characterisation:
  we have
  $\rho\plusedge \rho'$ if and only if 
  $\rho'=\rho\circ \pi$ for  some run $\pi$ starting at some stack $s_\rho$ and
  ending in some stack $s_\rho:w$, the first operation of $\pi$ is a clone
  and $\pi$ visits $s_\rho$ only in its initial configuration. 

  The motivation for the definition is the following. 
  If there are elements $\rho,\rho'\in \mathfrak{N}$ such that 
  $\rho'\preceq \rho$
  and there 
  is a path in $\mathfrak{N}$ of length at most $l$ that witnesses
  that $\rho'$ is an 
  ancestor of $\rho$, then we want that $\rho'\in \RelAnc{l}{\rho}$. 
  The relation $\plusedge$ is  tailored
  towards this idea. Assume that there are runs $\rho_1 \prec \rho_2
  \trans{\delta} \rho_3$ with $\delta$ performing some $\Pop{n}$
  operation such 
  that $\rho_2 \trans{\delta} \rho_3 
  \jumpleftedge \rho_1$. This path of length $2$ witnesses that
  $\rho_1$ is an ancestor of $\rho_2$. By definition, one  sees
  immediately that $\rho_1 \plusedge \rho_2$
  whence
  $\rho_1\in\RelAnc{1}{\rho_2}$. In this sense,
  $\plusedge$ relates the ancestor $\rho_1$ with
  $\rho_2$ if  the distance between $\rho_1$ and  $\rho_2$ in
  $\mathfrak{N}$ may be small because of a descendant of $\rho_2$ that
  is a jump edge successor of $\rho_1$. 

  In the following, it may be helpful to think of a relevant
  $l$-ancestor $\rho'$
  of a run $\rho$ as  an ancestor of $\rho$ that may have a path of
  length up to $l$  witnessing that $\rho'$ is an ancestor of $\rho$. 
\end{remark}
From the definitions, we obtain immediately the following lemmas. 
\begin{lemma}
  For each run $\rho'$ there is at most one run $\rho$ such that
  $\rho\plusedge \rho'$. $\rho$ is the maximal
  ancestor of $\rho'$ satisfying $\rho' = \Pop{n}(\rho)$. 
\end{lemma}
\begin{lemma}
  Let $\rho$ and $\rho'$ be runs such that $\rho \jumpedge
  \rho'$. Let $\hat\rho$ be the predecessor of $\rho'$, i.e.,
  $\hat\rho$ is the unique element such that $\hat\rho \trans{}
  \rho'$. Then $\rho\plusedge\hat\rho$. 
\end{lemma}

\begin{lemma} \label{LemmaRelAnclocalIso}
  If $\rho,\rho'\in \mathfrak{N}$  such that $\rho\trans{} \rho'$ or
  $\rho\jumpedge \rho'$, then $\rho\in \RelAnc{1}{\rho'}$. 
\end{lemma}

\begin{lemma}
  For all $l\in\N$ and $\rho\in \mathfrak{N}$, 
  $\lvert \RelAnc{l}{\rho} \rvert \leq 4^l$ and
  $\RelAnc{l}{\rho}$ is linearly ordered by $\preceq$. 
\end{lemma}



We now investigate the structure of relevant ancestors and 
the possible intersections of relevant ancestors of different runs.
First, we characterise the minimal element of $\RelAnc{l}{\rho}$. 

\begin{lemma}\label{LemmaMinRelAnc}
  Let $\mathfrak{N}$ be a $n$-NPT.
  Let $\rho_l\in \RelAnc{l}{\rho}$ be minimal with respect to $\preceq$. Then
  \begin{align*}
    \text{either }&\lvert \rho_l \rvert = 1 \text{ and }
    \lvert \rho \rvert \leq l \\  
    \text{or } &\rho_l = \Pop{n}^l(\rho) \text{ and } 
    \lvert \rho_l \rvert < \lvert \rho' \rvert \text{ for all }
    \rho'\in \RelAnc{l}{\rho}\setminus\{\rho_l\}.      
  \end{align*}
\end{lemma}
\begin{remark}
  Recall that $\lvert \rho \rvert \leq l$ implies that
  $\Pop{n}^{l}(\rho)$ is undefined. 
\end{remark}
\begin{proof}
  The proof is by induction on $l$. For $l=0$, there is nothing to
  show because \mbox{$\rho_0 = \rho = \Pop{n}^0(\rho)$.}
  Now assume that the statement is true for some $l$. 
  \begin{enumerate}
  \item   Consider the case $\lvert \rho \rvert \leq l+1$. Then
    $\rho_l$ satisfies $\lvert 
    \rho_l \rvert = 1$. If $\rho_l$ has no
    predecessor it is also the minimal element of $\RelAnc{l+1}{\rho}$
    and we are done. Otherwise, there is a maximal ancestor  $\hat
    \rho\prec \rho_l$ such that $\lvert \hat\rho \rvert = 1$. Since
    $\hat \rho\trans{} \rho_l$ or $\hat \rho\jumpedge \rho_l$,
    $\hat \rho\in
    \RelAnc{l+1}{\rho}\setminus\RelAnc{l}{\rho}$. Furthermore, 
     no  ancestor of $\hat 
    \rho$ can be contained in $\RelAnc{l+1}{\rho}$. 
    Heading for a contradiction, assume that there is some element
    $\tilde \rho\prec\hat\rho$ such 
    that \mbox{$\tilde\rho\in\RelAnc{l+1}{\rho}$}. Then there is some
    $\tilde\rho'\in\RelAnc{l}{\rho}$ with $\tilde\rho \prec \hat\rho \prec
    \tilde\rho'$ such that $\tilde\rho\jumpedge\tilde\rho'$ 
    or $\tilde\rho\plusedge\tilde\rho'$.
    But this leads to the contradiction
    $1\leq \lvert \tilde\rho' \rvert < \lvert \hat\rho \rvert =1$. 
    Thus, the minimal element of $\RelAnc{l+1}{\rho}$ is
    $\rho_{l+1}=\hat \rho$.
  \item 
    Now assume that $\lvert \rho \rvert > l+1$. 
    Let $\hat \rho$ be the maximal ancestor of $\rho_l$ such that $\lvert
    \hat \rho \rvert +1 = \lvert \rho_l \rvert$. Then 
    $\hat \rho \plusedge \rho_l$
    or $\hat\rho \trans{} \rho_1$, whence 
    $\hat \rho\in \RelAnc{l+1}{\rho}$.  
    We have to show that $\hat \rho$ is the minimal element of
    $\RelAnc{l+1}{\rho}$ and
    that there is no other element of width $\lvert \hat\rho \rvert$ in
    $\RelAnc{l+1}{\rho}$.  
    For the second part, assume that there is some 
    $\rho'\in \RelAnc{l+1}{\rho}$ with
    $\lvert \rho' \rvert = \lvert \hat \rho \rvert$. Then $\rho'$ has to
    be connected via 
    $\trans{}, \plusedge$, or
    $\jumpedge$  to some element 
    $\rho''\in \RelAnc{l}{\rho}$. By
    definition of these relations $\lvert \rho'' \rvert \leq \lvert \rho' \rvert
    +1$. By induction hypothesis, this implies $\rho''=\rho_l$. But then it is
    immediately clear that $ \rho' = \hat \rho$ by definition. 
    
    Similar to the previous case, the minimality of $\hat\rho$ in
    $\RelAnc{l+1}{\rho}$ is proved by 
    contradiction.
    Assume that there is some $\rho'\prec \hat \rho$
    such that $\rho'\in \RelAnc{l+1}{\rho}$. Then there is some 
    \mbox{$\hat \rho\prec \rho_l \preceq \rho''\in 
      \RelAnc{l}{\rho}$} such that 
    \mbox{$\rho'\plusedge \rho''$} or 
    $\rho' \jumpedge \rho''$. By the definition of  
    $\jumpedge$
    and $\plusedge$, we obtain
    $\lvert  \rho'' \rvert \leq \lvert \hat \rho \rvert$. But 
    this contradicts  
    $\lvert \rho'' \rvert \geq \lvert \rho_l \rvert > \lvert \hat \rho \rvert$. 
    Thus, we conclude that $\hat\rho$ is the minimal element of
    $\RelAnc{l+1}{\rho}$, i.e., $\hat \rho=\rho_{l+1}$. \qed
  \end{enumerate}
\end{proof}

The previous lemma shows that the width of stacks among the relevant
ancestors cannot decrease too much. Furthermore, the width cannot grow
too much. 
\begin{corollary} \label{CorDistRelAnc}
  Let $\pi,\rho\in \mathfrak{N}$ such that
  $\pi\in\RelAnc{l}{\rho}$. Then $\big\lvert 
  \lvert \rho \rvert - \lvert \pi \rvert \big\lvert \leq l$.
\end{corollary}
\begin{proof}
  From the previous lemma, we know that the minimal width of the last
  stack of an element in $\RelAnc{l}{\rho}$ is $\lvert \rho \rvert -l$. 
  We prove by induction that the maximal width is 
  $\lvert \rho \rvert +l$. The case $l=0$ is trivially true. 
  Assume that $\lvert \pi \vert \leq \lvert \rho \rvert +l-1$ for all
  $\pi\in\RelAnc{l-1}{\rho}$.
  Let $\hat\pi\in \RelAnc{l}{\rho}\setminus\RelAnc{l-1}{\rho}$. Then
  there is an 
  $\pi\in\RelAnc{l-1}{\rho}$ such that $\hat\pi\trans{} \pi$,
  $\hat\pi\jumpedge \pi$, $\hat\pi \plusedge \pi$.
  For the last two cases the width of $\hat\pi$ is smaller than the width
  of $\pi$ whence $\lvert \hat\pi \rvert \leq \lvert \rho \rvert +l-1$. 
  For the first case, recall that all stack operations of an $n$-PS
  alter the width of the stack by at most 
  $1$. Thus, $\lvert \hat\pi \rvert \leq \lvert \pi \rvert +1 \leq
  \lvert \rho \rvert + l$. 
\end{proof}

Next we show a kind of triangle inequality for the relevant
ancestor relation. If $\rho_2$ is a relevant ancestor of $\rho_1$ then
all relevant ancestors of $\rho_1$ that are ancestors of $\rho_2$ are
relevant ancestors of $\rho_2$.

\begin{lemma} \label{LemmaRelAncComposition}
  Let $\rho_1, \rho_2 \in \mathfrak{N}$ and let $l_1,l_2\in\N$. 
  If $\rho_1\in \RelAnc{l_1}{\rho_2}$, then
  \begin{align*}
    &\RelAnc{l_2}{\rho_1}\subseteq \RelAnc{l_1+l_2}{\rho_2} \text{ and}\\
    &\RelAnc{l_2}{\rho_2}\cap \{\pi: \pi\preceq \rho_1\} \subseteq
    \RelAnc{l_1+l_2}{\rho_1}.
  \end{align*}
\end{lemma}
\begin{proof}
  The first inclusion holds by induction on the definition of relevant
  ancestors.  

  For the second claim, we proceed by induction on $l_2$. For
  $l_2=0$ the claim holds because  $\RelAnc{0}{\rho_2}=\{\rho_2\}$ and
  $\rho_1\preceq\rho_2$ imply that $\RelAnc{0}{\rho_2}\cap\{\pi:
  \pi\preceq \rho_1\}\neq\emptyset$ if and only if $\rho_1=\rho_2$ whence
  $\{\rho_2\}\in\RelAnc{0}{\rho_1}$.  
  For the induction step assume that 
  \begin{align*}
    \RelAnc{l_2-1}{\rho_2}\cap \{\pi: \pi\preceq \rho_1\} \subseteq
    \RelAnc{l_1+l_2-1}{\rho_1}.
  \end{align*}
  Furthermore, assume that 
  $\pi\in \RelAnc{l_2}{\rho_2} \cap \{\pi:\pi \preceq \rho_1\}$. 
  We show that $\pi\in\RelAnc{l_1+l_2}{\rho_1}$.
  By
  definition there is 
  some $\pi\prec \hat\pi$ such that $\hat\pi\in
  \RelAnc{l_2-1}{\rho_2}$ and $\pi \in  \RelAnc{1}{\hat\pi}$.
  We distinguish the following cases. 
  \begin{itemize}
  \item Assume that $\hat\pi \preceq \rho_1$. By
    hypothesis,  
    \mbox{$\hat\pi\in \RelAnc{l_1+l_2-1}{\rho_1}$} whence
    \mbox{$\pi\in \RelAnc{l_1+l_2}{\rho_1}$.}
  \item  Assume that
    $\pi\prec \rho_1\prec \hat\pi \prec \rho_2$. 
    This implies that $\pi\jumpedge\hat\pi$ or
    $\pi\plusedge \hat\pi$ whence 
    \mbox{$\lvert \pi\rvert = \lvert \hat\pi \rvert -j<
    \lvert \rho_1 \rvert$} for some
    $j\in \{0,1\}$. 
    From Corollary \ref{CorDistRelAnc}, we know that
    \begin{align*}
      \left\rvert \lvert \hat\pi \rvert - \lvert \rho_2 \rvert \right\rvert
      \leq l_2-1 \text{ and}
      \left\rvert \lvert \rho_1 \rvert - \lvert \rho_2 \rvert \right\rvert
      \leq l_1.
    \end{align*}
    This implies that $\lvert \rho_1 \rvert - \lvert \pi \rvert \leq
    l_1+l_2$. 
    By definition of $\jumpedge$ and
    $\plusedge$, there cannot be any 
    element $\pi'$ with \mbox{$\pi\prec \pi' \prec \hat\pi$} and
    $\lvert \pi' \rvert = \lvert \pi \rvert$.
    Thus, $\pi$ is the maximal predecessor of $\rho_1$ with
    $\pi=\Pop{2}^{\lvert \rho_1 \rvert - \lvert \pi \rvert}(\rho_1)$. 
    Application of Lemma \ref{LemmaMinRelAnc} shows that 
    $\pi$ is the minimal
    element of 
    $\RelAnc{\lvert \rho_1 \rvert - \lvert \pi \rvert}{\rho_1}$. 
    Hence, $\pi \in \RelAnc{ \lvert \rho_1 \rvert - \lvert \pi \rvert}{\rho_1}
      \subseteq \RelAnc{l_1+l_2}{\rho_1}$.      \qed
  \end{itemize} 
\end{proof}

\begin{corollary} \label{CorRelAncDistBound}
  If $\rho\in \RelAnc{l}{\rho_1} \cap \RelAnc{l}{\rho_2}$ then
  \mbox{$\RelAnc{l}{\rho_1} \cap \{\pi: \pi\preceq \rho\} \subseteq
  \RelAnc{3l}{\rho_2}$.}  
\end{corollary}
\begin{proof}
  By the previous lemma,  $\rho\in \RelAnc{l}{\rho_1}$ implies
  $\RelAnc{l}{\rho_1} \cap \{\pi: \pi\preceq \rho\} \subseteq
  \RelAnc{2l}{\rho}$.  
  Using the lemma again, $\rho\in \RelAnc{l}{\rho_2}$ implies 
  $\RelAnc{2l}{\rho}\subseteq \RelAnc{3l}{\rho_2}$.
\end{proof}

The previous corollary shows that
if the relevant $l$-ancestors of two elements $\rho_1$ and $\rho_2$
intersect at some point $\rho$, then all relevant $l$-ancestors of
$\rho_1$ that are ancestors of $\rho$ are contained in the relevant
$3l$-ancestors of $\rho_2$.  Later, we use the contraposition of
this result in order to
prove that relevant ancestors of certain runs are disjoint sets. 

The following proposition describes how $\RelAnc{l}{\rho}$
embeds into the full $2$-\HONPT $\mathfrak{N}$. 
Successive relevant ancestors of some run $\rho$ are connected
by a single edge or by a $\plusedge$-edge.
We will see that this proposition allows us, given an arbitrary run
$\rho$, to explicitly construct  a relevant ancestor set 
isomorphic to $\RelAnc{l}{\rho}$ that consists of small runs.  

\begin{proposition} \label{Prop:NextRelAnc}
  Let $\rho_1\prec \rho_2 \prec \rho$ such that 
  $\rho_1,\rho_2\in  \RelAnc{l}{\rho}$.  
  If $\pi\notin \RelAnc{l}{\rho}$ for all $\rho_1\prec \pi\prec \rho_2$, then
  $\rho_1\trans{} \rho_2$ or $\rho_1\plusedge \rho_2$.
\end{proposition}
\begin{proof}
  Assume that   $\rho_1 \not{\trans{}} \rho_2$. 
  Consider the set
  \mbox{$M:=\{\pi\in\RelAnc{l}{\rho}:\rho_1\plusedge \pi\}.$}
  $M$ is nonempty because there is some $\pi\in\RelAnc{l-1}{\rho}$
  such that 
  either $\rho_1\plusedge \pi$ (whence $\pi\in M)$ or
  $\rho_1\jumpedge \pi$ (whence the predecessor $\hat\pi$ of
  $\pi$ satisfies $\hat\pi\in M$). 
  Let $\hat\rho\in M$ be minimal. It suffices to show that
  $\hat\rho=\rho_2$. For this purpose, we show that  
  $\pi\notin\RelAnc{l}{\rho}$ for all $\rho_1\prec \pi \prec
  \hat\rho$. Since $\hat\rho\in\RelAnc{l}{\rho}$, this implies that
  $\hat\rho=\rho_2$.   
  We start with two general observations.
  \begin{enumerate}
  \item \label{ImmerGroesser}
    For all $\rho_1\prec \pi\prec \hat\rho$,  
    $\lvert \pi \rvert \geq \lvert \hat\rho \rvert$
    due to the definition of $\rho_1\plusedge
    \hat\rho$. Furthermore, due to the minimality of $\hat\rho$ in
    $M$, for all  
    $\rho_1\prec \pi \prec \hat\rho $ with $\pi\in\RelAnc{l}{\rho}$, $\lvert
    \pi\rvert>\lvert \hat\rho \rvert$ (otherwise $\pi\in M$
    which contradicts the minimality of $\hat\rho$). 
  \item \label{KeineHaken}
    Note that there cannot exist $\rho_1\prec \pi \prec \hat\rho \prec \hat\pi$
    with $\pi \jumpedge \hat\pi$ or $\pi \plusedge
    \hat\pi$ because $\lvert \pi \rvert \geq \lvert \hat\rho \rvert$.
  \end{enumerate}
  Heading for a contradiction, assume that there is some 
  $\rho_1\prec \pi \prec \hat\rho$ such that \mbox{$\pi\in\RelAnc{l}{\rho}$.}

  Due to observation \ref{KeineHaken}, there is
  a chain $\pi_0:=\pi, \pi_1, \dots, \pi_{n-1}, \pi_n:=\hat\rho$ such
  that 
  for each $0\leq i<n$ there is $*\in\{\trans{}, \jumpedge,
  \plusedge\}$ such that $\pi_i\mathrel{*} \pi_{i+1}$
  and $\pi_i\in \RelAnc{l-i}{\rho}$. By assumption, 
  $n\neq 0$, whence $\hat\rho\in\RelAnc{l-1}{\rho}$. 
  Due to observation \ref{ImmerGroesser}, we have 
  $\lvert \rho_1 \rvert < \lvert \hat\rho \rvert < \lvert \pi
  \rvert$. Since each 
  stack operation alters the width 
  of the stack by at most $1$, we conclude that the set 
  \begin{align*}
    M':=\left\{\pi': \rho_1\prec \pi' \prec \hat\rho, \lvert \hat\rho \rvert =
      \lvert \pi'  \rvert\right\}     
  \end{align*}
  is nonempty because on the path from $\rho_1$ to $\pi$ there
  occurs at least one run with final stack of width $\lvert \hat\rho
  \rvert$. But the maximal element $\pi'\in M'$ satisfies  
  $\rho_1 \plusedge \pi'\trans{} \hat\rho$ or 
  $\rho_1 \plusedge \pi' \jumpedge
  \hat\rho$. 
  Since
  $\hat\rho\in\RelAnc{l-1}{\rho}$, this would imply $\pi' \in M$ which
  contradicts the minimality of $\hat\rho$ in $M$. Thus, no 
  $\rho_1\prec \pi 
  \prec \hat\rho$ 
  with $\pi\in\RelAnc{l}{\rho}$ can exist. 

  We conclude that
  $\pi\notin\RelAnc{l}{\rho}$ for all
  $\rho_1\prec\pi\prec \hat\rho$ and
  $\rho_1\plusedge \hat\rho=\rho_2$. 
\end{proof}

In the final part of this section, we consider relevant ancestors of
two different runs $\rho$ and $\rho'$. Since we aim at a construction
of small runs $\hat\rho$ and $\hat\rho'$ such that the relevant
ancestors of $\rho$ and $\rho'$ are isomorphic to the relevant
ancestors of $\hat\rho$ and $\hat\rho'$, we need to know how sets of
relevant ancestors touch each other. Every isomorphism from the
relevant ancestors of $\rho$ and $\rho'$ to those of $\hat\rho$ and
$\hat\rho'$ has to preserve edges between a relevant ancestor of
$\rho$ and another one of $\rho'$. 

The positions where the relevant $l$-ancestors of $\rho$ and
$\hat\rho$ touch can  be identified by looking at the intersection of
their relevant \mbox{$(l+1)$-ancestors}. This is shown in the next
lemma. 
For $A$ and $B$ subsets of some $n$-NPT $\mathfrak{N}$ and $\rho$ some
run of $\mathfrak{N}$, we 
say $A$ and $B$ \emph{touch after $\rho$} if there are runs $\rho\prec
\rho_A,\rho\prec\rho_B$ such  
that $\rho_A\in A$, $\rho_B\in B$ and either $\rho_A=\rho_B$ or
$\rho_A * \rho_B$ for some $*\in \{\trans{}, \invtrans{}, \jumpedge,
\jumpleftedge, \plusedge, \plusleftedge\}$. In this case we say $A$ and $B$ touch at 
$(\rho_A,\rho_B)$. 
In the 
following, we reduce the question whether $l$-ancestors of two
elements touch after 
some $\rho$ to the
question whether the $(l+1)$-ancestors of these elements intersect after
$\rho$. 

\begin{lemma}
  If $\rho_1,\rho_2$ are runs such that $\RelAnc{l_1}{\rho_1}$ and 
  $\RelAnc{l_2}{\rho_2}$
  touch after some $\rho_0$, then  
  $\RelAnc{l_1+1}{\rho_1} \cap \RelAnc{l_2+1}{\rho_2}
  \cap \{\pi: \rho_0\preceq \pi\} \neq \emptyset$.
\end{lemma}
\begin{proof}
  Let $\rho_0$ be some run, 
  $\rho_0\prec \hat\rho_1\in \RelAnc{l_1}{\rho_1}$, and $\rho_0\prec
  \hat\rho_2\in 
  \RelAnc{l_2}{\rho_2}$ such that the 
  pair $(\hat\rho_1,\hat\rho_2)$ is minimal and $\RelAnc{l_1}{\rho_1}$
  and $\RelAnc{l_2}{\rho_2}$ 
  touch at $(\hat\rho_1,\hat\rho_2)$. 
  Then one of the following holds.
  \begin{enumerate}
  \item $\hat\rho_1=\hat\rho_2$: there is nothing to prove because
    $\hat\rho_1 \in
    \RelAnc{l_1}{\rho_1}\cap\RelAnc{l_2}{\rho_2}\cap\{\pi:\rho_0\preceq\pi\}$.
  \item $\hat\rho_1\trans{} \hat\rho_2$ or
    $\hat\rho_1\jumpedge \hat\rho_2$ 
    or $\hat\rho_1 \plusedge \hat\rho_2:$ this implies that
    $\hat\rho_1\in \RelAnc{l_2+1}{\rho_2}\cap\RelAnc{l_1}{\rho_1}$.
  \item 
    $\hat\rho_1\invtrans{} \hat\rho_2$ or $\hat\rho_1\jumpleftedge \hat\rho_2$
    or $\hat\rho_1 \plusleftedge \hat\rho_2:$ this
    implies that  
    $\hat\rho_2\in \RelAnc{l_1+1}{\rho_1}\cap\RelAnc{l_2}{\rho_2}$. \qed
  \end{enumerate}
\end{proof}

\begin{corollary} \label{CorTouch}
  If $\rho$ and $\rho'$ are runs such that $\RelAnc{l_1}{\rho}$ and
  $\RelAnc{l_2}{\rho'}$ touch after some run $\rho_0$ then there 
  exists some $\rho_0\prec \rho_1 \in \RelAnc{l_1+1}{\rho} \cap
  \RelAnc{l_2+1}{\rho'}$ such that
  \begin{align*}
    \RelAnc{l_1+1}{\rho}\cap \{x: x \preceq \rho_1\} \subseteq 
    \RelAnc{l_2 + 2 l_1 + 3}{\rho'}.    
  \end{align*}
\end{corollary}
\begin{proof}
  Use the previous lemma and Lemma \ref{LemmaRelAncComposition}.
\end{proof}

\subsection{A Family of Equivalence Relations on Words and Stacks}
\label{subsec:EquivalenceonWords}

Next, we  define a family of
equivalences on words that is useful for constructing runs with
similar relevant ancestors. 
The
basic idea is to classify 
words according to the  $\FO{k}$-type of the word model associated
to the word $w$ expanded by information about certain runs between
prefixes of $w$. This additional information 
describes
\begin{enumerate}
\item  the number of possible loops and returns with certain initial
  and final state of each prefix $v\leq
  w$, and 
\item the number of runs from $(q,w)$ to $(q',v)$ for each  
  prefix $v\leq w$ and all pairs $q,q'$ of states.
\end{enumerate}
It turns out that this equivalence has the following property:
if $w$ and $w'$ are equivalent and $\rho$ is a run starting in $(q,w)$
and ending in $(q',w:v)$, 
then there is a run from $(q,w')$
to $(q,w':v')$ such that the loops and returns of $v$ and $v'$
agree. This is important because runs of this 
kind connect consecutive elements of relevant ancestor sets
(cf. Proposition \ref{Prop:NextRelAnc}). 

In order to copy relevant ancestors, we want to apply this kind of
transfer 
property iteratively. For instance, we want to take a run from
$(q_1, w_1)$ via $(q_2,w_1:w_2)$ to $(q_3,w_1:w_2:w_3)$ and translate
it into some run from $(q_1,w_1')$ via $(q_2, w_1':w_2')$ to 
$(q_3,w_1':w_2':w_3')$ such that the loops and returns of $w_3$ and $w_3'$
agree. 
Analogously, we want to take a run creating $n$ new words and
transfer it to a new run starting in another word and creating $n$
words such that the last words agree on their loops and returns. 
If we can do this, then we can transfer the whole set of relevant
ancestors from some run to another one. This allows us to construct isomorphic
relevant ancestors that consist only of short runs. 

The family of equivalence
relations that we define have the following transfer property. Words 
that are equivalent with respect to the
$n$-th relation allow a transfer of runs creating $n$ new words. 
The idea of the definition is as follows. Assume that we have already
defined the 
$(i-1)$-st equivalence relation. 
We  take the
word model of some word $w$ and annotate each prefix of the
word by its equivalence class with respect to the $(i-1)$-st
relation. Then we define two words to be equivalent with respect to
the $i$-th relation if the $\FO{k}$-types of their enriched word models
agree.  

These equivalence relations and the transfer properties that they induce
are an important tool in the next section. There we 
apply them to an arbitrary  set of relevant ancestors in order to
obtain isomorphic 
copies of the substructure induced by these ancestors. 

For the rest of this section, we fix some $2$-PS $\mathcal{N}$.
For $w$ some word, we use $w_{-n}$ as an abbreviation for
$\Pop{1}^n(w)$.   

\begin{definition}
  For each word $w\in\Sigma^*$,  
  we define a family of expanded word models
  $\Lin{n}{k}{z}(w)$ by induction on $n$. Note that for $n=0$ the
  structure will be 
  independent of the parameter $k$ but for greater $n$ this parameter
  influences the expansion of the structure.  
  Let
  $\Lin{0}{k}{z}(w)$  be the expanded word model 
  \begin{align*}
    \Lin{0}{k}{z}(w):=(\{0,1,\dots, \lvert w \rvert -1\},
    \mathrm{succ}, 
    (P_\sigma)_{\sigma\in\Sigma}, (S^j_{q,q'})_{(q,q')\in Q^2, j\leq z}, 
    (R_j)_{j\in J}, (L_j)_{j\in J}, 
    (H_j)_{j\in J})    
  \end{align*}
  such
  that for $0\leq i <\lvert w \rvert$ the following holds.
  \begin{itemize}
  \item $\mathrm{succ}$ and $P_\sigma$ form the standard word model of
    $w$ in reversed order, i.e., $\mathrm{succ}$ is the usual successor
    relation on the domain and 
    $i\in P_{\sigma}$ if and only if
    $\TOP{1}(w_{-i})=\sigma$, 
  \item $i\in S^j_{q,q'}$, if there are $j$ pairwise distinct runs 
    $\rho_1, \dots, \rho_j$ such that each run starts in $(q,w)$ and ends in 
    $(q',w_{-i})$.
  \item The predicates $R_j$ encode the function
    $i\mapsto\ReturnFunc{z}(w_{-i})$ (cf. Definition \ref{Def:ReturnFunc}). 
  \item The predicates $L_j$ encode the function
    $i\mapsto\LoopFunc{z}(w_{-i})$.  
  \item The predicates $H_j$ encode the function
    $i\mapsto\HighLoopFunc{z}(w_{-i})$. 
  \end{itemize}

  Now, set $\Typ{0}{k}{z}(w):=\FO{k}[\Lin{0}{k}{z}(w)]$, the quantifier
  rank $k$ 
  theory  of $\Lin{0}{k}{z}(w)$. 

  Inductively, we define $\Lin{n+1}{k}{z}(w)$ to be the expansion of
  $\Lin{n}{k}{z}(w)$ by predicates describing $\Typ{n}{k}{z}(v)$ for
  each prefix $v\leq w$. 
  More formally,  fix a maximal list $\theta_1, \theta_2, \dots,
  \theta_m$ of pairwise distinct $\FO{k}$-types that are realised 
  by some $\Lin{n}{k}{z}(w)$. 
  We define predicates $T_1, T_2, \dots, T_m$ such that 
  $i\in T_j$ if
  $\Typ{n}{k}{z}\left(w_{-i}\right)=\theta_j$ 
  for all $0\leq i\leq n$. Now, let $\Lin{n+1}{k}{z}(w)$ be the
  expansion of $\Lin{n}{k}{z}(w)$ by the 
  predicates $T_1, T_2, \dots, T_m$. 
  We conclude the inductive definition by setting
  $\Typ{n+1}{k}{z}(w):= \FO{k}[\Lin{n+1}{k}{z}(w)]$. 
\end{definition}
\begin{remark}
  Each element of $\Lin{n}{k}{z}(w)$ corresponds to a prefix of
  $w$. In this sense, we write $v\in S^j_{q,q'}$ for
  some prefix $v\leq w$ if $v=w_{-i}$ and 
  $\Lin{n}{k}{z}(w) \models i\in S^j_{q,q'}$.  

  It is an important observation that $\Lin{n}{k}{z}(w)$ is  a finite 
  successor structure with finitely many colours. Thus, for fixed
  $n,k,z\in\N$, $\Typ{n}{k}{z}$ has finite image.
\end{remark}
For our application, $k$ and $z$ can be chosen to be
some fixed large numbers, depending on the quantifier rank of the
formula we are interested in. 
Furthermore, it will turn out that the conditions on $k$ and $z$
coincide whence we will 
assume that $k=z$. This is due to the fact that both parameters are 
counting thresholds in some sense: $z$ is the threshold for counting
the existence of loops and returns, while $k$ can be seen as the
threshold for distinguishing different prefixes of $w$ which 
have the same atomic type. Thus, we
identify $k$ and $z$ in  
the following definition of the equivalence relation induced by
$\Typ{n}{k}{z}$. 
\begin{definition}
  For words $w,w'\in\Sigma^*$, we write $ w\wordequiv{n}{z} w'$ if
  $\Typ{n}{z}{z}(w)=\Typ{n}{z}{z}(w')$. 
\end{definition}

As a first step, we want to show that $\wordequiv{n}{z}$ is a right
congruence. We prepare the proof of this fact in the following lemma. 
\begin{lemma}
  Let $n\in\N$, $z\geq 2$ and $\mathcal{N}$ be some $2$-PS. 
  Let $w$ be some word and $\sigma\in\Sigma$ some letter. 
  For each $0\leq i <\lvert w \rvert$, the atomic types of $i$ and of
  $0$ in
  $\Lin{n}{z}{z}(w)$ determine the atomic type of $i+1$ in
  $\Lin{n}{z}{z}(w\sigma)$.
\end{lemma}
\begin{proof}
  Recall that $i\in\Lin{n}{z}{z}(w)$ represents 
  $w_{-i}$ and $i+1\in\Lin{n}{z}{z}(w\sigma)$ represents
  $w\sigma_{-(i+1)}$. 
  Since $w_{-i}=w\sigma_{-(i+1)}$, it follows directly that the two
  elements agree on  
  $(P_\sigma)_{\sigma\in\Sigma}$, $(R_j)_{j\in J}$, $(L_j)_{j\in J}$, and
  $(H_j)_{j\in J}$ and that 
  $w_{-i}\wordequiv{n-1}{z} w\sigma_{-(i+1)}$ (recall that the
  elements of $\Lin{n}{z}{z}(w)$ are coloured by $\wordequiv{n-1}{z}$-types). 

  We claim that the function $\ReturnFunc{z}(w)$ and the set
  \begin{align*}
    \{(j,q,q')\in\N\times Q\times Q: j\leq z, \Lin{n}{z}{z}(w)\models i\in
    S^j_{q,q'}\}
  \end{align*}
  determine whether $\Lin{n}{z}{z}(w\sigma)\models (i+1)\in
  S^j_{q,q'}$.  Recall that the predicates $S^j_{q,q'}$ in
  $\Lin{n}{z}{z}(w)$ encode at each position $i$ the number of runs
  from $(q,w)$ to $(q' , w_{-i})$. 
  We now want to determine the number of runs from $(q,w\sigma)$ to
  $(q', w\sigma_{-(i+1)}) = (q', w_{-i})$. 

  It is clear that such a run starts with a high loop from
  $(q,w\sigma)$ to some $(\hat q,w\sigma)$. Then it performs some
  transition of the form $(\hat q, \sigma,  \hat q', \Pop{1})$
  and then it continues with a run from $(\hat q', w)$ to 
  $(q', w_{-i})$. 
  
  In order to determine whether 
  $\Lin{n}{z}{z}(w\sigma)\models (i+1)\in S^j_{q,q'}$, we have to
  count whether $j$ runs of this form exist. 
  To this end,
  we define the  numbers
  \begin{align*}
    &k_{(\hat q, \hat q')}:=\HighLoopFunc{z}(w\sigma)(q,\hat q),\\
    &j_{(\hat q, \hat q')}:=\lvert \{ (\hat q, \sigma, 
     \hat q', \Pop{1})\in \Delta\}\rvert,
    \text{ and}  \\
    &i_{(\hat q, \hat q')}:=\max\{k: w_{-i}\in S^{k}_{(\hat q',q')}\}
  \end{align*}
  for each pair $\bar q=(\hat q, \hat q') \in Q^2$. 
  It follows directly that there are
  $\sum\limits_{\bar q\in Q^2} i_{\bar q}j_{\bar q} k_{\bar q}$
  many such runs up to threshold $z$. 
  Note that $j_{\bar q}$ only depends on the pushdown system. 
  Due to Proposition
  \ref{Prop:InductiveComputability},
  $\HighLoopFunc{z}(w\sigma)$ is determined by $\sigma$ and
  $\ReturnFunc{z}(w)$.
  Thus, $k_{\bar q}$ is determined by the atomic type of $0$ in
  $\Lin{n}{z}{z}(w)$. 
  $i_{\bar q}$ only depends on the atomic type of $i$ in 
  $\Lin{n}{z}{z}(w)$. These observations  complete the proof. 
\end{proof}

\begin{corollary}
  Let $n,z\in\N$ such that $z\geq 2$.  
  Let $w_1$ and $w_2$ be words such that 
  \mbox{$w_1 \wordequiv{n}{z} w_2$.}
  Any strategy of Duplicator in the $z$ round Ehrenfeucht-\Fraisse game on
  $\Lin{n}{z}{z}(w_1)$ and $\Lin{n}{z}{z}(w_2)$ translates directly
  into a strategy of Duplicator
  in the $z$ round Ehrenfeucht-\Fraisse game on
  $\Lin{n}{z}{z}(w_1\sigma){\restriction}_{[1,\lvert w_1\rvert]}$ and
  $\Lin{n}{z}{z}(w_2\sigma){\restriction}_{[1,\lvert w_2\rvert]}$.
\end{corollary}
\begin{proof}
  It suffices to note that the existence of Duplicator's strategy
  implies that the atomic types of $0$ in 
  $\Lin{n}{z}{z}(w_1)$ and $\Lin{n}{z}{z}(w_2)$ agree. Hence, the
  previous lemma applies. 
  Thus, if the atomic type of $i\in\Lin{n}{z}{z}(w_1)$ and
  $j\in\Lin{n}{z}{z}(w_2)$ agree, then the atomic types of
  $i+1\in\Lin{n}{z}{z}(w_1\sigma)$ and 
  $j+1\in\Lin{n}{z}{z}(w_2\sigma)$ agree. 
  Hence, we can obviously
  translate Duplicator's strategy on 
  $\Lin{n}{z}{z}(w_1)$ and $\Lin{n}{z}{z}(w_2)$ into a strategy on 
  $\Lin{n}{z}{z}(w_1\sigma){\restriction}_{[1,\lvert w_1\rvert]}$ and
  $\Lin{n}{z}{z}(w_2\sigma){\restriction}_{[1,\lvert w_2\rvert]}$.
\end{proof}

The previous corollary is the main ingredient for the following
lemma. It states that $\wordequiv{n}{z}$ is a right congruence. 

\begin{lemma} \label{LemmaTypeConcatenation}
  For $z\geq 2$, $\wordequiv{n}{z}$ is a right congruence, i.e., 
  if $\Typ{n}{z}{z}(w_1) = \Typ{n}{z}{z}(w_2)$ for some $z\geq 2$, then
  $\Typ{n}{z}{z}(w_1w) = \Typ{n}{z}{z}(w_2w)$ for all
  $w\in\Sigma^*$. 
\end{lemma}
\begin{proof}
  It is sufficient to prove the claim for $w=\sigma\in\Sigma$.
  The lemma then follows by induction on $\lvert w \rvert$. 
  First observe that 
  \begin{align*}
    &\LoopFunc{z}(w_1\sigma) = \LoopFunc{z}(w_2\sigma),\\ 
    &\HighLoopFunc{z}(w_1\sigma) = \HighLoopFunc{z}(w_2\sigma),\text{ and}\\
    &\ReturnFunc{z}(w_1\sigma) = \ReturnFunc{z}(w_2\sigma),     
  \end{align*}
  because these values are determined by the values of the corresponding
  functions at  $w_1$ and $w_2$ (cf. Proposition
  \ref{Prop:InductiveComputability}). These functions agree on $w_1$ and
  $w_2$ because the first elements 
  of $\Lin{n}{z}{z}(w_1)$ and $\Lin{n}{z}{z}(w_2)$ are
  $\FO{z}$ definable ($z\geq 2$!).

  For $i\in\{1,2\}$, $\Lin{n}{z}{z}(w_i\sigma) \models 0\in S^j_{(q,q')}$ if
  and only if there are $j$ loops from $(q,w_i\sigma)$ to 
  $(q', w_i\sigma)$ (at position $0$ 
  the runs counted by the $S^j_{(q,q')}$ coincide with loops). Since
  $\LoopFunc{z}(w_1\sigma)=\LoopFunc{z}(w_2\sigma)$, we
  conclude that the atomic types of the first elements of 
  $\Lin{0}{z}{z}(w_1\sigma)$ and of 
  $\Lin{0}{z}{z}(w_2\sigma)$ coincide. 
  
  Due to the previous corollary, we know that Duplicator has a
  strategy in the $z$ round Ehrenfeucht-\Fraisse game on
  $\Lin{n}{z}{z}(w_1\sigma){\restriction}_{[1,\lvert w_1\rvert]}$ and
  $\Lin{n}{z}{z}(w_2\sigma){\restriction}_{[1,\lvert w_2\rvert]}$.

  Standard composition arguments for Ehrenfeucht-\Fraisse games on
  word structures directly imply that
  $\Lin{0}{z}{z}(w_1\sigma) \equiv_z \Lin{0}{z}{z}(w_2\sigma)$. But
  this directly implies that the atomic types of $w_1\sigma$ in 
  $\Lin{1}{z}{z}(w_1\sigma)$ and of $w_2\sigma$ in 
  $\Lin{1}{z}{z}(w_2\sigma)$ coincide. If $n\geq 1$, we can apply the
  same standard argument and obtain that
  $\Lin{1}{z}{z}(w_1\sigma) \equiv_z \Lin{1}{z}{z}(w_2\sigma)$.
  By induction one concludes that
  $\Lin{n}{z}{z}(w_1\sigma) \equiv_z \Lin{n}{z}{z}(w_2\sigma)$, i.e., 
  $w_1\sigma \wordequiv{n}{z} w_2\sigma$. 
\end{proof}

In terms of stack operations, the previous lemma can be seen as a
compatibility result of $\wordequiv{n}{z}$ with the $\Push{\sigma}$
operation. Next, we lift the equivalences from words to $2$-stacks in such
a way that the new equivalence relations are compatible with all stack
operations. 
We compare the
stacks word-wise beginning with the topmost word, then the word below
the topmost one, etc. up to some threshold $m$. 

\begin{definition}
  Let $s,s'$ be stacks. We write $s \stackequiv{n}{z}{m} s'$ if
  either
  \begin{align*}
    \lvert s \rvert >m, \lvert s' \rvert > m, \text{ and }
    \TOP{2}\left(\Pop{2}^i(s)\right)\wordequiv{n}{z}
    \TOP{2}\left(\Pop{2}^i(s')\right)   \text{ for all }0\leq i \leq
    m, \text{ or}\\
    l:=\lvert s \rvert = \lvert s' \rvert \leq m, \text{ and }
    \TOP{2}\left(\Pop{2}^i(s)\right)\wordequiv{n}{z}
    \TOP{2}\left(\Pop{2}^i(s')\right)   \text{ for all }0\leq i < l
    .
  \end{align*}
\end{definition}

\begin{proposition} \label{Prop:CompatibilityStackOpTypeq}
  Let $z\geq 2$ and let $s_1, s_2$ be stacks such that
  $s_1 \stackequiv{n}{z}{m} s_2$. 
  Then $\Push{\sigma}(s_1)\stackequiv{n}{z}{m} \Push{\sigma}(s_2)$,
  $\Pop{1}(s_1)\stackequiv{n-1}{z}{m}\Pop{1}(s_2)$,
  $\Clone{2}(s_1) \stackequiv{n}{z}{m+1} \Clone{2}(s_2)$,
  and $\Pop{2}(s_1) \stackequiv{n}{z}{m-1} \Pop{2}(s_2)$. 
\end{proposition}
\begin{proof}
  Assume that $\op=\Pop{1}$.
  Quantifier rank $z$ suffices to define the second element of a
  word structure. Hence, $w\wordequiv{n}{z} w'$ implies that 
  $\Typ{n-1}{z}{z}(w_{-1}) =   \Typ{n-1}{z}{z}(w'_{-1})$. But this implies
  \mbox{$w_{-1} \wordequiv{n-1}{z} w'_{-1}$.}

  For 
  $\op=\Push{\sigma}$ we use Lemma \ref{LemmaTypeConcatenation}. 
  For $\Clone{2}$ and $\Pop{2}$, the claim is trivial.  
\end{proof}

The previous proposition shows that the equivalence relations on
stacks are compatible with the stack operations. 
Recall that successive relevant ancestors of a given run $\rho$ are
runs $\rho_1 \prec \rho_2 \preceq \rho$ such that 
$\rho_1\trans{}\rho_2$ 
or $\rho_1\plusedge \rho_2$ 
(cf. Proposition \ref{Prop:NextRelAnc}). 
In the next section, we are concerned with the construction of a
short run $\hat \rho$ such that its relevant ancestors are isomorphic
to those of $\rho$. A necessary condition for a run $\hat\rho$ to be
short is that it only passes small stacks. We
construct $\hat\rho$ using the following construction. Let 
$\rho_0 \prec \rho_1 \prec \rho_2 \ldots \prec \rho$ be the set of
relevant ancestors of $\rho$. 
We then first define a run $\hat\rho_0$ that ends in some small stack
that is equivalent to the last stack of $\rho_0$. Then, we iterate the
following construction. If $\rho_{i+1}$ extends $\rho_i$ by a single
transition, then we define $\hat\rho_{i+1}$ to be the extension of
$\hat\rho_{i}$ by the same transition. Due to the previous proposition
this preserves equivalence of the topmost stacks of $\rho_{i}$ and
$\hat\rho_{i}$. Otherwise, $\rho_{i+1}$ extends $\rho_i$ by some run
that creates a new word $w_{i+1}$ on top of the last stack of
$\rho_i$. Then we want to construct a short run that creates a new
word $w_{i+1}'$ on top of the last stack of $\hat\rho_i$ such that
$w_{i+1}$ and $w_{i+1}'$ are equivalent and $w_{i+1}'$ is small. Then
we define
$\hat\rho_{i+1}$ to be $\hat\rho_i$ extended by this run. 

Finally, this procedure defines a run $\hat\rho$ that corresponds
to $\rho$ in the sense that the relevant ancestors of the two runs are
isomorphic but $\hat\rho$ is a short run.

In the following, we prepare this construction. We 
show that for any run $\rho_0$ there is a run $\hat\rho_0$ that ends
in some small stack that is equivalent to the last stack of $\rho_0$. 
This is done in Lemma \ref{Lem:RunsendInShorStacks}. 
Furthermore, we show that for runs $\rho_i$ and $\hat\rho_i$ that end
in equivalent stacks, any run that extends the last stack of $\rho_i$
by some word $w$ can be transferred into a run that extends $\hat\rho_i$ by
some small word that is equivalent to $w$. This is shown in
Proposition \ref{Prop:ConstructOneStep}. 

The proofs of Lemma \ref{Lem:RunsendInShorStacks} and 
Proposition \ref{Prop:ConstructOneStep}
are based on the property that 
prefixes of equivalent stacks share the same number of loops and
returns for each pair of initial and final states. 
Recall that our analysis of generalised milestones showed that the
existence of loops with certain initial and final states has a crucial
influence on the question whether runs between certain stacks exist. 
We first define functions that are used to define what a small stack
is. Afterwards, we show that any run to some stack can be replaced by
some run to a short equivalent stack. 
\begin{definition} \label{Def:BoundTopWordHeightandWidht}
  Let $\mathcal{N}=(Q,\Sigma,\Delta,q_0)$ be a $2$-PS. 
  Set $\FuncBoundTopWord(n,z) = \lvert Q \rvert \cdot 
  \lvert \nicefrac{\Sigma^*}{\wordequiv{n}{z}} \rvert+1$,
  where $\lvert \nicefrac{\Sigma^*}{\wordequiv{n}{z}} \rvert$ is the
  number of equivalence classes of $\wordequiv{n}{z}$. 
  Furthermore, set
  $\ConstBoundHeightWord :=\lvert
  \nicefrac{\Sigma^*}{\wordequiv{0}{2}} \rvert 
  \cdot \lvert Q \rvert^2$ and    
  $\FuncBoundWidthWord(n):=\lvert Q \rvert \cdot
  (\lvert\Sigma\rvert+1)^n$. 
\end{definition}

\begin{remark}
  $\FuncBoundWidthWord(n)$ is an upper bound for the number of pairs
  of states and words of length up to $n$. Note that
  $\FuncBoundTopWord, \ConstBoundHeightWord$ and $\FuncBoundWidthWord$
  computably depend on  $\mathcal{N}$.
\end{remark}

\begin{lemma} \label{Lem:RunsendInShorStacks}
  Let $n\in\N$, $z\geq 2$, and let $\rho$ be a run from the initial
  configuration  
  to some configuration $(q,s)$. 
  There is 
  a run $\rho'$ starting in the initial configuration such that 
  \begin{align*}
    &  \lvert \TOP{2}(\rho) \rvert - \FuncBoundTopWord(n,z) \leq
    \lvert \TOP{2}(\rho') \rvert \leq \lvert \TOP{2}(\rho)\rvert, \\ 
    &\height(\rho') \leq \lvert \TOP{2}(\rho')\rvert + 
    \ConstBoundHeightWord,\\
    &\lvert \rho' \rvert \leq \FuncBoundWidthWord(\height(\rho')) \text{
      and}\\
    & \TOP{2}(\rho) \wordequiv{n}{z} \TOP{2}(\rho'). 
  \end{align*}
  Furthermore, if $\lvert\TOP{2}(\rho)\rvert >
  \FuncBoundTopWord(n,z)$, then $\lvert\TOP{2}(\rho')\rvert < \lvert
  \TOP{2}(\rho) \rvert$. 
\end{lemma}
\begin{figure}
  \centering
  $
  \begin{xy}
    \xymatrix@=0.2mm{
      &      &      &  f   &   f \\
      &  b   &   d  &  e   &   e \\
      &  a   &   a  &  a   &   a \\
      s= & \bot & \bot & \bot & \bot
    }
  \end{xy}
  $\hskip 3mm
  $
  \begin{xy}
    \xymatrix@=0.2mm{
      &\textcolor{white}{ f}\\
      &\textcolor{white}{ d}\\
      &\textcolor{white}{ a}\\
      m_1= & \bot &
    }
  \end{xy}
  $\hskip 3mm
  $
  \begin{xy}
    \xymatrix@=0.2mm{
      \textcolor{white}{ f}\\
      &  b   &   d  &   \\
      &  a   &   a  &  a    \\
      m_2= & \bot & \bot & \bot 
    }
  \end{xy}
  $\hskip 3mm
  $
  \begin{xy}
    \xymatrix@=0.2mm{
      &      &      &  \textcolor{white}{ f}   \\
      &  b   &   d  &  e    \\
      &  a   &   a  &  a    \\
      m_3= & \bot & \bot & \bot 
    }
  \end{xy}
  $
  \caption{Illustration for the construction in the first step of the
    proof of Lemma 
    \ref{Lem:RunsendInShorStacks}.} 
  \label{fig:LemmaPumpingTopExample}
\end{figure}
We prove this lemma in three steps.
\begin{enumerate}
\item For each run $\rho$ with long topmost word, we generate a run $\rho'$
  with equivalent but smaller topmost word.
\item For each run $\rho$, we generate a run $\rho'$ 
  such that $\TOP{2}(\rho)=\TOP{2}(\rho')$ and the height of $\rho'$
  is bounded by  $\lvert \TOP{2}(\rho)\rvert+ \ConstBoundHeightWord$.
\item For each run $\rho$, we generate a run $\rho'$ such that 
  $\TOP{2}(\rho)=\TOP{2}(\rho')$, the height of $\rho'$
  is bounded in terms of $\height(\rho)$ and the width of $\rho'$ is
  bounded in terms of its height.
\end{enumerate}

\begin{proof}[of step 1]
  Let $\rho$ be some run with \mbox{$\lvert \TOP{2}(\rho)\rvert >
  \FuncBoundTopWord(n,z)$.} 
  Let \mbox{$(q,s):=\rho(\length(\rho))$} be the final configuration of
  $\rho$. For each $k\leq \FuncBoundTopWord(n,z)$, there is a
  maximal milestone $m_k\in\Milestones(s)$ with
  $\lvert\TOP{2}(m_k)\rvert = k$. 
  Figure \ref{fig:LemmaPumpingTopExample} illustrates this
  definition. 
  Let $w_k:=\TOP{2}(m_k)$ and let $\rho_k
  \preceq \rho$ be the largest initial segment of $\rho$ that ends in
  $m_k$. Note that $m_k\prefixeq m_{k'} \prefixeq s$ for all $k\leq k'
  \leq \FuncBoundTopWord(n,z)$  by the maximality of $m_k$ and
  $m_{k'}$.  
  
  Then there are $i < j \leq \FuncBoundTopWord(n,z)$ such that
  $\TOP{2}(\rho_i) \wordequiv{n}{z} \TOP{2}(\rho_j)$ and the
  final states of $\rho_i$ and $\rho_j$ agree. 
  
  Due to the maximality of $\rho_j$, no substack of $\Pop{2}(m_j)$ is
  visited by $\rho$ after $k:=\length(\rho_j)$. 
  Thus, the run
  $\pi:=(\rho{\restriction}_{[k,\length(\rho)]})[m_j/m_i]$ is
  well-defined   (cf. Definition
  \ref{Def:StackreplacementinMilestones}). Note that
  $\pi$ starts in $(q',m_i)$ for $q'\in Q$ the final
  state of $\rho_i$. Thus, we can set $\hat\rho:=\rho_i\circ \pi$. 
  Since $w_i\wordequiv{n}{z} w_j$ and since 
  $\wordequiv{n}{z}$ is a right congruence, it
  is clear that 
  $\TOP{2}(\hat\rho)\wordequiv{n}{z}\TOP{2}(\rho)$.  
  Since $0< \lvert w_j \rvert - \lvert w_i \rvert <
  \FuncBoundTopWord(n,z)$, it also follows directly that
  \begin{align*}
    \lvert\TOP{2}(\rho)\rvert - \FuncBoundTopWord(n,z) \leq \lvert
    \TOP{2}(\hat\rho) \rvert <  \lvert \TOP{2}(\rho)\rvert. \qed
  \end{align*}
\end{proof}

\begin{figure}
  \centering
  $
  \begin{xy}
    \xymatrix@=0.2mm{
               &      &   d  &  f   &   \\
               &  b   &   d  &  d   &   d \\
               &  a   &   a  &  a   &   a &   a\\
      s=m_2' = & \bot & \bot & \bot & \bot & \bot
    }
  \end{xy}
  $\hskip 3mm
  $
  \begin{xy}
    \xymatrix@=0.2mm{
             &         &   d  &  f   \\
             &  b      &   d  &  d    \\
             &  a      &   a  &  a   \\
      m=m_4= & \bot & \bot & \bot
    }
  \end{xy}
  $\\
  $
  \begin{xy}
    \xymatrix@=0.2mm{
           &      &   d  &        \\
           &  b   &   d  &  d       \\
           &  a   &   a  &  a     \\
      m_3= & \bot & \bot & \bot 
    }
  \end{xy}
  $\hskip 3mm
  $
  \begin{xy}
    \xymatrix@=0.2mm{
         & \textcolor{white}{f}   &      \\
         &  b   &    \\
         &  a   &   a  \\
      m_2= & \bot & \bot 
    }
  \end{xy}
  $\hskip 3mm
  $
  \begin{xy}
    \xymatrix@=0.2mm{
            &      &   d  &  f   &  f \\
            &  b   &   d  &  d   &   d \\
            &  a   &   a  &  a   &   a \\
      m'_4= & \bot & \bot & \bot & \bot
    }
  \end{xy}
  $\hskip 3mm
  $
  \begin{xy}
    \xymatrix@=0.2mm{
            &      &   d  &  f   &   \\
            &  b   &   d  &  d   &   d \\
            &  a   &   a  &  a   &   a \\
      m_3'= & \bot & \bot & \bot & \bot
    }
  \end{xy}
  $\hskip 3mm
  \caption{Illustration for the construction in the second step of
    the proof of Lemma
    \ref{Lem:RunsendInShorStacks}.}
  \label{fig:LemmaPumpingHeightExample}
\end{figure}

\begin{proof}[of step 2]
  The proof is by induction on the
  number of words in 
  the last stack of $\rho$ that have length $h:=\height(\rho)$. 
  Assume that $\rho$ is some run such that
  \begin{align*}
    \height(\rho) > \lvert \TOP{2}(\rho)\rvert +
    \ConstBoundHeightWord.    
  \end{align*}
  In the following, we define several generalised milestones of the final
  stack $s$ of 
  $\rho$. An illustration of  these definitions can be found in 
  Figure \ref{fig:LemmaPumpingHeightExample}.

  Let $m\in \Milestones(s)$ be a milestone of the last stack of
  $\rho$ such that
  $\lvert \TOP{2}(m)\rvert = h$. 
  For each $\lvert \TOP{2}(\rho) \rvert \leq i \leq h$ let
  $m_i\in\Milestones(m)$ be the maximal milestone of $m$ with $\lvert
  \TOP{2}(m_i)\rvert = i$. 
  Let $n_i$ be maximal such that $\rho(n_i)= (q', m_i)$ for some $q'\in
  Q$. 
  Let $m'_i\in \genMilestones(s)\setminus \genMilestones(m)$ be the
  minimal generalised milestone after $m$ such that
  \mbox{$\TOP{2}(m'_i) = \TOP{2}(m_i)$.} 
  Let $n'_i$ be maximal with $\rho(n'_i) = (q', m'_i)$ for some $q'\in
  Q$. 

  There are 
  $\lvert \TOP{2}(\rho) \rvert \leq k < l \leq \height(\rho)$
  satisfying the following conditions.
  \begin{enumerate}
  \item There is a $q\in Q$ such that $\rho(n_k)=(q,m_k)$ and
    $\rho(n_l) = (q, m_l)$.
  \item There is a $q'\in Q$ such that $\rho(n'_k)= (q', m'_k)$
    and $\rho(n'_l) = (q', m'_l)$.
  \item $\TOP{2}(m_k)  \wordequiv{0}{2} \TOP{2}(m_l)$ (note that this
    implies that 
    \mbox{$\LoopFunc{1}(m_k)=\LoopFunc{1}(m_l)$} and
    \mbox{$\ReturnFunc{1}(m_k)= \ReturnFunc{1}(m_l)$}). 
  \end{enumerate}
  By definition, we have $m_l\prefixeq m_l'$. 
  Thus, the run $\pi_1:= (\rho{\restriction}_{[n_l,n'_l]})[m_l/m_k]$ is
  well defined (cf. Definition
  \ref{Def:StackreplacementinMilestones}).
  Note that $\pi_1$ starts in 
  $(q,m_k)$ and ends in $(q',\hat s)$ for $\hat s:=m'_l[m_l/m_k])$.
  Moreover, 
  $\TOP{2}(\pi_1) = \TOP{2}(m'_k) = \TOP{2}(\rho(n'_k))$. 
  Furthermore,
  $\rho{\restriction}_{[n'_k,\length(\rho)]}$ never looks below the
  topmost word of $m'_k$ because $n'_k$ is the maximal node where
  the generalised milestone $m'_k$ is visited. 
  Thus, $(\Pop{2}(m'_k):\bot) \prefixeq
  \rho{\restriction}_{[n'_k,\length(\rho)]}$ whence
  \begin{align*}
    \pi_2:=\rho{\restriction}_{[n'_k, \length(\rho)]}[\Pop{2}(m'_k):\bot/
    \Pop{2}(\hat s):\bot]    
  \end{align*}
  is well defined. It starts in the last stack of $\pi_1$. 
  Now, we define the run 
  $\hat\rho:= \rho{\restriction}_{[0,n_k]} \circ \pi_1 \circ \pi_2$.     
  Either $\height(\hat\rho) < \height(\rho)$ and we are done
  or there are less words of height
  $\height(\rho)$ in the last stack of $\hat\rho$ than in the last
  stack of $\rho$ and we conclude by induction. 
\end{proof}
\begin{proof}[of step 3]
  Assume that $\rho$ is a run with 
  $\length(\rho)> \FuncBoundWidthWord(\height(\rho))$.
  We denote by $n_i$ the maximal position in $\rho$ such that the stack
  at $\rho(n_i)$ is $\Pop{2}^i(\rho)$ for each 
  \mbox{$0\leq i \leq \lvert \rho \rvert$.}
  There are less than
  $\frac{\FuncBoundWidthWord(\height(\rho))}{\lvert Q 
    \rvert}$ many words of length up to $\height(\rho)$. 
  Thus, there are $i < j$  such that
  \begin{enumerate}
  \item there is a word $w$ with 
    $\TOP{2}\left(\Pop{2}^i(\rho)\right) =
      \TOP{2}\left(\Pop{2}^j(\rho)\right)=w$, and 
  \item $\rho(n_i)=\left(q,\Pop{2}^i(\rho)\right)$ and
    $\rho(n_j)=\left(q,\Pop{2}^j(\rho)\right)$ for some $q\in Q$. 
  \end{enumerate}
  Now, let $s_i:=\Pop{2}^{i+1}(\rho)$ and $s_j:=\Pop{2}^{j+1}(\rho)$. 
  There is a unique stack $s$ such that
  \mbox{$\rho(\length(\rho)) = (\hat q, s_i:s)$.} 
  The run $\rho{\restriction}_{[n_i, \length(\rho)]}$ (from $s_i:w$
  to $s_i:s$) never visits $s_i$. 
  Thus,
  \begin{align*}
    \hat\rho_1:=\rho{\restriction}_{[n_i, \length(\rho)]}[s_i:\bot/s_j:\bot]    
  \end{align*}
  is a run from $s_j:w$ to $s_j:s$. 
  The composition $\hat\rho:=\rho{\restriction}_{[0,n_j]} \circ \hat\rho_1$
  satisfies the claim. 
\end{proof}

The previous corollary deals with the reachability of some stack from the
initial configuration. The following proposition is concerned with the
extension of a given stack by just one word. Recall that such a run
corresponds to a $\plusedge$ edge. We first define the
function that is used to bound the size of the new word. Recall that
the equivalence relation $\wordequiv{n}{z}$ depends on the choice of
the fixed $2$-PS $\mathcal{N}=(Q,\Sigma,\Delta, q_0)$.
$\left\lvert \nicefrac{\Sigma^*}{\wordequiv{n}{z}}\right\rvert$ 
denotes the number of equivalence classes of $\wordequiv{n}{z}$.
\begin{definition}\label{Def:BH1Definition}
  Set $\BoundHeightOnestepConstructionSimultanious(a,b,c,d):= 1+b+ a(
  \lvert 
  Q \rvert \lvert \nicefrac{\Sigma^*}{\wordequiv{c}{d}} \rvert)$.
\end{definition}

Before we state the proposition concerning the compatibility of
$\wordequiv{n}{z}$ with $\plusedge$ edges, we
explain its meaning. The 
proposition says that given two equivalent words $w$ and $\hat w$ and a run
$\rho$ from $(q, s:w)$ to $(q', s:w:w')$ that does not pass any substack
of $s:w$, then, for each stack $\hat s:\hat w$, we find a run $\hat
\rho$ from  $(q, \hat s: \hat w)$ to $(q', \hat s: \hat w: \hat w')$
for some short word $\hat w'$ that is equivalent to $w'$. Furthermore, 
this transfer of runs works simultaneously on a tuple of such runs,
i.e., given $m$ runs starting at $s:w$ of the form described above, we
find $m$ corresponding runs starting at $\hat s :\hat w$. 
This simultaneous transfer becomes important when we search 
an isomorphic copy of the relevant ancestors of 
several runs. In this case the simultaneous transfer allows us to copy
the relevant ancestors of a certain run while avoiding  an intersection 
with the relevant ancestors of other given runs.

\begin{proposition} \label{Prop:ConstructOneStep}
  Let 
  $n,z,m\in \N$ such that $n\geq 1$, $z> m$, and $z\geq
  2$. 
  Let \mbox{$c=(q,s:w)$}, \mbox{$\hat c=(q,\hat s:\hat w)$} be
  configurations such that 
  $w \wordequiv{n}{z} \hat w$. 
  Let $\rho_1, \dots, \rho_m$ be pairwise distinct runs such that for
  each $i$, $\lvert \rho_i(j) \rvert > \lvert s:w \rvert$ for all
  $j\geq 1$ and such that $\rho_i$
  starts at $c$ and ends in $(q_i, s:w:w_i)$. 
  Analogously, let $\hat \rho_1, \dots, \hat \rho_{m-1}$ be pairwise
  distinct runs such 
  that each $\hat \rho_i$ starts at $\hat c$ and ends in $(q_i, \hat s:\hat
  w:\hat w_i)$ and 
  $\lvert \hat \rho_i(j) \rvert > \lvert \hat s:\hat w \rvert$ for all
  $j\geq 1$. 
  If 
  \begin{align*}
  &w_i\wordequiv{n-1}{z} \hat w_i \text{ for all }1\leq i
  \leq m-1,   
  \end{align*}
  then there is some run $\hat \rho_m$ from $\hat c$ to $(q_0,\hat
  s:\hat w:\hat w_m)$ such
  that 
  \begin{align*}
    & w_m \wordequiv{n-1}{z} \hat w_m, \\ 
    &\hat \rho_m \text{ is distinct from each } \hat\rho_i \text{ for
    } 1\leq i < 
    m \text{, and}\\
    &\lvert \hat w_m\rvert \leq 
    \BoundHeightOnestepConstructionSimultanious(m,\lvert \hat w \rvert, n, z).
  \end{align*}
\end{proposition}

We prepare the proof of this proposition with the following lemmas.

\begin{lemma} \label{Lemma:TransferRunsToPrefixes}
  Let $z,m,n\in\N$ such that $z\geq 2$ and $z > m$. 
  Let $w,w'$ be words and let
  $\rho_1, \rho_2, \dots, \rho_m$ be pairwise distinct runs such that
  $\rho_i$ starts in $(q_i,w)$ and ends in $(\hat q_i,v_i)$ for some
  prefix $v_i\leq w$. 
  If $ w \wordequiv{n+1}{z} w'$, then there are 
  prefixes $v_1', v_2', \ldots, v_m'$ of $w'$ such that $v_i \wordequiv{n}{z}
  v'_i$ for all $1\leq i \leq m$ and there are
  pairwise distinct runs 
  $\rho_1', \rho_2', \dots, \rho_m'$ such that $\rho'_i$ 
  starts in $(q_i,w')$, ends in $(\hat q_i,v'_i)$.

  Furthermore,  $v_i=w$ if and only if 
  $v'_i = w'$ and  
  $v_i<w$  implies that there is a letter $a_i$ and words $u_i,
  u_i'$ such that $w=v_i a_i u_i$ and $w'=v_i' a_i u_i'$. 
\end{lemma}
\begin{proof}
  Without loss of generality, assume that $q_i=q_j=q$ and 
  $\hat q_i= \hat q_j = \hat q$ for all $1\leq
  i,j\leq m$. 
  Since $\Lin{n+1}{z}{z}(w) \simeq_z \Lin{n+1}{z}{z}(w')$, a winning
  strategy in the Ehrenfeucht-\Fraisse game induces words $v_1', v_2',
  \dots, v_m'$ such that 
  $(v_1, v_2, \dots, v_m) \mapsto (v_1', v_2', \dots, v_m')$ is a
  partial isomorphism.  
  Thus, $v_i'=v_j'$ iff $v_i=v_j$. 
  Since $z>m$, Duplicator can maintain this
  partial isomorphism for at least one more round of the game.
  Therefore, the labels of the direct neighbours of  $v_i'$ agree with
  the labels of the direct neighbours of $v_i$ which especially
  implies  $v_i'=w'$ iff $v_i=w$. 
  Furthermore, if $v_{i_1}=v_{i_2}=\dots=v_{i_k}$, then 
  $\rho_{i_1}, \rho_{i_2}, \dots, \rho_{i_k}$ witness that $v_{i_1}$
  is coloured by $S^k_{(q,\hat q)}$ in $\Lin{n+1}{z}{z}(w)$. 
  Hence, $v_{i_1}'$ is coloured by $S^k_{(q,\hat q)}$ in
  $\Lin{n+1}{z}{z}(w')$. Hence, there are $k$ pairwise distinct runs
  $\rho_{i_1}', \rho_{i_2}', \dots, \rho_{i_k}'$ from $(q,w')$ to
  $(\hat q, v_{i_k}')$.
  Since, $v_i$ and $v_i'$ are labelled by the same
  $\wordequiv{n}{z}$-type, the claim follows immediately. 
\end{proof}

This lemma provides the transfer of runs from some stack  $s:w$ to
stacks $s:v_i$ with $v_i\leq w$ to another starting stack $s':w'$ if 
$w$ and $w'$ are equivalent words. 
We still need to investigate runs in the other direction. 
We provide a transfer property for runs from some word $w$ to 
extensions $wv_1, wv_2, \dots, wv_m$.

\begin{lemma}  \label{Lem:BuildWordsShort}
  Let $z,m,n\in\N$ such that $z\geq 2$ and $z > m$. 
  Let $\rho_1, \rho_2, \ldots, \rho_m$ be pairwise distinct runs such that
  for each
  $1\leq i \leq m$ the run $\rho_i$ starts in $(q_i, w)$, ends in
  $(q_i,wv_i)$ and never visits 
  $w$ after its initial configuration. Furthermore, let $w'$ be some
  word such that 
  $w\wordequiv{n}{z} w'$.  
  There are words $v_1', v_2', \dots, v_m'$ such that
  $\lvert v_i' \rvert \leq 1+
  m\cdot \lvert Q \rvert \cdot \left\lvert
    \nicefrac{\Sigma^*}{\wordequiv{n}{z}}\right\rvert$, 
  the first letter of $v_i$ and $v_i'$ agree (or
  $v_i=v_i'=\varepsilon$), 
  $wv_i \wordequiv{n}{z} w'v_i'$, and there are pairwise
  distinct runs $\rho_1', \rho_2' \dots, \rho_m'$ such that each run
  $\rho_i'$ starts in $w'$, ends in $w'v_i'$, and never visits $w'$
  after its initial configuration.  
\end{lemma}
\begin{proof}
  For each run $\rho_i$, there is a decomposition
  $\rho_i= \pi_n \circ \lambda_n \circ
  \dots \circ \pi_2 \circ \lambda_2 \circ \pi_1 \circ\lambda_1$ where
  the $\lambda_i$ are high loops and each $\pi_i$ is a run of length
  $1$ that performs a push operation. 
  Since the
  $\wordequiv{n}{z}$ type of a word $w$ determines $\ReturnFunc{z}(w)$
  and $\HighLoopFunc{z}(w)$, we conclude with Proposition
  \ref{Prop:InductiveComputability}  
  that $\HighLoopFunc{z}(wv)=\HighLoopFunc{z}(w'v)$ for all words
  $v\in\Sigma^*$. 
  Thus, 
  there is a run $\rho_i'= \pi_n \circ \lambda_n' \circ \dots \circ
  \pi_2 \circ \lambda_2' \circ \pi_1 \circ \lambda_1'$ where the
  $\lambda_i'$ are high loops such that the runs $\rho_1',
  \rho_2', \dots, \rho_m'$ are pairwise distinct. 
  Note that $\rho_i'$ ends with stack $w'v_i$ and
  $w'v_i \wordequiv{n}{z} wv_i$ because $\wordequiv{n}{z}$ is a right
  congruence. 
  
  If $\lvert v_i \rvert \leq 1+
  m\cdot \lvert Q \rvert \cdot \left\lvert
    \nicefrac{\Sigma^*}{\wordequiv{n}{z}}\right\rvert$
  for all $1\leq i \leq m$ we are done. 
  Otherwise we continue with the following construction.
  Without loss of generality assume that 
  $\lvert v_1\rvert >
  1+
  m\cdot \lvert Q \rvert \cdot \left\lvert
    \nicefrac{\Sigma^*}{\wordequiv{n}{z}}\right\rvert$. Then we find
  nonempty prefixes 
  $u_0 < u_1 < u_2 < \dots < u_m$ such that for all $0\leq i < j \leq
  m$ 
  \begin{enumerate}
  \item 
    $\rho_1'$ passes $w'u_i$ and $w'u_j$ in the same state
    $\hat q\in Q$ for the last time,
  \item $w'u_i \wordequiv{n}{z} wu_j$, and
  \item $1\leq \lvert u_i \rvert < \lvert u_j \rvert \leq 
    1+
    m\cdot \lvert Q \rvert \cdot \left\lvert
    \nicefrac{\Sigma^*}{\wordequiv{n}{z}}\right\rvert$.
  \end{enumerate}
  Let $n_i$ be the maximal position in $\rho_1'$ such that
  $\rho_1'(n_i)=(\hat q, w'u_i)$. 
  For each $0\leq i < j \leq m$, we can define the run
  $\rho_1^{i,j}:=\rho_1'{\restriction}_{[0,n_i]} \circ
  \rho_1'{\restriction}_{[n_j, \length(\hat\rho_1)]}[wu_j/wu_i]$ which ends in
  the stack $wv_1[wu_j/wu_i]$. Since $\wordequiv{n}{z}$ is a right
  congruence, this stack is equivalent to $wv_1$. Furthermore, it is
  shorter than $wv_1$. By pigeonhole principle, there are 
  $0\leq i < j \leq m$ such that $\rho_1^{i,j}$ is distinct from
  $\rho'_2, \rho'_3, \dots, \rho'_m$. Now, we replace $\rho_1'$ by
  $\rho_1^{i,j}$. 
  
  Repetition of this argument yields the claim. 
\end{proof}

For the proof of Proposition \ref{Prop:ConstructOneStep}, we now
compose the previous lemmas. 
Recall that the proposition says the following: given
$m$ runs $\rho_1, \dots, \rho_m$ starting in some stack $s$ that only
add one word to $s$
and given a  stack $\hat s$ whose topmost word is
$\wordequiv{n}{z}$-equivalent to the 
word on top of $s$, we can transfer the runs $\rho_1, \dots,
\rho_m$ to runs $\rho_1', \dots, \rho_m'$ that start at $\hat s$ such
that $\rho_i'$ extends $\hat s$ by one word 
that is $\wordequiv{n-1}{z}$-equivalent to the word 
created by $\rho_i$.

\begin{proof}[of Proposition \ref{Prop:ConstructOneStep}]
  Let $\rho_1, \rho_2, \dots, \rho_m$ and $\hat\rho_1, \hat\rho_2,
  \dots, \hat\rho_{m-1}$ be runs as required in the proposition. 
  Assume that $w_i \wordequiv{n-1}{z} w_j$ 
  and that all runs $\rho_i$ end in the same state, i.e., $q_i=q_j$,
  for all $1\leq i \leq j \leq m$. Later we deal with the other cases.
 
  We decompose each run $\rho_i$ as follows. 
  Let $w_i' :=w\sqcap w_i$. Then 
  $\rho_i=\rho^0_i\circ \rho^1_i\circ\rho^2_i$ where  
  $\rho^0_i$ is a run of length $1$ that performs exactly one $\Clone{2}$
  operation, and $\rho^1_i$ is the run from $s:w:w$ to the last
  occurrence of $s:w:w_i'$.  

  Due to $\TOP{1}(w)=\TOP{1}(\hat w)$, there are runs ${\hat\rho^0}_i$
  from $\hat c$ to
  $\hat s:\hat w:\hat w$ performing only one clone operation and ending in the
  same state as $\rho^0_i$. 

  By Lemma \ref{Lemma:TransferRunsToPrefixes}, we can transfer
  the $\rho^1_i$ to runs ${\hat\rho^1}_i$ starting at $(q,\hat s:\hat
  w:\hat w)$ and ending
  at $\hat s:\hat w:\hat u_i$ such that $\hat u_i\leq \hat w$ and
  $w\sqcap w_i \wordequiv{n-1}{z} \hat u_i$. 
  The lemma allows us to enforce that ${\hat \rho}^1_i={\hat\rho}^1_j$ iff
  $\rho^1_i=\rho^1_j$. 
  
  Let $v_i$ be the word such that $w_i= (w\sqcap w_i) \circ v_i$.
  We use Lemma
  \ref{Lem:BuildWordsShort} 
  and find words 
  $\hat v_1, \dots, \hat v_m$  and runs
  ${\hat\rho^2}_1, \dots,
  {\hat\rho^2}_m$ such that 
  ${\hat\rho^2}_i$ is a run from \mbox{$\hat s:\hat w:\hat u_i$} to
  \mbox{$(q_i,\hat  
    s:\hat w:\hat u_i\hat v_i)$} which visits $\hat s: \hat w: \hat u_i$
  only in its initial configuration such that $\hat u_i \hat v_i
  \wordequiv{n-1}{z} w_i$ and such that
  $\hat u_i\hat v_i$ has length bounded by 
  \begin{align*}
    \BoundHeightOnestepConstructionSimultanious(m,\lvert \hat w\rvert,
    n, z) = \lvert \hat w \rvert +
    1+ m\cdot \lvert Q \rvert \cdot \left\lvert
    \nicefrac{\Sigma^*}{\wordequiv{n}{z}}\right\rvert.    
  \end{align*}
  Furthermore,  
  $\hat\rho^2_i$ and $\hat\rho^2_j$ coincide if and only
  if $\rho^2_i$ and $\rho^2_j$ coincide.
  we claim that the runs
  \begin{align*}
    {\hat\rho^0}_1 \circ {\hat\rho^1}_1 \circ {\hat\rho^2}_1,
    {\hat\rho^0}_2 \circ {\hat\rho^1}_2 \circ {\hat\rho^2}_2, \dots, 
    {\hat\rho^0}_m \circ {\hat\rho^1}_m \circ {\hat \rho^2}_m    
  \end{align*}
  are
  pairwise distinct. 
  First of all we show that $\hat u_i = \hat w \sqcap \hat u_i \hat v_i$:
   Due to the last part of Lemma
  \ref{Lemma:TransferRunsToPrefixes}, there is a letter 
  $a_i$ such that $w=u_i a_i x_i$ for some word $x_i$ and 
  $\hat w = \hat u_i a_i \hat x_i$ for some word $\hat x_i$. 
  Furthermore, $v_i$ and $\hat v_i$ start with the same letter. 
  Due to  $u_i= w\sqcap u_i v_i$, this letter cannot be $a_i$ whence
  $\hat u_i = \hat w \sqcap \hat u_i \hat v_i$.
  
  Heading for a contradiction, assume that  
  ${\hat\rho^0}_i \circ {\hat\rho^1}_i \circ {\hat \rho^2}_i 
  = {\hat\rho^0}_j \circ {\hat\rho^1}_j \circ {\hat \rho^2}_j$.
  Since $\hat\rho^0_i$ and $\hat\rho^0_j$ have both length $1$, this
  implies that ${\hat\rho^0}_i={\hat\rho^0}_j$. 
  Furthermore, we have seen that ${\hat\rho^1}_i$ ends in the last
  occurrence of the greatest
  common prefix $\hat w \sqcap \hat u_i \hat v_i = \hat u_i$. Hence,
  the two runs can only 
  coincide if $\hat u_i = \hat u_j$. 
  But then ${\hat\rho^1}_i = {\hat\rho^1}_j$ because both parts end in
  the last occurrence of a stack with topmost word $\hat u_i$. 
  But this would also imply that ${\hat\rho^2}_i = {\hat\rho^2}_j$. 
  By construction of the three parts, this would imply that
  $\rho^0_i=\rho^0_j$, $\rho^1_i=\rho^1_j$,
  and $\rho^2_i=\rho^2_j$. But this contradicts the
  assumption that $\rho^0_i\circ \rho^1_i\circ\rho^2_i = \rho_i\neq 
  \rho_j = \rho^0_j\circ \rho^1_j\circ\rho^2_j$. 

  Since the runs are all distinct, there is some $j$ such that
  ${\hat\rho^0}_j\circ{\hat\rho^1}_j\circ{\hat\rho^2}_j$ does not
  coincide with any 
  of the $\hat\rho_i$ for \mbox{$1\leq i \leq m-1$}. 
  Note that $\hat\rho_m:={\hat\rho^0}_j \circ {\hat\rho^1}_j \circ
  {\hat\rho^2}_j$ satisfies the claim of the proposition.

  Now, we come to the case that the runs end in configurations with
  different states or different $\wordequiv{n-1}{z}$-types of their
  topmost words. 
  In this case, we just concentrate on those $\rho_i$ which end in the
  same state as 
  $\rho_m$ and with a topmost word of the same
  type as $w_m$. This is sufficient because some run $\rho$ can only
  coincide  with $\hat\rho_i$ if both runs end up in the same state
  and in stacks whose
  topmost words have the same type.  
\end{proof}

\section{Dynamic Small-Witness Property} 
\label{sec:EquivalenceonTuples}
In this section, we define a family of  equivalence relations
on  tuples  in $2$-NPT. 
The equivalence class of a tuple $\rho_1, \dots, \rho_m$ with respect to
one of these relations is the isomorphism type of the 
substructure induced by the relevant $l$-ancestors of $\rho_1, \dots,
\rho_m$ extended by some information for
preserving this isomorphism during an Ehrenfeucht-\Fraisse
game. Recall that such a game 
ends in a winning position for Duplicator if
the relevant $1$-ancestors of the elements that were chosen in the two
structures are isomorphic (cf. Lemma \ref{LemmaRelAnclocalIso}).

We then show how to construct small representatives for each
equivalence class. 
As explained in Section \ref{sec:EFGame}, this result can be turned
into an \FO{} model checking algorithm on the class of $2$-NPT. 

\begin{definition} \label{def:RAequivalence}
  Let $\bar \rho=(\rho_1, \rho_2, \dots, \rho_m)$ be runs of a 
  $2$-PS $\mathcal{N}$ and let 
  $\mathfrak{N}:=\HONPT(\mathcal{N})$.
  Let $l, n_1,n_2,z\in \N$. We define the following relations
  on $\RelAnc{l}{\bar\rho}$. 
  \begin{enumerate}
  \item For $k\leq l$ and $\rho\in \bar\rho$, let $P^k_\rho:=\{\pi \in
    \RelAnc{l}{\bar\rho}: \pi \in \RelAnc{k}{\rho}\}$.
  \item 
    Let $\stackequivTyp{n_2}{z}{n_1}$ be the function 
    that maps a run $\pi$ to the $\stackequiv{n_2}{z}{n_1}$-equivalence
    class of the last stack of $\pi$. 
  \end{enumerate}
  We write $\AncestorClass{l}{n_1}{n_2}{z}(\bar\rho)$ for the following
  expansion of the relevant ancestors of $\bar\rho$:
  \begin{align*}
    \AncestorClass{l}{n_1}{n_2}{z}(\bar\rho):=
    (\mathfrak{N}{\restriction}_{\RelAnc{l}{\bar\rho}},  
    (\trans{\delta})_{\delta\in\Delta},
    \jumpedge,
    \plusedge,
    \stackequivTyp{n_2}{z}{n_1},
    (P^k_{\rho_j})_{k\leq l, 1\leq j \leq m}).
  \end{align*}
  For tuples of runs $\bar\rho=(\rho_1, \dots, \rho_m)$ and 
  $\bar\rho'=(\rho_1', \dots, \rho_m')$ we set 
  $\bar\rho \RelAncequiv{l}{n_2}{z}{n_1} \bar\rho'$ if
  \begin{align*}
    \AncestorClass{l}{n_1}{n_2}{z}(\bar\rho) \simeq
    \AncestorClass{l}{n_1}{n_2}{z}(\bar\rho').
  \end{align*}
\end{definition}
\begin{remark}
  \begin{itemize}
  \item   If $\bar\rho \RelAncequiv{l}{n_2}{z}{n_1} \bar\rho'$ then there is a
    unique isomorphism $\varphi:
    \AncestorClass{l}{n_1}{n_2}{z}(\bar\rho) \simeq 
    \AncestorClass{l}{n_1}{n_2}{z}(\bar\rho')$ witnessing this
    equivalence:
    due to the predicate $P^0_j$, $\rho_j$ is mapped to
    $\rho_j'$ for all $1\leq j \leq m$. Due to the predicate $P^l_j$, the
    relevant ancestors of $\rho_j$ are mapped to the relevant ancestors
    of $\rho'_j$. Finally, $\varphi$ must preserve the order of the
    relevant ancestors of $\rho_j$ because they form a chain with
    respect to $\trans{} \cup \plusedge$ 
    (cf. Proposition \ref{Prop:NextRelAnc}).
  \item 
    Due to Lemma \ref{LemmaRelAnclocalIso}, it is clear that
    $\bar\rho \RelAncequiv{l}{n_2}{z}{n_1} \bar\rho'$ implies that there
    is a partial 
    isomorphism mapping $\rho_i\mapsto \rho_i'$ for all $1\leq i \leq m$. 
  \end{itemize}
\end{remark}

Since equivalent relevant ancestors induce partial isomorphisms, a
strategy that preserves the equivalence between relevant ancestors is
winning for Duplicator in the Ehrenfeucht-\Fraisse-game.

Given a $2$-PS $\mathcal{N}$, set $\mathfrak{N}:=NPT(\mathcal{N})$. 
We show that there is a strategy in the
Ehrenfeucht-\Fraisse game on 
$\mathfrak{N},\bar\rho$ and $\mathfrak{N},\bar\rho'$ 
in which Duplicator can always
choose small elements compared to the size of the elements chosen so
far in the structure where he has to choose. Furthermore, this
strategy will 
preserve equivalence of the relevant ancestors in the following
sense. 
Let 
$\bar\rho , \bar\rho'\subseteq 
\mathfrak{N}$ be the  $n$-tuples chosen in the previous
rounds of the game. Assume that Duplicator managed to maintain the
relevant ancestors of these tuples equivalent, i.e., it holds that 
$\bar\rho\RelAncequiv{l}{n}{z}{k} \bar\rho'$.
Now, Duplicator's strategy enforces that these tuples are extended by
runs $\pi$ and $\pi'$ satisfying the following.
There are numbers $k_i,l_i, n_i$ such that
$\bar\rho,\pi \RelAncequiv{l_i}{n_i}{z}{k_i} \bar\rho', \pi'$ and
furthermore, the size of the run chosen by Duplicator is small
compared to the elements chosen so far.
Before we state the exact claim, we define some functions that provide
bounds for Duplicator's choices. 

\begin{definition}\label{Def:BoundingFunctions}
  Let $\mathcal{N}$ be a $2$-PS and let $a,b\in\N$. 
  We define  the functions 
  \begin{align*}
    &\BoundHeight: \N^5 \to \N, 
    &&\BoundWidth: \N^5 \to \N, \text{ and }
    && \BoundRunLength: \N^5 \to \N
  \end{align*} 
  by induction on the first parameter.
  We set
  \begin{align*}
    \BoundHeight(0, x_2, x_3, x_4, x_5)=
    \BoundWidth(0, x_2, x_3, x_4, x_5)=    
    \BoundRunLength(0,x_2, x_3, x_4, x_5)=0\text{ for all } x_2, x_3,
    x_4, x_5\in \N.
  \end{align*}
  For the inductive step,  let
  $\bar x_{n+1}:=(n+1,z,l',n_1',n_2')\in\N^5$ be arbitrary.
  We set $l:=4l'+5$, $n_1:=n_1'+2(l'+1)+1$, $n_2:=n_2'+4^{l'+1}+1$,
  and $\bar x_n:=(n,z,l,n_1,n_2)$.
  We define auxiliary
  values $H_i^{\text{loc}}$ for $1\leq i \leq 4^{l'+1}$ and
  $H_i^{\text{glob}}$ for $1\leq i \leq n_1' + 4^{l'}$. 
  Recall that we introduced
  $\FuncBoundTopWord,
  \FuncBoundWidthWord$ and
  $\ConstBoundHeightWord $  in Definition 
  \ref{Def:BoundTopWordHeightandWidht}
  and $\BoundHeightOnestepConstructionSimultanious$  in
  Definition \ref{Def:BH1Definition}. Set
  \begin{align*}
    &H_1^{\text{loc}}:= 
    \BoundHeightOnestepConstructionSimultanious( n\cdot4^{4l'+3}, 
    \BoundHeight(\bar x_{n}), n_2-1, z), \\
    &H_{i+1}^{\text{loc}}:=
    \BoundHeightOnestepConstructionSimultanious(1, H_i^{\text{loc}},
    n_1-{i+1}, z), \\
    &H_1^{\text{glob}}:= \BoundHeight(\bar x_{n}) +
    \ConstBoundHeightWord + \FuncBoundTopWord(n_2'+n_1'+4^{l'+1}-1, z)
    ,\text{ and}\\
    &H_{i+1}^{\text{glob}}:=
    \BoundHeightOnestepConstructionSimultanious(1, H_i^{\text{glob}}, 
    n'_2+n'_1+4^{l'+1}-i, z).
  \end{align*}
  Now we set
  \begin{align*}
    &\BoundHeight(\bar x_{n+1}):=\max\left\{ H_{4^{l'+1}}^{\text{loc}},
      H_{n_1'+4^{l'}}^{\text{glob}}\right\},\\
    &\BoundWidth(\bar x_{n+1}) := \BoundWidth(\bar x_{n}) +
      \FuncBoundWidthWord(H_1^{\text{glob}})  + n_1' + 2(l' +1),
      \text{ and}\\ 
    &\BoundRunLength(\bar x_{n+1}):=
    \BoundRunLength(\bar x_{n}) + 
    (4^{l'+1}+1) \BoundHeight(\bar x_{n+1}) 
    \cdot \BoundWidth(\bar x_{n+1})
    \cdot 
    (1+\FuncBoundLoopLength{\mathcal{N}}{z}(\BoundHeight( \bar x_{n+1}))).
  \end{align*}
  where $\FuncBoundLoopLength{\mathcal{N}}{z}$ is the function from
  Proposition \ref{Prop:FuncBoundLoopLengthLemma} that bounds the
  length of short loops. 
\end{definition}
\begin{remark}
  Since   $\BoundHeightOnestepConstructionSimultanious$, 
  $\FuncBoundTopWord$,  
  $\FuncBoundWidthWord$, 
  $\ConstBoundHeightWord$, and 
  $\FuncBoundLoopLength{}{}$  depend computably on 
  $\mathcal{N}$,  the functions $\BoundHeight, \BoundWidth$ and
  $\BoundRunLength$ 
  also  depend computably on $\mathcal{N}$. 
\end{remark}

\begin{proposition} \label{Prop:2NPT-Strategy}
  Let $\mathcal{N}$ be a $2$-PS. Set $\mathfrak{N}:=\HONPT(\mathcal{N})$.  
  Let $n,z,n_1',n_2',l'\in\N$,  
  \mbox{$l := 4l'+5$}, $n_1 := n_1'+2(l'+1)+1,$ and
  $n_2 := n_2' + 4^{l'+1}+1$ such that $z\geq 2$ and $z> n\cdot 4^l$.
  Furthermore, let $\bar\rho$ and $\bar\rho'$ be $n$-tuples of runs of
  $\mathfrak{N}$ such that  
  \begin{enumerate}
  \item $\bar\rho \RelAncequiv{l}{n_2}{z}{n_1} \bar\rho'$, and
  \item $\length(\pi)\leq\BoundRunLength(n,z,l,n_1,n_2)$
    for all $\pi\in\RelAnc{l}{\bar\rho'}$,
  \item $\height(\pi) \leq  \BoundHeight(n,z,l,n_1,n_2)$ for all
    $\pi\in\RelAnc{l}{\bar\rho'}$, and
  \item $\lvert \pi \rvert \leq \BoundWidth(n,z,l,n_1,n_2)$ for all
    $\pi\in\RelAnc{l}{\bar\rho'}$. 
  \end{enumerate}
  For each $\rho\in \mathfrak{N}$ there is some $\rho'\in
  \mathfrak{N}$ such that 
  \begin{enumerate}
  \item $\bar\rho, \rho\RelAncequiv{l'}{n_2'}{z}{n_1'} \bar\rho', \rho'$,
  \item $\length(\pi)\leq\BoundRunLength(n+1,z,l',n_1',n_2')$
    for all $\pi\in\RelAnc{l'}{\bar\rho',\rho'}$,
  \item $\height(\pi) \leq  \BoundHeight(n+1,z,l',n_1',n_2')$ for all
    $\pi\in\RelAnc{l}{\bar\rho',\rho'}$, and
  \item $\lvert \pi \rvert \leq \BoundWidth(n+1,z,l',n_1',n_2')$ for all
    $\pi\in\RelAnc{l}{\bar\rho',\rho'}$. 
  \end{enumerate}
\end{proposition}

This proposition can be reformulated as a finitary constraint for
Duplicator's strategy in the Ehrenfeucht-\Fraisse game on every
$2$-NPT. 
This yields an \FO{} model checking algorithm on $2$-NPT. Before we
present this application of the proposition in Section 
\ref{sec:FODecidability}, we prove this proposition.  
For this purpose we split the claim into several pieces. 
The proposition asserts bounds on the length of the runs and on the
sizes of the final stacks of the relevant ancestors. As the first step
we prove that Duplicator has a strategy that chooses runs with
small final stacks. This result relies mainly on
Propositions \ref{Prop:CompatibilityStackOpTypeq} and
\ref{Prop:ConstructOneStep}. These results allow us to
construct equivalent relevant 
ancestor sets that only contain runs ending in small stacks.
Afterwards, we apply Corollaries \ref{Cor:GlobalBoundRun} and
\ref{Cor:GlobalBoundRun2} in order to shrink the length of the runs
involved. 


\subsection{Construction of Isomorphic Relevant Ancestors}

Before we prove that Duplicator can choose short runs, 
we state some 
auxiliary lemmas concerning the construction of isomorphic relevant
ancestors. The following lemma gives a sufficient criterion 
for the equivalence of the relevant ancestors of two runs. Afterwards,
we show that for each run $\rho$ we can 
construct a second run $\rho'$ satisfying this criterion.

\begin{lemma} \label{LemmaConstructedRelAncEquiv}
  Let $\rho_0 \prec \rho_1 \prec \dots \prec \rho_m=\rho$ be runs such that 
  $\RelAnc{l}{\rho}=\{\rho_i:0\leq i \leq m\}$.
  If 
  $\hat\rho_0\prec \hat\rho_1 \prec \dots \prec \hat\rho_m $ are
  runs such that 
  \begin{itemize}
  \item the final states of $\rho_i$ and $\hat\rho_i$ coincide,
  \item $\rho_0=\Pop{2}^l(\rho_m)$ or $\lvert \rho_0 \rvert = \lvert
    \hat \rho_0 \rvert=1$, 
  \item $\rho_0 \stackequiv{n_2}{z}{n_1} \hat\rho_0$, and
  \item $\rho_i \mathrel{*} \rho_{i+1}$ iff $\hat\rho_i \mathrel{*}
    \hat\rho_{i+1}$ for 
    all $1\leq i < m$ and $*\in\{\plusedge\}\cup\{
    \trans{\delta}: \delta\in\Delta\}$,
  \end{itemize}
  then
  \begin{align*}
    \RelAnc{l}{\hat\rho_m}=\{\hat\rho_i:0\leq i \leq m \}.     
  \end{align*}
  If additionally $\TOP{2}(\rho_i) \wordequiv{n_2-i}{z}
  \TOP{2}(\hat\rho_i)$ for all $0<i \leq m$, then 
  \begin{align*}
    \hat\rho_m \RelAncequiv{l}{n_2-4^l}{z}{n_1} \rho_m.
  \end{align*}
\end{lemma}
\begin{proof}
  First, we show that for all $0 \leq i < j \leq m$, 
  the following statements are true:
  \begin{align}
    \rho_i \trans{\delta} \rho_j &\text{ iff } 
    \hat\rho_i \trans{\delta} \hat\rho_j,  \label{nextbyop} \\
    \rho_i \jumpedge  \rho_j &\text{ iff } 
    \hat\rho_i \jumpedge \hat\rho_j, \text{ and}\\
    \rho_i \plusedge \rho_j &\text{ iff } 
    \hat\rho_i \plusedge \hat\rho_j. \label{nextbyplusone}
  \end{align}
  Note that $\rho_i\trans{\delta} \rho_j$ implies $j=i+1$.
  Analogously, $\hat\rho_i\trans{\delta} \hat\rho_j$ implies $j=i+1$.
  Thus, (\ref{nextbyop}) is true by definition of the sequences. 

  For the other parts, 
  it is straightforward to see that 
  $\lvert \rho_k \rvert - \lvert \rho_j \rvert = 
  \lvert  \hat\rho_k \rvert - \lvert \hat\rho_j \rvert$ for all 
  $0\leq j\leq k  \leq m$: for $k=j$ the claim holds trivially. For
  the induction step from $j$ to $j+1$, the claim follows from the
  assumption that  
  $\rho_j * \rho_{j+1}$ if and only if $\hat\rho_j * \hat\rho_{j+1}$
  for all
  $*\in\{\plusedge\}\cup
  \{\trans{\delta}:\delta\in\Delta\}$. 

  Furthermore, assume that there is some $\hat\pi$ such that
  $\hat\rho_k \prec \hat\pi \prec \hat\rho_{k+1}$. Then it cannot be
  the case that  
  $\hat\rho_k \trans{\delta} \hat\rho_{k+1}$. 
  This implies that $\rho_k \plusedge
  \rho_{k+1}$. By definition, it follows that
  $\hat\rho_k \plusedge \hat\rho_{k+1}$. 
  We conclude directly that $\lvert\hat\pi \rvert \geq \lvert
  \hat\rho_{k+1} \rvert > \lvert \hat \rho_k \rvert$. 
  Thus, 
  \begin{align*}
    &\rho_j \jumpedge \rho_k\text{ iff}\\
    &\lvert \rho_j \rvert = \lvert \rho_k \rvert\text{ and }
    \lvert \pi \rvert > \lvert \rho_j \rvert\text{ for all }\rho_j
    \prec \pi \prec \rho_k \text{ iff}\\    
    &  \lvert \hat\rho_j \rvert = \lvert \hat\rho_k \rvert\text{ and }
    \lvert \hat\pi \rvert > \lvert \hat\rho_j \rvert\text{ for all }\hat\rho_j
    \prec \hat\pi \prec \hat\rho_k \text{ iff}\\
    &\hat\rho_j \jumpedge \hat\rho_k.
  \end{align*}
  Analogously, one obtains (\ref{nextbyplusone}). 

  We now show by induction that
  $\RelAnc{l}{\hat\rho_m}=\{\hat\rho_i:0\leq i \leq 
  m\}$. 
  Note that trivially
  \begin{align*}
    \RelAnc{l}{\hat\rho_m} \cap \{\pi: \hat\rho_{m} \preceq \pi\} = 
    \{\hat\rho_m\}    
  \end{align*}
  holds. Assume that  there is some $0 \leq m_0 \leq m$ such that    
  \begin{align*}
    &\RelAnc{l}{\hat\rho_m} \cap \{\pi: \hat\rho_{m_0} \preceq \pi\} = 
    \{\hat\rho_i: m_0\leq i \leq m\} \text{ and}\\
    &\rho_i\in\RelAnc{k}{\rho}\text{ iff }
    \hat\rho_i\in\RelAnc{k}{\hat\rho_m}\text{ for all }  k\leq l\text{
      and } i\geq m_0.
  \end{align*}
  We distinguish the following cases. 
  \begin{itemize}
  \item 
    If $\rho_{m_0-1} \trans{\delta}
    \rho_{m_0}$ for some transition $\delta$ then $\hat\rho_{m_0-1}
    \trans{\delta} \hat\rho_{m_0}$ due to (\ref{nextbyop}). 
    Thus, there are no runs $\rho_{m_0-1} \prec \pi \prec
    \rho_{m_0}$. Hence, we only have to show that
    \mbox{$\rho_{m_0-1}\in\RelAnc{k}{\rho_m}$} if and only if
    $\hat\rho_{m_0-1}\in\RelAnc{k}{\hat\rho_m}$ for all $k\leq l$. 

    If $\rho_{m_0-1}\in\RelAnc{k}{\rho_m}$, then  there
    is some $j\geq m_0$ such that
    \mbox{$\rho_{j}\in\RelAnc{k-1}{\rho_m}$} and  
    $\rho_{m_0-1}$ is connected to $\rho_j$ via some edge. 
    But then \mbox{$\hat\rho_j\in\RelAnc{k-1}{\hat\rho_m}$} and 
    $\hat\rho_{m_0-1}$ is connected with $\hat\rho_j$ via the same
    sort of edge. Thus,
    \mbox{$\hat\rho_{m_0-1}\in\RelAnc{k}{\hat\rho_m}$}. 

    The other direction is completely analogous. 
  \item 
    Otherwise, assume that there is some $\rho_{m_0-1} \prec \pi
    \prec \rho_{m_0}$. Since its direct predecessor is not in
    $\RelAnc{l}{\rho_m}$, $\rho_{m_0}\notin\RelAnc{l-1}{\rho}$. 
    Thus, $\hat\rho_{m_0}\notin\RelAnc{l-1}{\hat\rho}$. 
    By
    construction, 
    $\hat\rho_{m_0-1} \plusedge \hat\rho_{m_0}$.
    Thus, $\lvert \hat \pi \rvert \geq \lvert \hat\rho_{m_0} \rvert$ for
    all $\hat\rho_{m_0-1} \prec \hat \pi \prec \hat\rho_{m_0}$.
    This implies that 
    $\pi \not\jumpedge \hat\rho_i$ and
    $\pi\not\plusedge \hat\rho_i$
    for all $m_0 < i \leq m$. This shows that
    $\pi\notin\RelAnc{l}{\hat\rho_m}$. 
    
    We obtain that 
    $\hat\rho_{m_0-1}\in\RelAnc{k}{\hat\rho_m}$ iff 
    $\rho_{m_0-1}\in\RelAnc{k}{\rho_m}$  for all $k\leq l$ analogously
    to the previous case. 
  \end{itemize}
  Up to now, we have shown that 
  $\RelAnc{l}{\hat\rho_m} \cap \{\pi: \hat\rho_0 \preceq \pi \} = 
  \{\hat\rho_i: 0 \leq i \leq m\}$.  
  In order to prove 
  $\RelAnc{l}{\hat\rho_m} = 
  \{\hat\rho_i: 0 \leq i \leq m\}$, we have to show that $\hat\rho_0$
  is the minimal element of $\RelAnc{l}{\hat\rho_m}$. 
  
  There are the following cases
  \begin{enumerate}
  \item $\rho_0=\Pop{2}^l(\rho_m)$. In this case, we conclude that
    $\hat\rho_0 = \Pop{2}^l(\hat\rho_m)$ by construction. But 
    Lemma \ref{LemmaMinRelAnc} then implies that $\hat\rho_0$ is the
    minimal element of $\RelAnc{l}{\hat\rho_m}$.
  \item $\lvert \rho_0 \rvert = \lvert \hat\rho_0 \rvert = 1$.
    Note that $\rho_0\notin\RelAnc{l-1}{\rho_m}$
    because $\rho_0$ is minimal in $\RelAnc{l}{\rho_m}$. Thus, we know
    that $\hat\rho_0\notin\RelAnc{l-1}{\hat\rho_m}$. 
    
    Heading for a contradiction, assume that there is some
    $\hat\pi\in\RelAnc{l}{\hat\rho_m}$ with 
    $\hat\pi\prec\hat\rho_0$. 
    We conclude that $\hat\pi\jumpedge
    \hat\rho_k$ or $\hat\pi\plusedge\hat\rho_k$ for
    some $\hat\rho_k\in\RelAnc{l-1}{\hat\rho_m}$. But this implies
    that $\lvert \hat\pi \rvert < \lvert \hat\rho_0 \rvert =1$. Since
    there are no stacks of width $0$, this is a contradiction.
    
    Thus, there is no $\hat\pi\in\RelAnc{l}{\hat\rho_m}$ that is a
    proper prefix of $\hat\rho_0$.
  \end{enumerate}
  We conclude that $\RelAnc{l}{\hat\rho_m}=\{\hat\rho_i: 0\leq i \leq
  m\}$. 
  
  Let us turn to the second part of the lemma. Assume that
  $\TOP{2}(\rho_i)\wordequiv{n_2-i}{z} \TOP{2}(\hat\rho_i)$ for all 
  \mbox{$0\leq i \leq m$}. Since $\hat\rho_i$ and $\hat\rho_{i+1}$
  differ in at most one 
  word, a straightforward induction shows that $\rho_i
  \stackequiv{n_2-i}{z}{n_1-\lvert\rho_0\rvert + \lvert \rho_i \rvert}
  \hat\rho_i$ (cf. Proposition
  \ref{Prop:CompatibilityStackOpTypeq}). But this implies $\hat\rho_m 
  \RelAncequiv{l}{n_2-4^l}{z}{n_1} \rho_m$ because 
  \mbox{$\lvert \rho_0 \rvert \leq \lvert \rho_i \rvert$}  as we have
  seen in Lemma \ref{LemmaMinRelAnc}. 
\end{proof}

The previous lemma gives us a sufficient condition for the equivalence
of relevant ancestors of two elements. Now, we show how to
construct such a chain of relevant ancestors. 

\begin{lemma} \label{LemmaAncestorConstruction}
  Let $l,n_1, n_2, m,z\in\N$ such that $n_2\geq 4^l$ and $z\geq 2$. 
  Let 
  \begin{align*}
    &\rho_0 \prec \rho_1 \prec \dots \prec \rho_m=\rho\text{ be runs
      such that}\\ 
    &\RelAnc{l}{\rho}\cap \{\pi: \rho_0\preceq \pi\preceq \rho\} =
    \{\rho_i: 0\leq i \leq m\}.      
  \end{align*}
  Let $\hat\rho_0$ be a run such that
  $\rho_0 \stackequiv{n_2}{z}{n_1} \hat\rho_0$ and such that both runs
  end in the same state.
  Then we can effectively construct runs 
  \begin{align*}
    \hat\rho_0\prec \hat\rho_1\prec \dots \prec \hat\rho_m=:\hat\rho    
  \end{align*}
  such that 
  \begin{itemize}
  \item the final states of $\rho_i$ and $\hat\rho_i$ coincide for all $0\leq
  i \leq m$,
  \item $\rho_i\trans{\delta} \rho_{i+1}$ iff
    $\hat\rho_i\trans{\delta} \hat\rho_{i+1}$ and
    $\rho_i \plusedge \rho_{i+1}$ iff
    $\hat\rho_i\plusedge \hat\rho_{i+1}$ 
    for all $0\leq
    i < m$, and
  \item $\TOP{2}(\rho_i) \wordequiv{n_2- i}{z} \TOP{2}(\hat\rho_i)$
    for all $0\leq
    i \leq m$.
  \end{itemize}
\end{lemma}
\begin{proof}
 Assume that we have constructed 
 \begin{align*}
   \hat\rho_0 \prec \hat\rho_1\prec \dots \prec
   \hat\rho_{m_0},   
 \end{align*}
 for some $m_0<m$ such that for all $0 \leq i\leq
  m_0$
  \begin{enumerate}
  \item the final states of $\rho_i$ and $\hat\rho_i$
    coincide, \label{condition_States} 
  \item $\rho_i\trans{\delta} \rho_{i+1}$ iff \label{TEST_three}
    $\hat\rho_i\trans{\delta} \hat\rho_{i+1}$ and
    $\rho_i \plusedge \rho_{i+1}$ iff
    $\hat\rho_i\plusedge \hat\rho_{i+1}$ 
    (note that $\rho_i \trans{} \rho_{i+1}$ or $\rho_i
    \plusedge \rho_{i+1}$ hold 
    due to Proposition \ref{Prop:NextRelAnc}), and
  \item \label{COndition_Wordequivalence}
    $\TOP{2}(\rho_i) \wordequiv{n_2-i}{z} \hat\rho_i$.
  \end{enumerate}
  We extend this chain by a new element $\rho'_{m_0+1}$ such that
  all these conditions are again satisfied. We distinguish two cases.
  
  First, assume that $\rho_{m_0} \trans{\delta} \rho_{m_0+1}$. 
  Since $\rho_{m_0} \wordequiv{n_2-m_0}{z} \hat\rho_{m_0}$, 
  $\TOP{1}(\rho_{m_0}) = \TOP{1}(\hat\rho_{m_0})$. Due to Condition 
  \ref{condition_States}, their final states also coincide. 
  Hence, there is a
  $\hat\rho_{m_0+1}$ such that
  $\hat\rho_{m_0}\trans{\delta} \hat\rho_{m_0+1}$. 
  Due to Proposition  
  \ref{Prop:CompatibilityStackOpTypeq},
  $\hat\rho_{m_0+1}$ satisfies Condition 
  (\ref{COndition_Wordequivalence}).
  
  Now, consider the case $\rho_{m_0} \plusedge
  \rho_{m_0+1}$.
  The run from $\rho_{m_0}$ to $\rho_{m_0+1}$ starts from
  some stack $s$ and ends in some stack $s:w$ for $w$ some word, the
  first operation is a clone and then $s$ is never reached
  again. Hence, we can use Proposition \ref{Prop:ConstructOneStep}
  in order to find some appropriate $\hat\rho_{m_0+1}$ that satisfies
  Condition (\ref{COndition_Wordequivalence}).
\end{proof}

The previous lemmas give us the possibility to construct an isomorphic
copy of the 
relevant ancestors of a single run $\rho$. In our proofs, we want to
construct such a copy while avoiding relevant ancestors of certain
other runs. Using the full power of Proposition \ref{Prop:ConstructOneStep}
we obtain the following stronger version of the lemma. 


\begin{corollary} \label{Cor:AncestorConstruction}
  Let $l,n_1, n_2, m,z\in\N$ be numbers such that $z> m \cdot 4^l$ and
  $n_2\geq 4^l$.
  Let $\bar\rho$ and $\bar\rho'$ be $m$-tuples such that 
  $\bar\rho \RelAncequiv{l}{n_2}{z}{n_1} \bar\rho'$ and $\varphi_l$ is an
  isomorphism witnessing this equivalence. 
  Furthermore, let $\rho_0 \prec \rho_1 \prec \dots \prec \rho_m$ be
  runs such that for each $i<m$ we have  $\rho_i \trans{}\rho_{i+1}$ or
  $\rho_i \plusedge \rho_{i+1}$. 
  
  If  $\rho_0\in\RelAnc{l}{\bar \rho}$,
  and if $\rho_1\notin
  \RelAnc{l}{\bar\rho}$   
  then we can construct $\hat\rho_0:=\varphi_l(\rho_0) \prec
  \hat\rho_1 \prec
  \hat\rho_2\prec \dots \prec \hat\rho_m$ 
  satisfying the conditions from the previous 
  lemma but additionally with the property that
  $\hat\rho_1\notin\RelAnc{l}{\bar\rho'}$. 
\end{corollary}
\begin{proof}
  We distinguish two cases. 
  \begin{enumerate}
  \item Assume that $\rho_0 \trans{} \rho_1$. 
    Due to the equivalence of $\rho_0$ and $\hat\rho_0$, we can apply
    the transition connecting $\rho_0$ with $\rho_1$ to $\hat\rho_0$
    and obtain a run $\hat\rho_1$. We have to prove that
    $\hat\rho_1\notin\RelAnc{l}{\bar\rho'}$.
    
    Heading for a contradiction assume that
    $\hat\rho_1\in\RelAnc{l}{\bar\rho'}$. Then $\varphi^{-1}_l$
    preserves the edge between $\hat\rho_0$ and $\hat\rho_1$, i.e., 
    $\rho_0=\varphi^{-1}_l(\hat\rho_0) \trans{}
    \varphi^{-1}_l(\hat\rho_1)$. But this implies that
    $\varphi^{-1}_l(\hat\rho_1)=\rho_1$ which contradicts the
    assumption that $\rho_1\notin\RelAnc{l}{\bar\rho}$. 
  \item 
    Assume that $\rho_0 \plusedge \rho_1$. 
    Up to threshold $z$, for each $\hat\pi$ such that
    $\hat\rho_0\plusedge \hat\pi$ and
    \mbox{$\hat\pi\in\RelAnc{l}{\bar\rho'}$} there is  a run
    $\rho_0 \plusedge \varphi_l^{-1}(\hat\pi)$. Since 
    $\rho_1\notin \RelAnc{l}{\bar\rho}$, we find another run
    $\hat\rho_1$ that satisfies 
    the conditions of 
    the previous lemma and
    $\hat\rho_1\notin \RelAnc{l}{\bar\rho'}$. This is due to the fact that
    Proposition \ref{Prop:ConstructOneStep} allows us to transfer up to
    $z>\lvert \RelAnc{l}{\bar\rho'} \rvert$ many runs simultaneously. \qed
  \end{enumerate}
\end{proof}

\subsection{Construction of Small Equivalent Stacks}

In this section, we prove that Duplicator has a strategy that
preserves the isomorphism type of the relevant ancestors while
choosing runs whose relevant ancestors end in small stacks.  
Later, we show how to bound the length of such
runs. 

The analysis of this strategy decomposes into the local and the
global case. 
We say Spoiler makes a local move if he chooses a new element such
that one of its relevant ancestors is an
relevant ancestor of the elements chosen so far. 
We say Spoiler makes a global move if he  chooses an element such that
its set of  relevant ancestors does not intersect with the
set of relevant ancestors of the elements chosen so far. 

We first head for the result that Duplicator can manage the
local case 
in such a way that he chooses an element such that all its relevant
ancestors end in small stacks. Then we show that Duplicator can manage
the global case analogously. 

\begin{lemma}   \label{LemmaLocalStep}
  Let $n,z,l',n_1',n_2'\in\N$ be numbers such that 
  $z\geq 2$, $n_1'>0$, $n_2'>0$, $l := 4l'+5$,  $z > n \cdot 4^l$, 
  \mbox{$n_1 := n_1' + 2 (l'+1)+1$}, and
   $n_2 := n_2' + 4^{l'+1}+1$.
  
  Let $\bar\rho, \bar\rho'$ be $n$-tuples of runs such that
  \mbox{$\bar\rho \RelAncequiv{l}{n_2}{z}{n_1} \bar\rho'$} and such
  that
  $\height(\pi)\leq \BoundHeight(n, z, l, n_1, n_2)$ and
  $\lvert \pi \rvert \leq \BoundWidth(n,z,l,n_1,n_2)$
  for all $\pi\in\RelAnc{l}{\bar\rho'}$. 
  Furthermore, let $\rho$ be some run such that 
  $\RelAnc{l'+1}{\rho}\cap \RelAnc{l'+1}{\bar\rho}\neq\emptyset$.
  Then there is some run $\rho'$ such that
 \begin{align*}
    &\height(\rho')\leq \BoundHeight(n+1, z, l', n_1', n_2'),
    &\lvert \rho'\rvert \leq  \BoundWidth(n+1, z, l', n_1', n_2')
    \text{, and}
    &&
    (\bar\rho,\rho) \RelAncequiv{l'}{n_2'}{z}{n_1'} (\bar\rho', \rho').
  \end{align*}
\end{lemma}
\begin{proof}
  Let $\varphi_l: \AncestorClass{l}{n_1}{n_2}{z}(\bar\rho) \simeq 
  \AncestorClass{l}{n_1}{n_2}{z}(\bar\rho')$ denote the isomorphism
  that witnesses \mbox{$\bar\rho \RelAncequiv{l}{n_2}{z}{n_1} \bar\rho'$}.
  Let $\rho_0\in\RelAnc{l'+1}{\rho}$ be maximal such that 
  \begin{align*}
    \RelAnc{l'+1}{\rho}\cap\{\pi : \pi \preceq \rho_0\}\subseteq
    \RelAnc{4l'+3}{\bar\rho}\subseteq\RelAnc{l}{\bar\rho}.    
  \end{align*}
  There are  numbers  $m_0 \leq 0 \leq m_1$ and runs 
  \begin{align*}
    \rho_{m_0} \prec \rho_{m_0+1} \prec
    \dots \prec \rho_0 \prec \rho_1 \prec \dots \prec \rho_{m_1}
  \end{align*}
  such that  
  $\RelAnc{l'+1}{\rho} = \{\rho_i: m_0 \leq i \leq m_1\}$. 
  We
  set $\rho_i':=\varphi_l(\rho_i)$ for all $m_0 \leq i\leq 0$. 
  Note that $\height(\rho_0')\leq \BoundHeight(n,z,l,n_1,n_2)$ and 
  \mbox{$\lvert \rho_0' \rvert \leq \BoundWidth(n,z,l,n_1,n_2)$.}
  
  Next, we
  construct $\rho'_1, \dots, \rho'_{m_1}$  such that
  $\rho':=\rho'_{m_1}$ has relevant ancestors isomorphic to those of $\rho$.
  We first define $\rho_1'$ such that
  \begin{enumerate}
  \item the final state of $\rho'_1$ and $\rho_1$ coincide,
  \item $\TOP{2}(\rho_1') \wordequiv{n_2-1}{z} \TOP{2}(\hat\rho_1)$,
  \item $\rho'_0\trans{\delta} \rho'_1$ iff $\rho_0\trans{\delta} \rho_1$,
  \item $\rho'_0\plusedge \rho'_1$ iff
    $\rho_0 \plusedge \rho_1$,
  \item $\rho'_1\notin \RelAnc{4l'+3}{\bar\rho'}$, and
  \item $\height(\rho_1') \leq 
    \BoundHeightOnestepConstructionSimultanious( n\cdot4^{4l'+3}, 
    \height(\rho_0'), n_2-1, z) \leq H_1^{\text{loc}}$ (cf. Definition
    \ref{Def:BoundingFunctions}).
  \end{enumerate}
  If $\rho_0\trans{\delta}\rho_1$, this construction is trivial. Note
  that there is an element $\rho_0'\trans{\delta}\rho_1'$ and
  \mbox{$\rho_1'\notin\RelAnc{4l'+3}{\bar\rho'}$} (otherwise,
  $\rho_1=\varphi_l^{-1}(\rho_1)\in\RelAnc{4l'+3}{\bar\rho'}$
  contradicting the maximality of $\rho_0$). 
  If $\rho_0\plusedge \rho_1$, we just apply
  Proposition \ref{Prop:ConstructOneStep}. 
  
  Now, we continue constructing  
  $\rho_2', \dots, \rho'_{m_1}=:\rho'$ such that 
  \begin{enumerate}
  \item the final states of $\rho_i$ and $\rho'_i$ coincide 
    for all \mbox{$2\leq i \leq m_1$}, 
  \item $\TOP{2}(\rho_i') \wordequiv{n_2-i}{z} \TOP{2}(\hat\rho_i)$,
  \item for all
    \mbox{$2\leq i < m_1$},
    $\rho'_i\trans{\delta} \rho'_{i+1}$  if
    $\rho_i\trans{\delta} \rho_{i+1}$; in this case
    $\height(\rho'_{i+1})\leq \height(\rho'_i)+1$,
  \item for all
    \mbox{$0\leq i < m_1$},
    $\rho'_i\plusedge \rho'_{i+1}$ if
    $\rho_i \plusedge \rho_{i+1}$; in this case 
    the use of Proposition \ref{Prop:ConstructOneStep} ensures that
    $\height(\rho'_{i+1})\leq
    \BoundHeightOnestepConstructionSimultanious^{n,z}_{l',n_2}(1, 
    \height(\rho'_i), n_2-i, z)$
  \end{enumerate}
  By definition, it is clear that conditions 1--3 hold also for all
  $m_0\leq i < 0$. Using Lemma
  \ref{LemmaConstructedRelAncEquiv}, we obtain that 
  $\rho \RelAncequiv{l'+1}{n_2'}{z}{n_1'} \rho'$. 
  A simple induction shows that
  $\height(\rho_i')\leq H_i^{\text{loc}}$ whence
  $\height(\rho')\leq
  \BoundHeight(n+1, z, l', n_1', n_2')$.
  Furthermore, due to Lemma \ref{CorDistRelAnc} $\lvert \rho_i \rvert - \lvert
  \rho_0 \rvert \leq 2(l'+1)$ whence
  $\lvert \rho_i'\rvert \leq \lvert \rho_0'\rvert + 2(l'+1) \leq \BoundWidth(n+1, z, l',
  n_1', n_2')$.

  We still have to show that the isomorphism between
  $\RelAnc{l}{\bar\rho}$ and $\RelAnc{l}{\bar\rho'}$ and the
  isomorphism between 
  $\RelAnc{l'}{\rho}$ and $\RelAnc{l'}{\rho'}$ are compatible in the sense
  that they induce an isomorphism between
  $\RelAnc{l'}{\bar\rho,\rho}$ and $\RelAnc{l'}{\bar\rho',\rho'}$. 
  The only possible candidate is 
  \begin{align*}
    \varphi_{l'}:\RelAnc{l'}{\bar\rho,\rho} &\to
    \RelAnc{l'}{\bar\rho',\rho'}\\ 
    \pi &\mapsto
    \begin{cases}
      \rho'_i &\text{for }  \pi=\rho_i, m_0\leq i\leq m_1\\
      \varphi_l(\pi) & \text{for } \pi\in\RelAnc{l'+1}{\bar\rho}.
    \end{cases}
  \end{align*}
  In order to see that this is a well-defined function, we have to
  show that if $\rho_i\in\RelAnc{l'+1}{\bar\rho}$
  then $\rho_i'=\varphi_l(\rho_i)$ for each $m_0\leq i \leq m_1$.
  Note that
  $\rho_i\in\RelAnc{l'+1}{\bar\rho}\cap\RelAnc{l'+1}{\rho}$ implies
   that $\pi\in\RelAnc{3l'+3}{\bar\rho}$ for all
  $\pi\in\RelAnc{l'+1}{\rho}$ with $\pi\preceq \rho_i$
  (cf.\ Corollary \ref{CorRelAncDistBound}). But then by definition
  $i\leq 0$ and $\rho'_i = \varphi_l(\rho_i)$. 

  We claim that $\varphi_{l'}$ is an isomorphism. 
  Since we composed $\varphi_{l'}$ of existing isomorphisms
  $\RelAnc{l'}{\bar\rho}\simeq \RelAnc{l'}{\bar\rho'}$ and
  $\RelAnc{l'}{\rho}\simeq \RelAnc{l'}{\rho'}$, respectively, we only have
  to consider the following question: 
  let $\pi\in \RelAnc{l'}{\bar\rho}$ and
  $\hat\pi\in\RelAnc{l'}{\rho}$;
  does  $\varphi_{l'}$  preserve  edges
  between $\pi$ and $\hat\pi$  and does $\varphi_{l'}^{-1}$
  preserve edges between the images of $\pi$ and $\hat\pi$?
  In other words, we have to show that for each
  \begin{align*}
    *\in\{\jumpedge, \jumpleftedge,
    \plusedge, \plusleftedge\}&\cup\{
    \trans{\delta}:
    \delta\in\Delta\}\cup\{\invtrans{\delta}:\delta\in\Delta\},\text{
      we have }
    \pi * \hat\pi \text{ iff } \varphi_{l'}(\pi) *
    \varphi_{l'}(\hat\pi).    
  \end{align*}
   The following case distinction treats all these cases. 
  \begin{itemize}
  \item   Assume that there is some 
    $*\in\{\jumpedge, 
    \plusedge\}\cup
    \{\trans{\delta}:\delta\in\Delta\}$ such that  
    $\pi * \hat\pi$. Then 
    \mbox{$\pi\in\RelAnc{l'+1}{\rho}$}. Thus, there are $m_0 \leq i < j \leq
    m_1$ such that $\pi=\rho_i$ and $\hat\pi=\rho_j$. We have already seen
    that then $\varphi_{l'}(\pi) = \rho'_i$ and
    $\varphi_{l'}(\hat\pi)=\rho'_j$ and 
    these elements are connected by an edge of the same type due to
    the construction of $\rho'_i$ and $\rho'_j$. 
  \item Assume that there is some 
    $*\in\{\jumpleftedge, 
    \plusleftedge\} \cup
    \{\invtrans{\delta}:\delta\in\Delta\}$ such that
    $\pi * \hat\pi$. Then 
    $\hat\pi\in\RelAnc{l'+1}{\bar\rho}$ whence $\varphi_{l'}$
    coincides with the 
    isomorphism $\varphi_l$ on $\pi$ and $\hat\pi$. But $\varphi_l$
    preserves edges whence $\pi*\hat\pi$  implies
    $\varphi_{l'}(\pi) *  \varphi_{l'}(\hat\pi)$.  
  \item Assume that there is some
    $*\in\{\jumpedge, 
    \plusedge\}\cup
    \{\trans{\delta}:\delta\in\Delta\}$ such that  
    $\varphi_{l'}(\pi) * \varphi_{l'}(\hat\pi)$. 
    By definition,
    $\varphi_{l'}(\hat\pi)\in\RelAnc{l'}{\rho'}$ 
    whence \mbox{$\varphi_{l'}(\hat\pi)=\rho'_j$} for some $m_0\leq j \leq
    m_1$. Thus,  \mbox{$\varphi_{l'}(\pi)\in\RelAnc{l'+1}{\rho'}$} whence
    \mbox{$\varphi_{l'}(\pi)=\rho'_i$} for some $m_0\leq i <j$. 
    We claim that \mbox{$\pi = \rho_i$}. Note that due to Corollary
    \ref{CorRelAncDistBound} for all $m_0\leq k \leq i$ we have
    $\rho'_k\in\RelAnc{3l'+3}{\varphi_{l'}(\pi)}$. 
    Since $\varphi_{l'}(\pi)=\varphi_l(\pi)\in\RelAnc{l'}{\bar\rho'}$, 
    we conclude that $\rho'_k\in\RelAnc{4l'+3}{\bar\rho'}$ for all
    \mbox{$m_0\leq k \leq i$}. By
    construction,  
    this implies $i\leq 0$ and
    $\varphi_{l'}(\pi) = \rho_i'=\varphi_l(\rho_i)$.
    Furthermore, since $\pi\in\RelAnc{l'}{\bar\rho}$, 
    $\varphi_{l'}(\pi)=\varphi_l(\pi)$. 
    Since $\varphi_l$ is an isomorphism, it follows that 
    $\pi=\rho_i$.  
    But this implies that there is an edge from $\pi=\rho_i$ to
    $\hat\pi=\rho_j$.
  \item 
    Assume that there is some 
    $*\in\{\jumpleftedge, 
    \plusleftedge\} \cup
    \{\invtrans{\delta}:\delta\in\Delta\}$ such that
    $\varphi_{l'}(\pi)*\varphi_{l'}(\hat\pi)$. This implies
    \begin{align}
     \label{Ancestorinstersection}
     \varphi_{l'}(\hat\pi)\in\RelAnc{l'+1}{\bar\rho'}\cap\RelAnc{l'+1}{\rho}.  
    \end{align}
    By definition, 
    $\hat\pi=\rho_j$ and 
    $\varphi_{l'}(\hat\pi)=\rho'_j$ for some $m_0\leq j \leq m_1$. Due
    to (\ref{Ancestorinstersection}), 
    \mbox{$\rho'_i\in\RelAnc{4l'+3}{\bar\rho'}$} for all $m_0 \leq i \leq j$. 
    Since $\rho_1'\notin\RelAnc{4l'+3}{\bar\rho'}$,  $j\leq 0$. Thus, 
    $\rho_j\in\RelAnc{4l'+3}{\bar\rho}$ and 
    $\varphi_{l'}(\hat\pi)=\varphi_l(\hat\pi)$. Since
    $\varphi_l$ preserves the relevant ancestors of $\bar\rho$ level
    by level, we 
    obtain that $\hat\pi\in\RelAnc{l'+1}{\bar\rho}$. Since
    $\pi\in\RelAnc{l'+1}{\bar\rho}$, we obtain that 
    $\varphi_{l'}(\pi) = \varphi_l(\pi)$ and $\varphi_{l'}(\hat\pi) =
    \varphi_l(\hat\pi)$. Since
    $\varphi_l$ is an isomorphism, we conclude that $\pi*\hat\pi$ 
  \end{itemize}
  Thus, we have shown that $\varphi_{l'}$ is an isomorphism witnessing 
  $\bar\rho,\rho \RelAncequiv{l'}{n_2'}{z}{n_1'} \bar\rho', \rho'$. 
\end{proof}

The previous lemma shows that Duplicator can respond to local moves 
in such a way that she preserves 
isomorphisms of  relevant  ancestors while choosing small stacks.
In the following we deal with global moves of Spoiler. 
We present a strategy for Duplicator that answers a global move
by choosing a run with the following property. 
Duplicator chooses a run such that the isomorphism of relevant
ancestors is preserved and such that all relevant ancestors of
Duplicator's choice end in small stacks. 
We split this proof into two lemmas. 
First, we address
the problem that Spoiler may choose an element  far away from
$\bar\rho$ but close to 
$\bar\rho'$. In this situation, Duplicator has to find a run that has
isomorphic 
relevant ancestors  and which is far away from $\bar\rho'$. Afterwards,
we show that Duplicator can even choose such an element whose relevant
ancestors all end in small stacks. 

\begin{lemma} \label{LemmaGlobalStep1}
  Let $n, l',  n_1',  n_2'\in\N$ be numbers. We set
  $l:=4l'+5$, 
  $n_1 := n_1' + 2 (l'+1)+1$, and $n_2 := n_2' + 4^{l'+1}+1$.
  Let $z\in\N$ satisfy $z > n \cdot 4^l$ and
  $z\geq 2$. 
  
  Let $\bar\rho$ and $\bar\rho'$ be $n$-tuples of runs such that 
  \mbox{$\bar\rho \RelAncequiv{l}{n_2}{z}{n_1}
    \bar\rho'$}. Furthermore, let $\rho$ be a run such 
  that $\RelAnc{l'+1}{\bar\rho} \cap \RelAnc{l'+1}{\rho} = \emptyset$. Then
  there is some run 
  $\rho'$ such that $\bar\rho,\rho \RelAncequiv{l'}{n_2'}{z}{n_1'}
  \bar\rho',\rho'$. 
\end{lemma}
\begin{proof}
  Let  $\varphi_l$ be the isomorphism witnessing
  $\RelAnc{l}{\bar\rho} \RelAncequiv{l}{n_2}{z}{n_1}
  \RelAnc{l}{\bar\rho'}$.  
  If 
  $\RelAnc{l'+1}{\bar\rho'} \cap \RelAnc{l'+1}{\rho} = \emptyset$,
  we can
  set $\rho':=\rho$ and we are done. 
  Otherwise, let $\pi^0_0 \prec \pi^0_1 \prec \dots \prec \pi^0_{n_0}$
  be an enumeration of all elements of 
  \mbox{$\RelAnc{l'+1}{\bar\rho'} \cap \RelAnc{l'+1}{\rho}$}. 
  Due to Corollary \ref{CorTouch}, 
  \mbox{$\RelAnc{l'+1}{\rho} \cap\{ \pi:\pi\preceq \pi^0_{n_0}\} \subseteq 
  \RelAnc{3l'+3}{\bar\rho'}$}. Since $l> 3l'+3$, we can set
  $\pi^1_i:=\varphi_{l}^{-1}(\pi^0_i)$ for all $0\leq i \leq n_0$. 
  Due to Lemmas  \ref{LemmaAncestorConstruction} and 
  \ref{LemmaConstructedRelAncEquiv}, there
  is an 
  extension $\pi^1_{n_0}\prec \rho^1$ such that 
  $\RelAnc{l'+1}{\rho} \RelAncequiv{l'+1}{n_2'}{z}{n_1'}
  \RelAnc{l'+1}{\rho^1}$ and $\pi^1_i\in\RelAnc{l'+1}{\rho^1}$ for all $0\leq
  i \leq 
  n_0$. If $\RelAnc{l'+1}{\rho^1} \cap \RelAnc{l'+1}{\bar\rho'}= 
  \emptyset$ we set 
  $\rho':=\rho^1$ and we are done. 

  Otherwise we can repeat this process,
  defining $\pi^2_i:=\varphi_{l}^{-1}(\pi^1_i)$ for the maximal $n_1\leq
  n_0$ such that  
  $\pi^1_{i}\in \RelAnc{3l'+3}{\bar\rho'}$
  for all $0\leq i\leq n_0$. Then we extend this run to some run
  $\rho^2$. 
  If this process 
  terminates with the construction of some run $\rho^i$ such that 
  $\RelAnc{l'+1}{\rho^i} \cap \RelAnc{l'+1}{\bar\rho'} = \emptyset$,
  we
  set $\rho':=\rho^i$ and we are done. 
  If this is not the case, recall that $\RelAnc{3l'+3}{\bar\rho'}$ is
  finite. 
  Thus, we eventually reach the
  step were we have defined $\pi^0_0, \pi^1_0, \dots, \pi^m_0$ for some $m\in\N$
  such that for the first time $\pi^m_0 = \pi^i_0$ for some
  $i< m$.  But if $i>0$, then 
  \begin{align*}
    \pi^{m-1}_0 =\varphi_{l}(\pi^m_0)= \varphi_{l}(\pi^i_0)=\pi^{i-1}_0.  
  \end{align*}
  This
  contradicts the minimality of 
  $m$. We conclude that $\pi^m_0=\pi^0_0$ which implies that 
  \mbox{$\pi^0_0\in \RelAnc{3l'+3}{\bar\rho}$.} 
  Furthermore, by definition 
  we have  $\pi^0_0\in \RelAnc{l'+1}{\rho}$ and
  there is a maximal $i$ such that  $\pi^0_i\in\RelAnc{3l'+3}{\bar\rho}$. 
  Since $z>\lvert \RelAnc{l}{\bar\rho} \rvert$, 
  we can apply Corollary \ref{Cor:AncestorConstruction} 
  and construct a chain $\varphi_l(\pi^0_i)\prec \rho_{i+1}' \prec
  \rho_{i+2}'\prec  \dots
  \prec \rho'$ such that
  $\bar\rho, \rho \RelAncequiv{l'}{n_2'}{z}{n_1'} \bar\rho',
  \rho'$. 
\end{proof}

We have seen that  that Duplicator can answer every global 
challenge of Spoiler. But we still need to prove that she can choose
an element whose relevant ancestors all end in small stacks. 
The use of the pumping construction from Lemma 
\ref{Lem:RunsendInShorStacks} allows to prove this fact. 

\begin{lemma} \label{LemmaGlobalStep2}
  Let $n,l',n_1',n_2',z\in\N$, let
  $l:=4l'+5$, $n_1:=n_1'+2(l'+1)+1$, and let
  $n_2:=n_2'+4^{l'+1}+1$.  
  Furthermore, let $\bar\rho$ be an $n$-tuple of runs 
  such that  
  \mbox{$\height(\pi) \leq \BoundHeight(n,z,l,n_1,n_2)$} and
  $\lvert \pi\rvert \leq \BoundWidth(n,z,l,n_1,n_2)$ for all
  $\pi\in\RelAnc{l}{\bar\rho}$. 
  If $\rho$ is a run such
  that $\RelAnc{l'+1}{\bar\rho} \cap \RelAnc{l'+1}{\rho}=\emptyset$, 
  then there is
  some run $\rho'$ such that
  \begin{align*}
    &\bar\rho,\rho \RelAncequiv{l'}{n_2}{z}{n_1} \bar\rho, \rho',\\
    &\height(\pi)\leq \BoundHeight(n+1, z, l', n_1', n_2'), \text{ and}\\
    &\lvert \pi \rvert\leq \BoundWidth(n+1, z, l', n_1', n_2')    
  \end{align*}
  for 
  all $\pi\in \RelAnc{l'}{\bar\rho,\rho'}$.
\end{lemma}
\begin{proof}
  Let $\rho_0 \prec \rho_1 \prec \dots \prec \rho_m := \rho$ be runs such that
  $\RelAnc{l'+1}{\rho}=\{\rho_i:0\leq i \leq m\}$. 

  Let $m_0 > -n_1'$ be minimal such that there are runs
  $\rho_{m_0}  \prec  \rho_{{m_0}+1} \prec
    \dots \prec 
  \rho_0$ such that  for each $m_0\leq i < 0$ we have
  $\rho_i \plusedge \rho_{i+1}$ or $\rho_i \trans{\delta}
  \rho_{i+1}$ for some $\Clone{2}$ transition $\delta$. Due to the
  construction,  either $m_0=1-n_1'$ or $\lvert \rho_0\rvert< n_1'$.

  We construct runs $\rho_{m_0}' \prec \rho_{m_0+1}' \prec
  \dots \prec \rho_m'$ ending in small stacks such that 
  \mbox{$\rho_m'
    \RelAncequiv{l'}{n_2'}{z}{n_1'} \rho_m$} as follows.
  If $\height(\rho_{m_0}) \leq \BoundHeight(n, z, l, n_1, n_2)$ and 
  $\lvert \rho_{m_0}\rvert \leq \BoundWidth(n, z, l, n_1, n_2)$, then we
  set $\rho_{m_0}':=\rho_{m_0}$. 
  
  Otherwise,  Lemma 
  \ref{Lem:RunsendInShorStacks} provides a run $\rho'_{m_0}$ such
  that 
  \begin{align*}
    &\height(\rho'_{m_0})\leq H_1^{\text{glob}}:=\BoundHeight(n, z, l, n_1, n_2) +
    \ConstBoundHeightWord + \FuncBoundTopWord(n_2'+n_1'+4^{l'+1}-1)
    \\
    &\lvert \rho'_{m_0}\rvert \leq
    \BoundWidth(n, z, l, n_1, n_2) + \FuncBoundWidthWord( H_1^{\text{glob}})
    ,\\
    &\text{either } \height(\rho'_{m_0})> \BoundHeight(n,z,l,n_1,n_2)
    \text{ or }
    \lvert \rho'_{m_0} \rvert > \BoundWidth(n,z,l,n_1,n_2), \text{ and}\\
    &\rho_{m_0} \stackequiv{n_2' + n_1' + 4^{l'+1}-1}{z}{1} \rho'_{m_0}.
  \end{align*}
  The last condition just says that $\TOP{2}(\rho_{m_0})
  \wordequiv{n_2'+n_1'+4^{l'}-1}{z} \TOP{2}(\rho'_{m_0})$. 

  Having constructed $\rho'_{i}$ for $i< m$, we construct
  $\rho'_{i+1}$ as follows. 
  \begin{enumerate}
  \item If $\height(\rho_j) \leq \BoundHeight(n, z, l, n_1, n_2)$ and
    $\lvert \rho_j \rvert \leq \BoundWidth(n, z, l, n_1, n_2)$ for all 
    $m_0 \leq j \leq i+1$, set \mbox{$\rho_{i+1}':=\rho_{i+1}$.}
  \item Otherwise, if $\rho_i \trans{\delta} \rho_{i+1}$ then define
    $\rho'_{i+1}$ such that $\rho'_i \trans{\delta} \rho'_{i+1}$.
  \item If none of the previous cases applies, 
    then $\rho_i \plusedge \rho_{i+1}$ and using 
    Proposition \ref{Prop:ConstructOneStep} we construct 
    $\rho'_{i+1}$ such that 
    \begin{enumerate}
    \item $\rho'_i \plusedge \rho'_{i+1}$,
    \item $\height(\rho'_{i+1}) \leq 
      \BoundHeightOnestepConstructionSimultanious(1,
      \max\{\height(\rho'_i), \BoundHeight(n,z,l,n_1,n_2)\},
      n_2'+n_1'-(i+1-m_0))$,
    \item $\height(\rho'_{i+1}) > \BoundHeight(n,z,l,n_1,n_2)$ if
      $\height(\rho_j)\leq \BoundHeight(n, z, l, n_1, n_2)$ for all
      $m_0\leq j\leq i$ and if
      $\lvert \rho'_j \rvert \leq \BoundWidth(n, z, l, n_1, n_2)$ for
      all $m_0 \leq j \leq i+1$, and
    \item $\TOP{2}(\rho_{i+1})
      \wordequiv{n_2'+n_1'+4^{l'+1}-(i+1-m_0)-1}{z} \TOP{2}(\rho'_{i+1})$.
    \end{enumerate}
  \end{enumerate}
  First of all, note that 
  $n_2' + n_1' +4^{l'+1} -(0-m_0) - 1 \geq n_2' +  4^{l'+1}$ whence
  $\rho_0 \stackequiv{n_2'+4^{l'+1}}{z}{n_1'} \rho_0'$. 
  Since $m\leq 4^{l'+1}$ and since $\lvert \rho'_j \rvert \geq \lvert
  \rho'_0 \rvert$ for all $0 \leq j \leq m$, we conclude that 
  $\rho_{j} \stackequiv{n_2'}{z}{n_1'} \rho_j'$ for all $0\leq j \leq
  m$. 
  Using Lemma \ref{LemmaConstructedRelAncEquiv}, we see that
  $\rho' \RelAncequiv{l'+1}{n_2'}{z}{n_1'} \rho$. 
  
  Furthermore, $\RelAnc{l'+1}{\rho'} \cap
  \RelAnc{l'+1}{\bar\rho}=\emptyset$:
  heading for a contradiction, assume that 
  \begin{align*}
    \rho'_i\in
    \RelAnc{l'+1}{\rho'}\cap\RelAnc{l'+1}{\bar\rho}    
  \end{align*}
 for some $0\leq i \leq m$. 
  Then $\rho'_j\in\RelAnc{3l'+3}{\bar\rho} \subseteq
  \RelAnc{l}{\bar\rho}$ for all 
  $0\leq j \leq i$. Thus, 
  \mbox{$\height(\rho'_j)\leq \BoundHeight(n, z, l, n_1, n_2)$} and 
  $\lvert\rho'_j\rvert \leq  \BoundWidth(n, z, l, n_1, n_2)$ for all
  $0\leq j \leq i$.  
  By construction, it follows that $\rho'_j=\rho_j$ 
  for all $m_0 \leq i \leq j$. But then
  $\rho_j=\rho'_j\in\RelAnc{l'+1}{\rho}\cap\RelAnc{l'+1}{\bar\rho}$
  which contradicts 
  $\RelAnc{l'+1}{\rho}\cap\RelAnc{l'+1}{\bar\rho} = \emptyset$.

  Thus, $\RelAnc{l'}{\bar\rho}$ and $\RelAnc{l'}{\rho'}$ do not touch whence 
  $\bar\rho,\rho \RelAncequiv{l'}{n_2}{z}{n_1} \bar\rho, \rho'$. 
\end{proof}

\subsection{Construction of Short Equivalent Runs}
\label{SubsectionRunBounds}

Combining the results of the previous section, we obtain that for each
$n$-tuple in $\HONPT(\mathcal{S})$ there is an
$\FO{k}$-equivalent one such that the relevant 
ancestors of the second tuple only contain runs that end in small
stacks. 
In order to prove
Proposition \ref{Prop:2NPT-Strategy}, 
we still have to bound the length of these runs. 
For this purpose, we use
Corollaries \ref{Cor:GlobalBoundRun} and \ref{Cor:GlobalBoundRun2}
in order to replace long runs between relevant
ancestors by shorter ones. 


\begin{proof}[of Proposition \ref{Prop:2NPT-Strategy}]
  Using the Lemmas  \ref{LemmaLocalStep} -- \ref{LemmaGlobalStep2}, we
  find some candidate $\hat\rho'$ such that
  \mbox{$\bar\rho,\rho \RelAncequiv{l'}{n_2'}{z}{n_1'} 
    \bar\rho', \hat \rho'$}
  and the height and width 
  of the last stacks of all $\pi\in \RelAnc{l'+1}{\hat \rho'}$ are bounded by
  \mbox{$\BoundHeight(n+1,z,l',n_1',n_2')$} and 
  $\BoundWidth(n+1,z,l',n_1',n_2')$, respectively. 

  Recall that there is a chain
  $\hat \rho_0'\prec 
  \hat \rho_1' \prec \dots \prec \hat \rho_m'=\hat \rho'$ such that
  $\RelAnc{l'+1}{\hat \rho'}  
  = \{\hat \rho_i': 0 \leq i \leq m\}$. This chain satisfies 
  $0\leq m\leq 4^{(l'+1)}$ 
  and
  $\hat \rho_i'\trans{\delta} \hat \rho_{i+1}'$ or 
  $\hat \rho_i' \plusedge \hat \rho_{i+1}'$ for all
  $0\leq i < m$.

  If $\hat \rho_0'\notin\RelAnc{3l'+3}{\bar\rho'}$, then
  we can use Corollary \ref{Cor:GlobalBoundRun} and choose some $\rho_0'$  
  that ends in the same configuration as $\hat \rho_0'$ such
  that $\rho_0'\notin \RelAnc{3l'+3}{\bar\rho'}$ and 
  \begin{align*}
    \length(\rho_0')\leq 1+ &2\cdot \BoundHeight(n+1,z,l',n_1',n_2') \cdot
    \BoundWidth(n+1, z, l', n_1', n_2'))\\ &\cdot 
    (1+ \FuncBoundLoopLength{\mathcal{N}}{z}(\BoundHeight(n+1, z, l',
    n_1', n_2'))).     
  \end{align*}

  If $\hat \rho_0'\in\RelAnc{3l'+3}{\bar\rho'}$ 
  let $0 \leq i \leq m$ be maximal such that 
  $\hat \rho_i'\in\RelAnc{l}{\bar\rho'}$.
  In this case let $\rho_j' :=\hat \rho_j'$ for all $0\leq j \leq i$. 
  
  By now, we have obtained a chain $\rho_0' \prec \rho_1' \prec \dots \prec
  \rho_i'$ for some $0\leq i \leq m$.   
  Using Corollary \ref{Cor:GlobalBoundRun2}, we can extend this chain to a chain
  $\{\rho_i': 0\leq i \leq m\}$ such that 
  \begin{enumerate}
  \item $\rho'_i(\length(\rho'_i)) = \hat\rho_i'(\length(\hat\rho_i'))$,
    i.e., $\rho_i'$ and $\hat\rho_i'$ end in the  same configuration, 
  \item   $\rho_i'\trans{\delta} \rho_{i+1}'$  iff  
    $\hat \rho_i'\trans{\delta} \hat \rho_{i+1}'$   for all $0\leq i < m$,
  \item $\rho_i' \plusedge \rho_{i+1}'$  iff 
    $\hat \rho_i' \plusedge \hat \rho_{i+1}'$ for all
    $0\leq i < m$,
  \item   $\length(\rho_{i+1}') \leq \length(\rho_i') + 
    2 \cdot\BoundHeight(n+1,z,l',n_1',n_2') \cdot 
    (1+ \FuncBoundLoopLength{\mathcal{N}}{z}(
    \BoundHeight(n+1, z, l', n_1',
    n_2')))$, and
  \item  $\hat \rho_j'\in\RelAnc{3l'+3}{\bar\rho'}$  for
    all $0\leq j \leq i$ implies  $\rho_i'=\hat \rho_i'$ (here we use
    that $z> n \cdot 4^{3l'+3}$).
  \end{enumerate}

  Using Lemma \ref{LemmaConstructedRelAncEquiv}, we conclude that $\rho'
  \RelAncequiv{l'+1}{n_2'}{z}{n_1'} 
  \hat \rho'$ for $\rho':=\rho_m'$. Furthermore, we claim that 
  $\RelAnc{l'+1}{\bar\rho'} \cap \RelAnc{l'+1}{\hat \rho'} =   
  \RelAnc{l'+1}{\bar\rho'} \cap \RelAnc{l'+1}{\rho'}$. 
  By definition the inclusion from left to right is clear. For the
  other direction, assume that
  there is some element 
  \mbox{$\rho_i'\in \RelAnc{l'+1}{\bar\rho'} \cap
  \RelAnc{l'+1}{\rho'}$.}
  By Lemma \ref{CorRelAncDistBound}, this implies that 
  $\rho_j'\in \RelAnc{3l'+3}{\bar\rho'}$ 
  for all $0\leq j \leq i$. Thus, $\rho_i'=\hat \rho_i'$, which
  implies that 
  $\rho_i'\in \RelAnc{l'+1}{\bar\rho'} \cap \RelAnc{l'+1}{\hat\rho'}$. 
  
  We conclude that $\bar\rho, \rho 
  \RelAncequiv{l'}{n_2'}{z}{n_1'} \bar\rho', \hat \rho'
  \RelAncequiv{l'}{n_2'}{z}{n_1'} \bar\rho', \rho'$.
  By definition, the length of $\rho'$ is bounded by 
  $\BoundRunLength(n+1, z, l', n_1', n_2')$ (cf. Definition 
  \ref{Def:BoundingFunctions}).  
\end{proof}

\section{FO Model Checking Algorithm for Level 2  
  Nested Pushdown  Trees} 
\label{sec:FODecidability}

Fix a $2$-PS $\mathcal{N}$ and set $\mathfrak{N}:=\HONPT(\mathcal{N})$.
We have shown that  if $\mathfrak{N} \models \exists x \varphi(x)$
then  there is a small run $\rho$ such that 
$\mathfrak{N}\models \varphi(\rho)$. Even when we add parameters
$\rho_1, \dots, \rho_n$
this result still holds, i.e., there is a short witness $\rho$ 
compared to the length of the parameters. Hence, we can decide
$\FO{}$ on $2$-NPT with the
following algorithm.

\begin{enumerate}
\item Given a $2$-PS $\mathcal{N}$ and a
  first-order sentence $\varphi$, the algorithm first computes the
  quantifier rank $q$ of $\varphi$.
\item Then it computes numbers $z,l^1,l^2,l^3,\dots, l^q, n_1^1, n_1^2,
  n_1^3 \dots, n_1^q, n_2^1, n_2^2, n_2^3, \dots, n_2^q\in\N$ such
  that for each $i< q$ the numbers $z,l^i,l^{i+1},n_1^i, n_1^{i+1},
  n_2^i, n_2^{i+1}$ can be used as parameters in 
  Proposition \ref{Prop:2NPT-Strategy}.
%
\item These numbers define a constraint $S=(S^{\mathfrak{N}}(i))_{i\leq
    q}$ for Duplicator's strategy in the $q$-round game on
  $\mathfrak{N}$ and $\mathfrak{N}$ as follows. 
  We set $(\rho_1, \rho_2, \dots, \rho_m)\in S^{\mathfrak{N}}_m$
  if for each $i\leq m$ and $\pi\in\RelAnc{l_i}{\rho_i}$
  \begin{align*}
    &\length(\pi)\leq \BoundRunLength(i,z,l^i,n_1^i,n_2^i), \\
    &\height(\pi)\leq \BoundHeight(i,z,l^i,n_1^i,n_2^i), \text{ and}\\
    &\lvert\pi\rvert \leq \BoundWidth(i,z,l^i,n_1^i,n_2^i). 
  \end{align*}
\item Due to Proposition \ref{Prop:2NPT-Strategy}, Duplicator has an
  $S$-preserving strategy in the $q$-round game on $\mathfrak{N}$
  and $\mathfrak{N}$. Thus, applying the algorithm SModelCheck
  (cf. Algorithm \ref{AlgoSPReservingModelCheck} in Section
  \ref{sec:EFGame}) decides whether $\mathfrak{N}\models \varphi$. 
\end{enumerate}

\subsection{Complexity of the Algorithm}
At the moment, we do not know any bound on the complexity of the
algorithm presented. 
The problem towards giving a complexity bound on our algorithm 
is that the function 
$\FuncBoundLoopLength{\mathcal{N}}{z}$ bounding the size of the
shortest $z$ loops of each stack depends nonuniformly on the $2$-PS
$\mathcal{N}$. For a fixed $2$-PS $\mathcal{N}$ we can compute this
dependence, but we have no general bound on the result in terms of
$\lvert \mathcal{N} \rvert$. 
A possible approach to concrete bounds on
$\FuncBoundLoopLength{\mathcal{N}}{z}$ may be the application of 
Pumping lemmas for $2$-PS (cf.~\cite{%
  Hayashi73,Parys2012HOPGpumping}). But further
investigations on this question are necessary. 

If we restrict our attention to the case of  $1$-NPT the picture
changes notably. The $\FO{}$ model checking problem on $1$-NPT can be
solved by an 
$2$-EXPTIME alternating Turing machine, i.e., $\FO{}$ model checking
on $1$-NPT is in ATIME($\exp_2$). We already proved this bound
in \cite{Kartzow09} using 
a different approach. In the final part of this section, we sketch how 
this result also follows from the approach in this paper.
Recall the characterisation of
$\plusedge$ 
in the $1$-NPT case from Remark 
\ref{rem:CharacterisePlusOneRelation}. We have 
$\rho\plusedge \rho\circ\pi$ if $\pi$ performs a
push transition followed by level $1$-loop (i.e., a run that starts
and ends in the same stack and never inspects this stack). 
Furthermore, for $\rho_n$ the minimal element of $\RelAnc{n}{\rho}$ we
have $\rho_n=\Pop{1}^n(\rho)$. Furthermore, successive elements of
$\RelAnc{n}{\rho}$ are connected by a single edge or by 
$\plusedge$. Thus, the final stacks of two
successive relevant ancestors differ in at most one letter. 
The following lemma tells us that the number of
$\plusedge$ edges starting at some element $\rho$
only depends on the state of $\rho$ and the symbol on top of the
stack. 
\begin{lemma}
  Let $\mathcal{N}$ be a $1$-PS. Let $q,\hat q\in Q$,
  $w,w'\in\Sigma^*$ and $a\in \Sigma$. 
  Then there is a bijection between the runs from $(q,wa)$ to $(\hat
  q, wa)$ that never visit $w$ and the runs from 
  $(q,w'a)$ to $(\hat q, w'a)$ that never visit $w'$. 
\end{lemma}
\begin{proof}
  The bijection is given by the stack replacement $[w/w']$
  (cf. Lemma \ref{Lem:BlumensathHOLevel2}). 
\end{proof}

Using the previous observation, it is straightforward to prove the
following lemma. 
\begin{lemma} \label{lem:LoopsContextFree}
  Let $\mathcal{N}$ be a $1$-PS and let $q,\hat q\in Q$,
    $w\in\Sigma^*$ and $a\in \Sigma$. 
    The set 
    \begin{align*}
      \{\rho: \rho(0)=(q,wa), \rho(\length(\rho))=(q',wa),
      \text{ and } \rho(i)\neq w\text{ for all }0\leq i \leq
      \length(\rho)\}       
    \end{align*}
    is a context-free language accepted by some $1$-PS of 
    size linear in $\lvert \mathcal{N} \rvert$. 
    Furthermore, the set of runs from the
    initial configuration to $(q,w)$ of $\mathcal{N}$ forms a
    context-free language that is accepted by some $1$-PS of size
    linear in $\lvert \mathcal{N} \rvert$. 
\end{lemma}

Using the pumping lemma for context
free-languages \cite{Bar-HillelPS61}, we derive the following bound on
short elements of context-free languages. 
\begin{lemma} \label{lem:CFPumping}
  There is a fixed polynomial $p$ such that the following holds.
  Let $L$ be some context-free language that is accepted by a $1$-PS
  $\mathcal{N}$. If $L$ contains $k$ elements, then there are pairwise
  distinct words $w_1, w_2, \ldots, w_k\in L$ such that  
  length $\lvert w_i \rvert$ is bounded by  
  $k \cdot \exp(p(\lvert N \rvert))$ for all $1\leq i \leq k$. 
\end{lemma}

These observations imply that it is rather easy to construct runs with
similar relevant ancestors in a $1$-NPT. 
In the following, we define a simpler notion of equivalent
relevant ancestors in the $1$-NPT case that replaces the one for 
$2$-NPT in Definition \ref{def:RAequivalence}. 

\begin{definition}
  Let $\mathcal{N}$ be some $1$-PS generating
  $\mathfrak{N}:=\HONPT(\mathcal{N})$.  
  Let $\bar\rho:=\rho_1, \rho_2, \dots, \rho_n$ and 
  \mbox{$\bar\rho':=\rho_1', \rho_2', \dots, \rho_n'$} be elements of
  $\mathfrak{N}$. We define  $\bar\rho\approx_l \bar\rho'$ if there is a
  bijection $\varphi: \RelAnc{l}{\bar\rho} \to \RelAnc{l}{\bar\rho'}$ such
  that the following holds:
  \begin{enumerate}
  \item For all $l'\leq l$ and $\pi\in\RelAnc{l'}{\rho_i}$, we have 
    $\varphi(\pi)\in \RelAnc{l'}{\rho_i'}$.
  \item $\varphi$ preserves $\trans{\delta}$-, $\jumpedge$-
    and $\plusedge$-edges.
  \item $\varphi$ preserves states and topmost stack elements, i.e., 
    for $\pi\in\RelAnc{l}{\bar\rho}$ such that $\pi$ ends in
    configuration $(q, wa)$, then $\varphi(\pi)$ ends in $(q, w'a)$
    for some word $w'$. 
  \end{enumerate}
\end{definition}

Using corresponding constructions as in the proof of proposition 
\ref{Prop:2NPT-Strategy}, we can prove that Duplicator has a winning
strategy that only uses elements of doubly exponential size.

\begin{proposition} \label{Prop:1NPT-Strategy}
  There is a polynomial $p$ such that the following holds:
  Let $\mathcal{N}$ be a $1$-PS generating the 
  NPT $\mathfrak{N}:=\HONPT(\mathcal{N})$.  
  Let $l', C\in\N$, and  
  $l := 4l'+5$.
  Furthermore, let $\bar\rho$ and $\bar\rho'$ be $n$-tuples of runs of
  $\mathfrak{N}$ such that  
  $\bar\rho  \approx_l \bar\rho'$, and
  $\length(\pi)\leq  C$
  for all $\pi\in\RelAnc{l}{\bar\rho'}$,
  
  For each $\rho\in \mathfrak{N}$ there is some $\rho'\in
  \mathfrak{N}$ such that 
  \begin{align*}
    \bar\rho, \rho \approx_{l'} \bar\rho', \rho'\text{ and }
    \length(\pi)\leq  C+ 4^{l'} \cdot n\cdot 4^l\exp(p(\lvert N \rvert)) 
    \text{ for all }\pi\in\RelAnc{l'}{\bar\rho',\rho'}.
  \end{align*}
\end{proposition}
\begin{proof}
  As in the $2$-NPT case, we distinguish a local and a global case.
  Let $\varphi$ witness that $\bar\rho \approx_l \bar\rho'$. 
  \begin{itemize}
  \item If $\RelAnc{l'}{\rho} \cap \RelAnc{l}{\bar\rho}\neq\emptyset$
    let $\rho_0$ 
    be maximal in this set. We define $\rho_0':=\varphi(\rho_0)$. 
    There are runs $\rho_0 \prec \rho_1 \prec \rho_2 \prec\dots\prec
    \rho_m$  which form the set $\RelAnc{l'}{\rho} \cap\{\pi:
    \rho_0\preceq \pi\}$. We now inductively construct $\rho_0' \prec
    \rho_1' \prec \rho_2'\prec\dots\prec \rho_m'$ such that
    $\rho':=\rho_m'$ satisfies the claim. 
    Therefore, we copy the edge connecting $\rho_i$ with
    $\rho_{i+1}$. If
    $\rho_i\plusedge\rho_{i+1}$ we can construct a
    run $\rho'_{i+1}\notin\RelAnc{l}{\bar\rho}$ such that $\rho'_i
    \plusedge 
    \rho_{i+1} $. Using Lemmas 
    \ref{lem:LoopsContextFree} and \ref{lem:CFPumping} 
    we can ensure that $\rho'_{i+1}=\rho'_i\circ\pi$ with
    $\length(\pi)$ bounded by  
    $n\cdot 4^l \cdot \exp( p(\lvert \mathcal{N}\rvert))$. 
    Since $m\leq 4^{l'}$ the claim follows by iterative use of this
    observation.
  \item If $\RelAnc{l'}{\rho}\cap \RelAnc{l}{\bar\rho}=\emptyset$, we
    proceed analogously to Lemma \ref{LemmaGlobalStep1}. 
    Let $\rho_0$ be the minimal element of $\RelAnc{l'}{\rho}$. 
    Let $q$ be its final state and $a$ be its final topmost stack
    entry. Let $m$ be the number of occurrences of elements in
    $\RelAnc{l}{\bar\rho}$ that end in state $q$ and with final
    topmost symbol $a$. 
    Due to $\bar\rho \approx_l \bar\rho'$, $m$ is also the number of
    elements in 
    $\RelAnc{l}{\bar\rho'}$ that end in state $q$ and with final
    topmost symbol $a$. 
    Due to Lemmas \ref{lem:LoopsContextFree} and \ref{lem:CFPumping} 
    there are $m+1$ elements of size at most
    $n \cdot 4^l \cdot \exp(p(\lvert \mathcal{N}\rvert ))$
    that end in state $q$ and end with topmost stack entry $a$. 
    By pigeonhole principle, we can set $\rho'_0$ to be one of these
    such that $\rho'_0\notin\RelAnc{l}{\bar\rho'}$. 
    Now, we proceed as in the local case. 
    Let $\rho_0\prec \rho_1 \prec \rho_2 \prec \dots \prec \rho_k$ be
    the enumeration of $\RelAnc{l'}{\rho}$. 
    We define runs $\rho'_0
    \prec \rho'_1 \prec \rho'_2 \prec\dots\prec \rho'_k$ such that 
    $\rho_i$ is connected to $\rho_{i+1}$ via the same edge as
    $\rho_i'$ to $\rho_{i+1}'$. Due to Lemmas  
    \ref{lem:LoopsContextFree} and \ref{lem:CFPumping}, $\rho_{i+1}$
    can be chosen such that the run
    connecting $\rho_i'$ with $\rho_{i+1}'$ is bounded by
    $n \cdot 4^l \cdot \exp(p(\lvert \mathcal{N}\rvert ))$. Since
    $k < 4^{l'}$ we conclude that 
    $\rho':=\rho_k'$ has length at most 
    $4^{l'} \cdot n \cdot 4^l \cdot \exp(p(\lvert \mathcal{N}\rvert
    ))$. \qed
  \end{itemize}
\end{proof}

\begin{corollary}
  $\FO{}$ model checking on $1$-NPT can be solved by an alternating
  Turing machine in $2$-EXPTIME with linearly many alternations in the
  size of the formula.
\end{corollary}
\begin{proof}
  Let $\varphi\in\FO{r}$ and $\mathcal{N}$ some $1$-PS. 
  Iterated use of the previous lemmas shows that Duplicator has a
  winning strategy in the $r$ round game on $\HONPT(\mathcal{N})$ and
  $\HONPT(\mathcal{N})$ where he chooses elements 
  of size bounded by $r\cdot \exp(\exp(q(r)))\cdot
  \exp(p(\lvert\mathcal{N}\rvert))$ for some fixed polynomials $p$ and
  $q$.
\end{proof}

\section{Conclusion}
\label{sec:Conclusion}

In this paper we extended the notion of a nested pushdown tree and
developed the hierarchy of higher-order nested pushdown trees. These 
are  higher-order pushdown trees enriched by a
jump relation that makes corresponding push and pop operations (of the
highest level) visible. This new hierarchy is an intermediate step
between the 
hierarchy of pushdown trees and the hierarchy of collapsible
pushdown graphs in the sense that it contains expansions of
higher-order pushdown trees and its $n$-th level is uniformly
first-order interpretable 
in the class of collapsible pushdown graphs of level $n+1$. We hope
that further 
study of this hierarchy helps to clarify the relationship between the
hierarchies defined by higher-order pushdown systems and by
collapsible pushdown systems. 
We have shown the decidability of the first-order model checking on
the first two levels of the nested pushdown tree hierarchy (this
contrasts the undecidability of first-order logic on the class of all
level $3$ collapsible pushdown graphs \cite{Broadbent2012}). 
The algorithm is obtained from the analysis of restricted strategies
in Ehrenfeucht-\Fraisse games on 
nested pushdown trees of level $2$. The game analysis gets tractable
due to the theory of \emph{relevant ancestors} in combination with
shrinking constructions for level $2$ pushdown systems. It is open
whether this approach extends to 
higher-levels. The theory of relevant ancestors generalises to all
levels of the hierarchy and provides an understanding of the
Ehrenfeucht-\Fraisse games. Unfortunately, we do not have any kind of
shrinking constructions for pushdown systems of level $3$
or higher. The development of such shrinking
construction may yield 
the necessary bridge between the theory of relevant ancestors of
$n$-NPT and the dynamic-small-witness property needed for the
development of a model checking algorithm. 
Furthermore, better shrinking construction may imply elementary
complexity bounds for the \FO{} model checking on $2$-NPT. 
Another open question concerns the modal $\mu$-calculus model checking
on this new hierarchy. Note that the interpretation of $n$-NPT in
collapsible pushdown 
graphs of level $n+1$ is almost modal but for the reversal of the
jump edges in comparison to the collapse edges used for their
simulation. We conjecture that modal $\mu$-calculus is decidable on
the whole hierarchy of nested pushdown trees. Moreover, it is open
how the nested pushdown tree hierarchy relates exactly to the
(collapsible) higher-order pushdown hierarchies. We conjecture that
the relationships obtained in this paper are optimal but we lack
proofs. For instance, it is open whether $n$-NPT can be
interpreted in level $n$ collapsible pushdown graphs or vice versa (we
conjecture the answer is no).

\bibliography{/u/kartzow/Paper/Kartzow/Standard}
\bibliographystyle{acmtrans}

\end{document}